\documentclass[acmsmall,screen,nonacm]{acmart}

\usepackage{booktabs}   %
\usepackage{subcaption} %

\usepackage{listings}
\usepackage{mathpartir}
\usepackage{xspace}
\usepackage{mathtools} %
\usepackage{tabularx}
\usepackage{cprotect}
\usepackage{enumitem} %
\usepackage{multirow}

\setlist[itemize]{noitemsep, topsep=1pt} %

\usepackage[shorthand]{semantic} %
\mathlig{=<}{\leq}
\mathlig{>=}{\geq}
\mathlig{!=}{\neq}
\mathlig{==}{\equiv}
\mathlig{|=}{\vDash} %
\mathlig{=|}{\Dashv} %
\mathlig{-|}{\dashv} %
\mathlig{||-}{\Vdash}
\mathlig{[|}{\left\llbracket}
\mathlig{|]}{\right\rrbracket}
\mathlig{|->}{\mapsto}
\mathlig{~>}{\rightsquigarrow}
\mathlig{<<}{\ll}
\mathlig{<<<}{\lll}
\mathlig{>>=}{\gg\!=}
\mathlig{|>}{\times}

\usepackage{zref-savepos}
\makeatletter
\@ifundefined{zsaveposx}{\let\zsaveposx\zsavepos}{}
\makeatother
\newcounter{hposcnt}
\renewcommand*{\thehposcnt}{hpos\number\value{hposcnt}}
\newcommand*{\SP}{%
  \stepcounter{hposcnt}%
  \zsaveposx{\thehposcnt s}%
}
\makeatletter
\newcommand*{\UP}{%
  \zsaveposx{\thehposcnt u}%
  \zref@refused{\thehposcnt s}%
  \zref@refused{\thehposcnt u}%
  \kern\zposx{\thehposcnt s}sp\relax
  \kern-\zposx{\thehposcnt u}sp\relax
}
\makeatother

\newcommand*{\TechReport}{}%

\ifdefined\TechReport
\newcommand{\tr}[2]{#1}
\newcommand{\thetr}{the technical report version of this paper\textcolor{orange}{TODO\cite{?}}}
\newcommand{\thet}{the technical report\textcolor{orange}{TODO\cite{?}}}
\else
\newcommand{\tr}[2]{#2}
\newcommand{\thetr}{the technical report version of this paper\xspace}
\newcommand{\thet}{the technical report\xspace}
\fi

\newcommand{\System}{Lumberhack\xspace}

\newcommand{\Sec}[1]{\S\ref{#1}}
\newcommand{\Fig}[1]{Figure~\ref{#1}}

\newcommand{\App}[1]{Appendix~\ref{#1}}
\newcommand{\Def}[1]{Definition~\ref{#1}}

\newcommand{\Thm}[1]{Theorem~\ref{#1}}
\newcommand{\Lem}[1]{Lemma~\ref{#1}}

\newcommand{\aka}{a.k.a.\ }
\newcommand{\ie}{i.e.,\ }

\ifdefined \Draft
\newcommand{\TODO}[1]{\noindent\textcolor{orange}{{TODO: #1}}}
\newcommand{\NOTE}[1]{\noindent\textcolor{teal}{{NOTE: #1}}}
\newcommand{\TODOac}[1]{\noindent\textcolor{orange}{{TODO[Anto]: #1}}}
\newcommand{\TODOlp}[1]{\noindent\textcolor{orange}{{TODO[Lionel]: #1}}}

\newcommand{\FIXME}[1]{\noindent\textcolor{red}{{FIXME: #1}}}

\newcommand{\TODOlater}[1]{\textcolor{gray}{[...]}\xspace}
\else
\newcommand{\TODO}[1]{}
\newcommand{\NOTE}[1]{}
\newcommand{\TODOac}[1]{}
\newcommand{\TODOlp}[1]{}

\newcommand{\TODOlater}[1]{}
\newcommand{\FIXME}[1]{}

\fi

\newcommand{\TODOlaterlater}[1]{}

\newcommand{\cutForLater}[1]{}
\newcommand{\cutRedundant}[1]{}
\newcommand{\cut}[1]{}

\newcommand\al{\alpha}
\newcommand\be{\beta}
\newcommand\ga{\gamma}
\newcommand\de{\delta}
\newcommand\ep{\varepsilon}

\newcommand\lam{\lambda}

\newcommand{\kw}[1]{\textsf{\textbf{#1}}}
\newcommand{\kws}[1]{\textsf{\textbf{#1}}\ }

\newcommand{\skws}[1]{\ \textsf{\textbf{#1}}\ }

\newcommand{\id}[1]{\mathsf{#1}}

\newcommand{\fname}[1]{\mathit{#1}}

\newcommand{\tv}[1]{\fname{tv}(#1)}
\newcommand{\ntv}[1]{\fname{ntv}(#1)}

\newcommand\cat{{\cdot}}
\newcommand\catsp{{\,\cdot\,}}

\newcommand{\ty}{\tau}
\newcommand{\tz}{\sigma}

\newcommand{\Pos}{^{+}}
\newcommand{\Neg}{^{-}}

\newcommand{\conR}{\gg}

\newcommand{\SCtx}{\Delta}
\newcommand{\SDtx}{\nabla}

\newcommand{\CCtx}{\Xi}
\newcommand{\CDtx}{\Sigma}

\newcommand{\fresh}[1]{{#1}\,\ \textit{\textbf{fresh}}}

\newcommand{\ruleName}[1]{\textsc{#1}}

\newcommand{\mypara}[1]{\smallskip\textbf{\textit{#1.}}}

\lstdefinelanguage{Haskell}{%
  aboveskip=2pt,
  belowskip=3pt,
  basicstyle=\footnotesize\ttfamily,
  flexiblecolumns=false,
  lineskip=0pt,
  basewidth={0.5em,0.45em},
  morekeywords={let,in,case,of,build,foldr,if,then,else,class,instance,where,type,fun,rec,augment,unfoldr,destroy},
  moredelim=**[is][\color{\codeEmphasisColor}]{⌜}{⌝}, %
  morecomment=[f][\color{gray}][0]{--},
  showstringspaces=true,
  morestring=[b]',%
  literate={+}{{$+$}}1 {/}{{$/$}}1 {*}{{$*$}}1
      {€}{\$}1
      {=}{{$=$}}1 %
      {>}{{$>$}}1 {<}{{$<$}}1 {\\}{{$\lambda$}}1
      {->}{{$\rightarrow$}}2 
      {<-}{{$\leftarrow$}}2
      {<=}{{$\leq$}}2 {=>}{{$\Rightarrow$}}2
      {\ .\ }{{$\circ$}}2 {(.)}{{($\circ$)}}2
      {<<}{{$\ll$}}2
      {>>=}{{$\gg\!=$}}3
      {<<<}{{$\lll$}}2 {>>>}{{$\ggg$}}2 {-<}{{$\leftY$}}1 {^<<}{{$\hat{}\!\!\ll$}}2 {^>>}{{$\hat{}\!\!\gg$}}2
      {|}{{$\mid$}}1
      {undefined}{{$\bot$}}1
      {forall}{{$\forall$}}1
}[keywords,comments,strings]%

\begin{document}

\setlength{\abovedisplayskip}{2pt}
\setlength{\belowdisplayskip}{2pt}
\setlength{\abovedisplayshortskip}{2pt}
\setlength{\belowdisplayshortskip}{2pt}

\AtBeginEnvironment{theorem}{\makeatletter}
\AtBeginEnvironment{definition}{\makeatletter}

\lstMakeShortInline[
  columns=flexible,
  breaklines=true,
  basewidth=0.5em,
]@

\newcommand{\titleText}[1]{%
  The Long Way to Deforestation
  (Technical Report)%
}
\subtitle{%
  A Type Inference and Elaboration Technique for
  Removing Intermediate Data Structures%
}

\title[\titleText{}]{\titleText{\\}}
\titlenote{This is the full version of an
  ICFP 2024 paper \citet{Chen-icfp22:deforestation}.}

\author{Yijia Chen}
\orcid{0009-0006-1273-8055}
\affiliation{%
  \institution{HKUST, Hong Kong}
  \country{China}
}

\author{Lionel Parreaux}
\orcid{0000-0002-8805-0728}
\affiliation{%
  \institution{HKUST, Hong Kong}
  \country{China}
}

\begin{abstract}

Deforestation is %
a compiler optimization that removes
intermediate data structure allocations
from functional programs to improve their efficiency.
This is an old idea,
but previous approaches have %
proved limited or impractical:
they either only worked on compositions of predefined combinators (shortcut fusion)
or involved %
the aggressive unfolding of recursive definitions until a 
depth limit was reached
or a reoccurring pattern was found to tie the recursive knot,
resulting in
impractical algorithmic complexity and large amounts of code duplication.
We present \System,
a general-purpose deforestation approach
for purely functional 
call-by-need \emph{and}
call-by-value programs.
\System uses subtype inference to reason about data structure production and consumption
and uses an elaboration pass to fuse the corresponding recursive definitions.
It fuses large classes of
mutually recursive definitions 
while
avoiding much of the unproductive (and sometimes {counter-productive})
code duplication %
inherent in previous approaches.
We prove the soundness of \System using step-indexed logical relations
and experimentally demonstrate significant speedups
in the standard \emph{nofib} benchmark suite.
We manually adapted 
\emph{nofib}
programs to call-by-value semantics
and compiled them using the OCaml compiler.
The average speedup over the 38 benchmarked programs is 8.2\%
while the average code size increases by just about 1.79x.
In particular, 19 programs see their performance mostly unchanged,
17 programs improve significantly (by an average speedup of 16.6\%),
and only two programs visibly worsen (by an average slowdown of 1.8\%).
As a point of comparison, we measured that the well-proven but \emph{semi-manual}
list fusion technique of the Glasgow Haskell Compiler (GHC),
which only works for call-by-need programs,
had an average speedup of 6.5\%.
Our technique is still in its infancy and misses some deforestation opportunities.
We %
are confident that further refinements %
will yield greater performance
improvements in the future.

\end{abstract}

\ccsdesc[500]{Software and its engineering~Software performance}
\ccsdesc[500]{Software and its engineering~Functional languages}
\ccsdesc[300]{Software and its engineering~Compilers}
\ccsdesc[300]{Theory of computation~Type theory}

\keywords{deforestation, fusion, type inference, elaboration, optimization}  %

\maketitle

\section{Introduction}
\label{sec:intro}

In functional programming, \emph{deforestation} refers to
program transformation techniques that reduce the number of intermediate data structures
created during program execution.

Consider the %
example in \Fig{fig:map-of-map}, 
a poster child of deforestation.
\begin{figure}
\begin{lstlisting}
              let rec map f xs = case xs of { []       -> [];
                                              x :: xs  -> f x :: map f xs }
\end{lstlisting}
\vspace{-.6em}
\begin{lstlisting}
              let incr x = x + 1
              let double x = x * 2
              let main ls = map incr (map double ls)
\end{lstlisting}
\caption{The typical ``map-of-map'' motivating example for deforestation.}
\label{fig:map-of-map}
\end{figure}
It should be obvious to any functional programmer that the intermediate list created by @map double ls@
is immediately consumed by @map incr@ and could thus be shunned by \emph{fusing} the two
list traversals into a single one.
A natural way to do this is to rewrite the nested call to
@map@ @(fun@ @x@ @->@ @incr@ @(double@ @x))@ @ls@.
This only works because @map@'s behavior is parameterized
and because the two traversals have the same structure.
However, it would fail for other compositions of recursive functions,
such as @map@ and @sum@, for example.
A more general approach is to fold both @map@ calls
into a call to a single definition --- call it @map2@ --- that performs both transformations at once:

\begin{lstlisting}
  let rec map2 f g xs = case xs of { []      -> [];
                                     x :: xs -> f (g x) :: map2 f g xs }
  let main ls = map2 incr double ls
\end{lstlisting}
The intermediate list is no longer created,
drastically reducing the amount of memory allocated during program execution.
This results in a significant performance improvement.\!\footnote{
  In this program, the speedup can typically be around 40\%,
  depending on the language and compiler used to run the code.
}

The intuition for this %
transformation is so clear, and its manual application so natural, that
one may be surprised to learn that it is \emph{not} performed by
\emph{any} of the existing major compilers,
and in particular by any of the mainstream functional programming language compilers.
Indeed, performing this style of general recursive function deforestation automatically
in a way that is \emph{practical for real compilers}
turns out to be surprisingly difficult.

Many approaches to deforestation have been proposed over the years
\cite{%
wadler-tcs90:deforestation,%
chin-deforestation,%
hamilton-higher-order,%
turchin-supercompiler,%
mitchell-superohaskell,%
hamilton-distillation,%
Jonsson09:pos-supercomp-ho-cbv,%
Bolingbroke-hs10:supercomp-by-eval,%
shortcut-deforest,%
Gill-96:cheap-deforestation,%
Chitil-ifp99,%
coutts-icfp07:stream-fusion,%
warmfusion,%
hu-hylo,%
Onoue-hylo-haskell,%
Ohori-popl07:light-fusion%
},
but they all suffer from one or more of the following 
limitations:
\begin{enumerate}[topsep=3pt]
  \item\label{item:code-dup} they result in large amounts of code duplication,
    leading to code size %
    explosion and impractical compilation times;
  \item\label{item:complex} they are complex,
    making their soundness hard to justify and making them hard to implement correctly;
  \item\label{item:cbv} they cannot handle call-by-value semantics,
    being typically limited to nonstrict languages;
  \item\label{item:restricted} they are simple and practical but
      apply only to very limited classes of programs,
      such as programs that only use predefined combinators on lists instead of explicit recursion.
\end{enumerate}

The original algorithm presented by \citet{wadler-tcs90:deforestation}
suffered from both limited applicability (\ref{item:restricted})
and extensive code duplication (\ref{item:code-dup}).
Related approaches following the more general \emph{supercompilation} technique
are typically %
quite complex, making their implementations brittle 
(\ref{item:complex})
--- in fact, all the available implementations we could try were broken in some ways
(see \Sec{related-implementation}) ---
and most cannot handle call-by-value semantics (\ref{item:cbv}).
Finally, \emph{shortcut deforestation} \cite{shortcut-deforest},
one of the few of these techniques to see industrial use,
most notably in the \emph{Glasgow Haskell Compiler} (GHC),
only applies to the two primitive list combinators @build@ and @foldr@,
from which all other list operations must be redefined
(\ref{item:restricted}).
This restriction was somewhat relaxed by \citet{peytonjones2001playing}
with the introduction of \emph{user-defined rewrite rules},
which can be added for arbitrary combinations of predefined functions.
However, the approach is still severely limited (\ref{item:restricted}) in that
it still cannot deal with arbitrary recursive functions.
Moreover, it is %
error-prone, as there are no guarantees that
these rules preserve program semantics,
and proving that they do is a subtle and difficult task.

In this paper, we present \textbf{\emph{\System}},
the first approach to deforestation that is all of:
\textbf{\emph{generic}},
in that it can be applied to both call-by-need and call-by-value languages;
\textbf{\emph{general}},
in that it handles large classes of programs with mutually recursive definitions
and not just predefined combinators;
\textbf{\emph{practical}}, meaning that it does not result in runaway code duplication;
and relatively \textbf{\emph{simple}} to implement and to prove \textbf{\emph{sound}}.
\System
automatically synthesizes an equivalent of the @map2@ function presented above
and can fuse many %
transformations of %
general functional data structures,
such as binary trees,
not being limited to lists.

\System is a Curry-style system:
it does not require the base language to have an existing type system
and could %
be applied to
pure programs in
dynamic languages
like Scheme and JavaScript.
Indeed, \System infers its own descriptive equi-recursive types
in a best-effort fashion, falling back to leaving the program unchanged when
satisfactory types cannot be inferred.

Our specific contributions are as follows:
\begin{itemize}[topsep=3pt]
  \item We present the problem of deforestation and how %
    it is addressed
    in \System
    (\Sec{sec:presentation}).
  \item We formalize the declarative core of \System
    as a type-based program transformation,
    commonly called \emph{elaboration},
    and give examples of its applications
    (\Sec{sec:formal}).
  \item We formally prove the soundness of 
    this core %
    system
    (\Sec{sec:correctness}\tr{ and \Sec{app:proofs}}{})
    by
    adapting a standard step-indexed logical relation technique
    from \citet{amal-step-index,appel-equi-rec-type}.
  \item We formalize the main algorithmic aspects of
    \System as a subtype inference and unification system
    and give an overview of
    their correctness proofs
    (\Sec{sec:inference}\tr{ and \Sec{app:inference-proofs}}{}).
  \item We experimentally demonstrate 
    significant improvements in
    the running times of
    programs in the standard \emph{nofib}
    \cite{Partain93:nofib} benchmark suite
    (\Sec{sec:bench}),
    which we ported to
    idiomatic OCaml code in
    and ran %
    through the OCaml compiler.
    Out of 38 ported programs,
    17 showed a significant speedup
    averaging 16.6\%.
    The overall speedup across all programs was 8.2\%.
    Moreover, \System did
    not significantly degrade the
    performance of any of the benchmarked programs 
    (only two programs slowed down noticeably, by an average of 1.8\%)
    and only increased their compiled binary size by an average 
    of
    1.79x.
\end{itemize}

\section{Presentation}
\label{sec:presentation}

Let us first consider 
the 
toy example in \Fig{fig:toy-example}.
In this program, an optional value
(constructed either by @Some@ or @None@)
is returned by @producer@ and sent to be destructed in @consumer@.
How to transform this program so that
no intermediate optional value is allocated at runtime?
In effect, we want to fuse the @Some@ and @None@ constructors
with their corresponding branches in the @case@ destructor.

\begin{figure}[h]
\begin{lstlisting}
          let consumer x = foo (case x of { Some v -> v + 1;  None -> 0 })
          let producer y = if y then Some 123 else None
          let main = consumer (producer bar)
\end{lstlisting}
\caption{A simplistic example program to fuse.}
\label{fig:toy-example}
\end{figure}

\subsection{Inlining in Anger}

The obvious thing to do here is to \emph{inline} both functions:
\begin{lstlisting}
  let main = foo (case if bar then Some 123 else None of { Some v -> v + 1; None -> 0 })
\end{lstlisting}
and perform some local rewriting (here @case@-of-@case@ commuting \cite{peytonjones1996compiling})
to bring together the introduction form (constructor) and the elimination form (destructor):
\begin{lstlisting}
  let main = foo (if bar then case Some 123 of { Some v -> v + 1; None -> 0 }
                         else case None     of { Some v -> v + 1; None -> 0 })
\end{lstlisting}
so as to annihilate them by simple pattern matching reduction:
\begin{lstlisting}
  let main = foo (if bar then 123 + 1 else 0)
\end{lstlisting}
The intermediate @Some@ and @None@ constructors have now been eliminated,
resulting in fewer allocations and an overall more efficient program.

This is essentially the approach followed by
the vast majority of 
\emph{general} deforestation techniques,
modulo (crucially) some extra care around recursion.
Recursive functions can be handled by trying to detect repetitive patterns encountered
while %
unfolding/inlining definitions,
which is typically done by some form of
\emph{homeomorphic term embedding} technique \cite{leuschel-homeo}.
In the worst case, after a recursion depth threshold is reached,
these approaches typically
stop the process and insert calls to the original definitions,
ensuring that the process eventually terminates.

However, while this ``\emph{inlining in anger}'' approach works in theory,
it is not reasonable in practice,
which explains why no industry-strength compiler has integrated it so far.
First, this sort of inlining causes code explosion.
This is especially true when there are multiple recursive call paths in the definitions being transformed
(a situation that often occurs in practice),
which leads the number of %
unfolded recursion paths to grow exponentially.
This tends to produce counter-productively large residual programs.
Second, 
the embedding process trying to find recurrent patterns 
to tie the knot is prohibitively expensive:
trying to tie the knot usually takes linear time in the number of terms we have currently
unfolded,
which we need to match against,
so the process takes quadratic time along \emph{each} recursion path,
of which there can be exponentially many.
Third, such solutions are typically very hard to implement correctly.
Of all the implementations of general deforestation that we have tried so far,
none of them was able to successfully transform all basic examples we tried on them,
and many produced silently incorrect programs or simply crashed (see \Sec{related-implementation}).

\subsection{Preserving the Structure of Definitions with Strategies}
\label{sec:preserve-struct-defs}

In this paper, we strive to avoid the ``naive'' approach of aggressive inlining
and instead try to be more cautious and targeted.
Instead of blindly inlining definitions into each other and hoping for the best,
we use a \emph{type inference and elaboration} technique to
try and preserve the original structure of definitions as much as possible,
duplicating only those definitions whose duplication is necessary to enable fusion.

Our idea is to first \emph{infer} some descriptive information,
in the form of \emph{types}, from the definitions of the original program.
The ``types'' we infer actually include information about the terms of the program
and how they can be transformed;
thus, we call such types \emph{fusion strategies} or just \emph{strategies}.
Concretely, a fusion strategy $S$ is one of:
\begin{itemize}
  \item A \emph{function} strategy $S_1 -> S_2$,
    similar to a function type,
    describing how to transform 
    functions and their arguments.
    $S_1$ and $S_2$ are respectively the \emph{parameter} and \emph{result} sub-strategies.
  \item A \emph{constructor application} strategy
    $\left\{\,
        c_1\,\langle S_{11},\,\ldots,\,S_{1n} \rangle
        \,,\ c_2\,\langle S_{21},\,\ldots \rangle
        \,,\ \ldots
      \,\right\}$,
    similar to a structural sum-of-products type,
    describing how to transform the arguments of a constructor application
    without fusing the constructor itself.
    Each constructor $c_i$ in the sum
    is associated with one sub-strategy $S_{ij}$
    for each of its fields.
  \item A \emph{constructor fusion} strategy
    $\left\{\,
        c_1\,\langle x_{11} |-> S_{11},\,\ldots,\,x_{1n} |-> S_{1n} \rangle->t_1
        \,,\ \ldots
      \,\right\}$,
    which is similar to a constructor application strategy,
    except that it includes terms $t_i$
    describing how to rewrite each constructor application
    into fused code that no longer allocates said constructor.
    The fields of each constructor are now named
    and referred to in
    the corresponding terms $t_i$.
  \item A \emph{recursive} strategy
    $\mu F_S$,
    where $F_S$ is a meta-level function
    (we use higher-order abstract syntax),
    similar to an equirecursive type,
    and equivalent to its unfolding $F_S(\mu F_S)$.
  \item A \emph{top} strategy $\top$,
    somewhat similar to a top type,
    indicating that we are not interested in transforming the corresponding term.
  \item A \emph{bottom} strategy $\bot$,
    somewhat similar to a bottom type,
    representing non-termination and the absence of values,
    and used mainly in the context of recursive types.
\end{itemize}

In the
example of \Fig{fig:toy-example},
we see that @consumer@ destructs constructors of the form @Some v@ into @v@
and constructors of the form @None@ into @0@.
The extracted field @v@ is consumed as a primitive integer and is thus not fused;
we associate it with a primitive strategy $\textsf{Int}$.
The inferred strategy for the parameter @x@ of @consumer@ is:
$$S_{\textsf{x}} \, = \, 
  \{\, \textsf{Some} \langle v |-> \textsf{Int} \rangle -> v + 1\,,\ 
    \textsf{None} \langle \rangle -> 0 \,\}$$
On the other hand,
we see that @producer@ is a function which results in a {union} of possible constructors (@Some@ and @None@).
Its inferred strategy is $S_{\textsf{producer}} \, = \, 
  \textsf{Bool}
  -> \al$,
where $\al$ is a meta-variable, representing a yet undetermined strategy,
with an associated \emph{lower bound}
  $\al >= \{\, \textsf{Some} \langle \textsf{Int} \rangle\,,\ 
    \textsf{None} \langle \rangle \,\}$
  indicating the two constructors that flow into that position.
Finally, we see that the @consumer@ and @producer@ definitions are \emph{cheap to duplicate},
in that doing so does not incur any additional runtime work,
so we can %
transform these definitions even if they
happen to also be used elsewhere (in which case we duplicate them
before the transformation, as mentioned later in \Sec{sub-sec:duplication-to-resolve-conflict},
so that the rest of the program is unaffected).
When analyzing the call site @consumer@ @(producer@ @bar)@,
we discover that $\al$ can be substituted with fusion strategy
$\{\, \textsf{Some} \langle v |-> \textsf{Int} \rangle -> v + 1\,,\ 
    \textsf{None} \langle \rangle -> 0 \,\}$,
    which means that deforestation can proceed:
we can \emph{import} the relevant bits of the consumer \emph{into} the definition of the producer,
replacing the original pattern matching expression in the consumer
by %
the scrutinee itself (formerly @x@, now @x2@)
as follows:
\begin{lstlisting}
  let consumer2 x2 = foo x2
  let producer2 y = if y then let v = 123 in v + 1 else 0
  let e = consumer2 (producer2 bar)
\end{lstlisting}
Notice that in the rewritten program,
the type of @x@ (now @x2@) has \emph{changed}:
it used to hold \emph{an optional value}
(constructed with either @Some@ or @None@),
but in the rewritten program it now holds \emph{the result of consuming this optional value}.
Therefore, our approach is not \emph{locally type-preserving},
even though it {is} \emph{type-safe}:
it keeps programs well-typed
(in our type system)
although the types of individual function definitions and intermediate terms may change.

It turns out that performing type inference for such \emph{producer} and \emph{consumer} strategies
is most conveniently done through a \emph{polarized subtype inference} approach
similar to algebraic subtyping \cite{dolan_polymorphism_2017,dolan_algebraic_2017,Parreaux20:simple-essence-alg-subt}.
In this framework,
producers are \emph{positive} types and consumers are \emph{negative} types,
as explained in \Sec{sec:inference}.

\subsection{Preserving Call-By-Value Evaluation Order}
\label{sec:preserve-cbv}

The transformation demonstrated in the previous subsection
is perfectly fine in a pure call-by-name setting,
where the order of expressions in a program does not matter.
But in a call-by-value context,
even assuming \emph{pure} functional programming (\ie no mutation, I/O, etc.),
one cannot simply move expressions around 
in this manner
because doing so could change the program's termination behavior.
Consider the following program:
\begin{lstlisting}
  let foo x y = if x then case y of { () -> error () } else 0
  let main x = foo x ()
\end{lstlisting}
where @error ()@ crashes the program (for example by running into an infinite loop)
and where expression @()@ and pattern @()@ denote the unit value
(a data constructor with no fields).
If we fuse this program following the %
simple approach outlined %
above, we obtain the following result:
\begin{lstlisting}
  let foo2 x y2 = if x then y2 else 0
  let main x = foo2 x (error ())
\end{lstlisting}
We replaced the construction of a unit value @()@ by the result of consuming it
as was originally done in the corresponding branch of the @foo@ function.
This transformation is \emph{unsound} because the result has different call-by-value semantics:
calling @main@ in the new program will crash regardless of the value of its parameter @x@,
whereas the original program would only crash when @x = True@.

To solve this problem, we first \emph{thunk}
all computations from pattern matching branches
and execute these thunks outside the corresponding destructors,
yielding:
\begin{lstlisting}
  let foo x y = if x then (case y of { () -> fun () -> error () }) () else 0
  let main x = foo x ()
\end{lstlisting}
so that the actual rewritten program becomes:
\begin{lstlisting}
  let foo2 x y2 = if x then y2 () else 0
  let main x = foo2 x (fun () -> error ())
\end{lstlisting}
The astute reader will have noticed that
this transformation %
does not look like an improvement on its own:
we replace the use of a \emph{data constructor} value (here @()@)
with a \emph{function} value (here @fun () -> error ()@),
and function values are also be allocated on the heap!
In fact, the real performance gains we observe using our technique come from
its ability to \emph{reorganize} recursive functions by moving computations around
and \emph{then} performing some simplifications which actually reduce the number of overall allocations.
Our transformation does not \emph{always} lead to significant simplification opportunities
down the line, and thus does not necessarily lead to improved performance,
but
our experiments on the standard \emph{nofib} benchmark suite show that
(1) it can significantly improve the efficiency of many programs
in practice;
and (2) it 
rarely leads to noticeable performance degradations,
and when it does, those are very small.
Overall, it %
is a net positive.

\subsection{Fusing Recursive Functions}
\label{subsec:fuse-rec}

Let us now consider how to transform recursive programs.
As a simple example, consider the following program,
which builds a list and then immediately sums up its elements:
\begin{lstlisting}
  let rec enumerate n = if n >= 0 then n :: enumerate (n - 1) else []
  let rec sum xs = case xs of { [] -> 0; x :: xs -> x + sum xs }
  let main x = sum (enumerate x)
\end{lstlisting}
Our approach handles recursive functions by inferring \emph{recursive fusion strategies}.
There is nothing particularly difficult in doing so:
we use precisely the same %
technique as previous type inference %
approaches to inferring equi-recursive types
(like that of, for example,
\citet{dolan_polymorphism_2017,Parreaux20:simple-essence-alg-subt,Parreaux22:mlstruct}).
In the case above, the final inferred strategies are:
\begin{align*}
  S_{\textsf{sum}}
  \ &= \ 
    S_{\textit{fused}} -> \textsf{Int}
  &
  \quad
  \textit{where:}
  \quad
  &&
  \\
  S_{\textsf{enumerate}}
  \ &= \ 
    \textsf{Int} ->
    S_{\textit{fused}}
  &
  S_{\textit{fused}}
  \ &= \ 
    \mu(\lam \al.\ \{\, \textsf{::} \langle x |-> \textsf{Int},\, xs |-> \al \rangle -> x + \textsf{sum}\ xs
      \,,\ \textsf{[]} \langle \rangle -> 0 \,\})
\end{align*}
We can now
apply the same fusion technique as described previously in
\Sec{sec:preserve-struct-defs} and \Sec{sec:preserve-cbv},
matching the producer and consumer strategies and moving the consumer's computations to the producer,
finally arriving at the program below:
\begin{lstlisting}
  let rec enumerate2 n = if n >= 0 then let x = n in let xs = enumerate2 (n - 1) in
                                        fun () -> x + sum2 xs
                                   else fun () -> 0
\end{lstlisting}
\vspace{-1em}
\begin{lstlisting}
  let rec sum2 xs = xs ()
  let main x = sum2 (enumerate2 x)
\end{lstlisting}
Note that it is crucial to extract
the @enumerate2 (n - 1)@ computation into an @xs@ binding
which occurs \emph{outside} of the thunking function literal.
This is a necessary precaution to avoid changing
the evaluation order and amount of runtime work performed by the original program,
which could have disastrous effects
such as changing its algorithmic complexity.

Again, at this point, we have not yet truly improved the efficiency of our program
---
we have simply replaced the allocation of intermediate list cells with the allocation of function values.
But we can now improve this program by
(1) inlining @x@ and the @sum2@ function, which %
are now trivial:
\begin{lstlisting}
  let rec enumerate2 n = if n >= 0 then let xs = enumerate2 (n - 1) in fun () -> n + xs ()
                                   else fun () -> 0
\end{lstlisting}
\vspace{-1em}
\begin{lstlisting}
  let main x = enumerate2 x ()
\end{lstlisting}
and
(2) \emph{floating} the `@fun () ->@'
function binders \emph{out} the let binding and conditional expression,
resulting in the following definition of @enumerate2@:
\begin{lstlisting}
  let rec enumerate2 n () = if n >= 0 then let xs = enumerate2 (n - 1) in n + xs () else 0
\end{lstlisting}
after which we can inline @xs@, since it is used immediately and only once,
and drop the @()@ parameter:
\begin{lstlisting}
  let rec enumerate2 n = if n >= 0 then n + enumerate2 (n - 1) else 0
\end{lstlisting}
This program sums up the integers from @n@ down to @0@
without allocating any intermediate %
data structures
or function values,
achieving our goal of deforestation.

We must now slow down and carefully consider the validity of the latter transformation.
\emph{Floating} function binders \emph{out} of conditional expressions
and let bindings
is a {leap}
that cannot always be made while preserving the precise call-by-value semantics of programs.
First, it can lead to \emph{more termination} than in the original program.
For instance, notice that
@case error () of { () -> fun () -> () }@
never terminates while its floated-out version
@fun () -> case error () of { () -> () }@
terminates.
We do not consider this a serious problem
because non-termination in pure programs can be treated as
a form of undefined behavior.\!\footnote{%
  Indeed, the standard definition of contextual approximation used
  to show the soundness of a transformation
  considers that a program that terminates soundly approximates
  a program that does not terminate under an empty context.
}
Second, and more worryingly,
floating out could result in additional runtime work being performed
by the program.
Thankfully, two situations commonly arise where
runtime work duplication is guaranteed not to happen:
first, when the function is floated out across
\emph{cheap} expressions like variables and lambdas;
second, when the function being floated out is {known} to be ``one-shot'',
\ie it is always applied at most once.

For instance, it is safe to float out `@fun n ->@' in the following program:
\begin{lstlisting}
  let foo x   = case some_computation() of { () -> fun n -> n + x } in foo 1 2 + foo 3 4
\end{lstlisting}
resulting in\cprotect\footnote{We write
  \,\lstinline[basicstyle=\scriptsize\ttfamily]@let x p = e@\,
  as a shorthand for
  \,\lstinline[basicstyle=\scriptsize\ttfamily]@let x = fun p -> e@.
}
\begin{lstlisting}
  let foo x n = case some_computation() of { () ->          n + x } in foo 1 2 + foo 3 4
\end{lstlisting}
because @foo@ is always applied to two arguments,
so the results of @foo@ @1@ and @foo@ @2@
are never applied more than once.
However, %
floating out may not be safe in a program where @foo@ is used as in
`@foo 1 2 + bar (foo 3) 4@' unless we also know that 
the first argument of @bar@ is itself {one-shot}.

Determining whether it is safe to float out the function binder in the motivating example of @enumerate2@
is slightly nontrivial. Because @enumerate2@ is recursive,
this requires a simple form of inductive reasoning
on the recursive call contexts or, equivalently, on the call stack.
The base cases %
are the \emph{outer calls} to @enumerate2@.
Here, there is only one,
@enumerate2@ @x@, which is one-shot as it is immediately applied to @()@.
The inductive cases %
are the inner \emph{recursive calls} to @enumerate2@.
Here, the recursive call @enumerate2 (n - 1)@ is stored
in the @xs@ binding, which is captured
by the returned lambda and \emph{not} immediately applied,
so it is not ``obviously'' known to be one-shot.
But \emph{assuming} that the returned lambda is one-shot,
we can conclude that
@xs@ is also one-shot,
concluding the induction.
This reasoning justifies the %
safety of our
transformation
of @enumerate2@.

To determine when performing these floating-out transformations is safe,
we use a simple heuristic 
\emph{cardinality analysis} 
\cite{arity,xu2005arity,BREITNER201865,Sergey-popl14:cardinality}
following the general %
ideas described above.

\vspace{-0.6em}
\subsection{Fusing Computations with Free Variables}
\label{sec:computations-fvs}

We are now finally ready to study how the motivating example of
\Fig{fig:map-of-map}
(a double @map@ application) is fused by \System.
The first step is to duplicate the @map@ definition so that it can be fused
individually as a producer and as a consumer.
\System does this automatically when it finds fusion opportunities
through type inference, as explained later in \Sec{sub-sec:duplication-to-resolve-conflict}.
\begin{lstlisting}
  let rec map1 f1 xs1 = case xs1 of { [] -> []; x1 :: xs1  -> f1 x1 :: map1 f1 xs1 }
  let rec map2 f2 xs2 = case xs2 of { [] -> []; x2 :: xs2  -> f2 x2 :: map2 f2 xs2 }
  let main ls = map1 incr (map2 double ls)
\end{lstlisting}
Now, there is an important technicality to beware of.
Because the consumer @map1@ takes a parameter @f1@
that is used in its @case@ branches, its computations \emph{cannot} be immediately
moved to %
@map2@, lest we end up with a broken program that contains a free variable:
\begin{lstlisting}
  let rec map12 f1 xs1 = xs1 ()
  let rec map22 f2 xs2 = case xs2 of { [] -> fun () -> []; x2 :: xs2 ->
        let x = f2 x2 in let xs = map22 f2 xs2 in fun () -> $\color{red}\tt f1$ x :: map12 $\color{red}\tt f1$ xs }
\end{lstlisting}
Our solution is simple:
at the time we \emph{thunk} the computations of the consumer,
instead of the unit value @()@, we can pass the variables captured by the computations (here just @f1@),
resulting in:
\begin{lstlisting}
  let rec map12 f1 xs1 = xs1 ($\color{blue}\tt f1$)
  let rec map22 f2 xs2 = case xs2 of { [] -> fun ($\color{blue}\tt f1$) -> []; x2 :: xs2 ->
        let x = f2 x2 in let xs = map22 f2 xs2 in fun ($\color{blue}\tt f1$) -> $\color{blue}\tt f1$ x :: map12 $\color{blue}\tt f1$ xs }
\end{lstlisting}
After floating out the inner functions following the same reasoning as in \Sec{subsec:fuse-rec}
and performing inlining of trivial definitions,
we end up with the following rewritten program:
\begin{lstlisting}
  let rec map22 f2 xs2 f1 = case xs2 of { [] -> []; x2 :: xs2 -> f1 (f2 x2) :: map22 f2 xs2 f1 }
  let main ls = map22 double ls incr
\end{lstlisting}
which is %
isomorphic to the fused program we desired in 
the introduction of this paper.

\subsection{Limitations}
\label{subsec:limitations}

To conclude this presentation,
we list some of our main limitations,
to be addressed in future work.

\subsubsection{Cardinality Analysis and Floating Out}

The heuristic arity analysis  mentioned in \Sec{subsec:fuse-rec}
currently will not always find
all cases where floating out functions binders is safe
and may also too aggressively float out lambdas in cases
where this could potentially increase runtime work.
Nevertheless, 
our experiments seem to show that
this does not outweigh the benefits of our transformations.

\subsubsection{Nonfusible Programs}
\label{subsec:accum-fusion}

An interesting case is when eliminating intermediate functions by floating out
is fundamentally impossible. This happens for example in
functions building their intermediate data structures in accumulator parameters.
For instance, consider the following program, which reverses a list and then maps the result:
\begin{lstlisting}
  let rec rev xs acc = case xs of { x :: xs -> rev xs (x :: acc); [] -> acc }
  let main xs = map incr (rev xs [])
\end{lstlisting}
Our transformation turns this programs into the following one
after matching the recursive consumption of @map@ against the recursive production of @rev@:
\begin{lstlisting}
  let rec map2 f xs2 = xs2 f
  let rec rev2 xs acc2 = case xs of { x :: xs -> rev2 xs (fun f -> f x :: map2 f acc2);
                                      []      -> acc2 }
  let main xs = map2 incr (rev2 xs (fun f -> []))
\end{lstlisting}
Here we can see that the thunking `@fun f ->@' function is ``stuck''
inside a recursive accumulator argument and cannot be floated out.
Fusing such tail-recursive functions while ensuring that the number of allocations
actually decreases 
would require the accumulated computations to be associative
\cite{gibbons2021continuation}.
In future work, we could investigate %
automated associativity reasoning,
which would allow fusing expressions like @sum (rev xs [])@.
Failing that, we could at least recover a program isomorphic to the original,
which was possibly more efficient,\!\footnote{
  A concrete data structure is often more efficient than
  a series of closures, as the latter involves virtual dispatch.
}
by performing a pass of defunctionalization.
We have not yet implemented such as pass
(see also \Sec{sec:accum-cont}).

As a side note, no such problem occurs when @rev@ is a \emph{consumer},
as in @rev (map incr xs)@, which reverses a list that is being mapped and fuses just fine in \System.

\subsubsection{Misaligned and Conflicting Producer/Consumer Pairs}
\label{sub-sec:misaligned-length}

Because we \emph{import} consumer computations \emph{into} producers,
in general, we can only fuse at most one consumption strategy for any given producer expression.
So while \System still fuses a program like @map f (0 :: map g xs)@ without problems,
where the producer is wrapped inside an extra constructor,
the dual situation @case map f xs of [] -> []; a :: as -> map g as@,
where the producer is wrapped inside an extra destructor,
cannot be fused.
Indeed, we get a strategy clash: several candidate consumers of @::@ and @[]@ become
available to rewrite the producer's list constructors.\!\footnote{
  This could possibly be alleviated by fusing the \emph{producer} into the consumer,
  but we have not yet investigated this possibility.}
A similar situation arises when the periodicity (or ``recursion length'') of a producer is mismatched
with that of a consumer; for instance, consider:
\begin{lstlisting}
  let rec pair_up xs = case xs of { x :: y :: xs -> (x, y) :: pair_up xs; _ -> [] }
  let rec mk n = if n > 1 then (n - 1) :: n :: (n + 1) :: mk (n - 3)
  let main x = pair_up (mk x)
\end{lstlisting}
Here, @pair_up@ consumes two elements per recursion
while @mk@ produces three.
This leads to a strategy clash 
which prevents fusion.
There is an easy fix in both of these situations:
we could inline and unroll recursive definitions just enough to make
the producer strategies align with the consumer strategies.
From our experiments, doing this %
makes some of the benchmarked programs around 10\% faster.
However, this tends to lead to higher code duplication and longer compilation times,
so we did not include this technique in our experiments.
In the future, a better-tuned heuristic
(notably avoiding the unrolling of recursive definitions with more than one recursion path)
could lead to retaining the speedups without incurring significantly more code duplication.

\subsection{Characterization of \System's Capabilities}
\label{subsec:characterization}

All in all, we can characterize
\System's fusion capabilities
as requiring the following conditions:
\begin{itemize}
  \item \textbf{Constructors are uniquely consumed.}
    For any given constructor application to be fused away,
    it must be matched with at most one corresponding pattern matching expression.\!\footnote{
      Note that this \emph{static} requirement is stronger than the classical notion of \emph{runtime} linearity:
      fusion is hampered when there is more than one consumer in the program text,
      even when the data structure might actually be consumed only once at runtime.
      For instance, the $\mathit{xs}$ consumer `$\kws{if} x \skws{then}\id{length}\ \mathit{xs}\skws{else}\id{sum}\ \mathit{xs}$'
      will not be able to fuse with $\mathit{xs}$'s producer.
    }
    Otherwise, we cannot know statically to which expression the constructor should be rewritten.
  \item \textbf{Constructors are in tail position.}
    The (possibly nested) constructor applications to be fused
    must be returned immediately by the producer
    --- \ie they should be ``in tail position''
    (see \Sec{subsec:accum-fusion}).
    In particular, if a data structure is produced through an accumulator parameter,
    then after rewriting
    we end up with a program that allocates a shadow
    of the original data structure encoded as closures, which is likely not an improvement on its own.
  \item \textbf{Recursions align.}
    The consumer must process elements at a rate that is a divisor of the producer's rate
    (see \Sec{sub-sec:misaligned-length}).
    For instance, it is fine if four elements are produced per producer recursion
    and two are consumed per consumer recursion,
    but the converse is not.
  \item \textbf{Rewritings are acyclic.}
    In general, performing fusion can lead to new fusion opportunities arising.
    In some rare cases, there might be an infinite number of such opportunities
    arising as we transform the program.\!\footnote{
      In this case, \System detects the cyclic behavior and aborts the rewriting.
    }
    We are not aware of this occurring in any real-world setting,
    as we have only been able to elicit such behavior on artificially contrived examples.
\end{itemize}

\section{Formalization}
\label{sec:formal}

We now formalize our approach.
We rigorously define the concept of \emph{fusion strategy}
and formally describe how %
fusion strategies are used to transform programs.
Our development is in the form of a declarative type system
with \emph{elaboration},
meaning that we transform terms based on
their types.

\subsection{Syntax}
\label{sec:formal-syntax}

\begin{figure}
{\small
\begin{align*}
e,t &~::=~
  x
  ~\mid~
  t\ t
  ~\mid~
  c\ \overline{t_j}^j
  ~\mid~
  \lam x.\ t ~\mid~ 
  \kw{case}\ t\ \kw{of}\ \overline{c_i\ \overline{x_{i,j}}^j -> l_i}^i
  ~\mid~
  \kw{let rec}\ x = t\ \kw{in}\ t
\\
v &~::=~
  c\ \overline{v_j}^j
  ~\mid~
  l
  \qquad
  \textit{where}
  \quad
  l ~::=~ \lam x.\ t
\\
S &~::=~
  \top
  ~\mid~
  \bot
  ~\mid~
  S -> S
  ~\mid~
  \mu F_S
  ~\mid~
  \left\{\,\overline{
    c_i\,\langle\,\overline{\,S_{i,j}\,}^j\,\rangle
  }^i\right\}
  ~\mid~
  \left\{\,\overline{
    c_i\,\langle\,\overline{x_{i,j} |-> S_{i,j}}^j\,\rangle->l_i
  }^i\,\right\}
\\
I &~::=~
  \top
  ~\mid~
  \bot
  ~\mid~
  I -> I
  ~\mid~
  \mu F_I
  ~\mid~
  \left\{\,\overline{
    c_i\,\langle\,\overline{\,I_{i,j}\,}^j\,\rangle
  }^i\,\right\}
\\
\Gamma &~::=~
  \ep
  ~\mid~
  \Gamma \cdot (x|->S)
\\
\Xi &~::=~
  \ep
  ~\mid~
  \Xi \cdot (x|->I)
\end{align*}
}
\caption{Syntax of the formal calculus.}
\label{fig:syntax-new-new}
\end{figure}

The syntax of our formalism is given in \Fig{fig:syntax-new-new}.
Metavariable $S$ denotes \emph{full} fusion strategies guiding the transformation of terms,
while $I$ represents \textit{identity} fusion strategies,
which do not lead to any fusing transformations of the corresponding terms.

We use indexed overlines to denote the repetition of syntax constructs.
This leads to a more concise presentation of constructs like
algebraic data types, which have multiple constructors where each constructor has
multiple fields.
For instance, we write
$\left\{\overline{c_i\ \langle\overline{S_{i,j}}^j \rangle}^i\right\}$
to denote a set of constructors
$\overline{c_i}^i$ where each constructor $c_i$ has multiple fields with associated strategies $\overline{S_{i,j}}^j$.

Notice that the syntax of $I$ is a proper subset of the syntax of $S$.
One can regard $I$
as the usual \textit{types} in a system with functions,
structural sum and product types,
recursive types, bottom, and top.
Full fusion strategies $S$
simply extend $I$
with 
\emph{constructor fusion strategies} $\left\{\overline{c\ \langle\overline{x |-> S}\rangle -> l}\right\} $ that attach
the information needed to perform fusion:
$l$ is the body of the pattern-matching %
branch that needs to be imported
into the site where data constructor $c$ is called
and $x |-> S$ represents the variables bound in the matching pattern
along with their corresponding strategies.
Note that branch bodies in $\kw{case}$ expressions and
the corresponding terms %
in {constructor fusion strategies}
are \emph{lambda abstractions} $l$ instead of general terms $t$.
This is to ensure that call-by-value semantics is preserved,
as explained in \Sec{sec:preserve-cbv},
and is not a limitation of our system:
all programs can be put into this form via a simple
preprocessing pass we call \emph{thunking},
described in \Sec{sub-sec:thunking-free-vars}.

Recursive types in $S$ and $I$ are denoted as
$\mu F_S$ and $\mu F_I$,
where $F_S$ and $F_I$ are 
\emph{functions} which take a strategy as input and output another strategy.
This \emph{higher-order abstract syntax} approach \cite{pfenning88:hoas}
allows for a cleaner formal development and was also
adopted in the %
work that introduced the proof technique we use \cite{appel-equi-rec-type,amal-step-index}.
We use $F_S(S_0)$ to denote the application of
function $F_S$ to %
strategy $S_0$.

\paragraph{Shorthands and syntactic sugar} We make use of the following shorthands and syntactic sugar:
\begin{itemize}
  \item Because we support recursive types, we can directly define fixed
  point combinators, so the \kw{let rec}
  construct is simply treated as syntactic sugar used
  to clarify our examples.

  \item Non-recursive \kw{let} bindings are also
  used as a clarifying notation for lambda abstractions that are
  immediately applied.
  As we shall see,
  this helps clarify the fusion rule \textsc{F-Ctor}.

  \item In the examples, we use $::$ and $[]$ as the usual list constructors.
  And similar to \Sec{sec:preserve-cbv}, ``$()$'' is
  used as a shorthand for the value
  of a $\mathit{Unit}$ constructor taking no parameters,
  and ``$\lam().\,e$'' is a shorthand for $\lam x.\,\kw{case }x\kw{ of }\mathit{Unit}->e$
  representing a \emph{thunked} computation.
  
  \item By abuse of notation,
    in examples, we will often write fusion strategies
    of the form $S_0 = \mu(F_S)$ as just $S_0 = F_S(S_0)$,
    for example writing $S = \top -> S$
    instead of $S = \mu(\lam\al.\ \top -> \al)$.
\end{itemize}

\subsection{Fusion Rules}\label{fusionrulessubsection}
\newcommand{\dummy}{()}

The fusion rules are given in \Fig{fig:fusion-prescient}.
Judgement $\Gamma |- s ~> t <= S$
reads ``in context $\Gamma$, term $s$ rewrites into term $t$
under fusion strategy $S$'',
and $\fname{FV}(t)$ denotes the set of free variables in $t$.
Rules
\textsc{F-Var},
\textsc{F-Lam},
\textsc{F-LetRec},
and
\textsc{F-App}
are standard and simply transform subterms recursively.
\textsc{F-Top}
is obvious: any well-typed term can be typed as $\top$
(essentially the rule proposed by \citet{MacQueen-86:ideal-model-rec-poly-tys}
in their seminal work).
This is to handle terms we are not interested in transforming.
Rules \textsc{F-RecUnfold} and \textsc{F-RecFold} deal with folding and unfolding
equirecursive types; they are also standard.
\textsc{F-CtorSkip} and
\textsc{F-CaseSkip}
are used when a constructor or destructor shall not be fused;
they simply transform the subterms recursively.
\textsc{F-Ctor} and
\textsc{F-Case} 
are used to fuse a constructor and a matching $\kw{case}$ expression.
In \ruleName{F-Ctor}, constructors are rewritten to let bindings followed by the computation
imported from their corresponding $\kw{case}$ branch, which is book-kept in the strategy;
in \ruleName{F-Case}, $\kw{case}$ expressions are
rewritten into the results of transforming their scrutinees.
Since we do not wish to change the result of
the main top-level term, 
we restrict
the root of any fusion derivation tree to be of the form
$\ep |- t ~> t' <= I$, where $I$ is an \textit{identity} fusion
strategy, which is always type- and semantics-preserving.

\begin{figure*}
\mprset{sep=1.3em}
{\small
\begin{mathpar}
\boxed{\mathstrut\Gamma |- t ~> t <= S}
\hfill
\\
\inferrule[F-Var]{
  (x |-> S) \in \Gamma
}{
  \Gamma |-
  x ~> x
  <= S
}

\inferrule[F-Lam]{
  \Gamma\cat(x|->S') |-
  t
  ~>
  t'
  <= S
}{
  \Gamma |-
  \lam x.\ t
  ~>
  \lam x.\ t'
  <= S' -> S
}

\inferrule[F-RecFold]{
  \Gamma |- t ~> t' <= S(\mu S)
}{
  \Gamma |- t ~> t' <= \mu S
}

\inferrule[F-RecUnfold]{
  \Gamma |- t ~> t' <= \mu S
}{
  \Gamma |- t ~> t' <= S(\mu S)
}

\inferrule[F-LetRec]{
  \Gamma \cat (x |-> S_1) |-
  t_1 ~> t_1'
  <= S_1
  \\\\
  \Gamma \cat (x |-> S_1) |-
  t_2 ~> t_2'
  <= S_2
}{
  \Gamma |-
  \kw{let rec}\ x = t_1\ \kw{in}\ t_2
  ~>
  \kw{let rec}\ x = t_1'\ \kw{in}\ t_2'
  <=
  S_2
}

\inferrule[F-App]{
  \Gamma |- t_1 ~> t_1'
  <= (S' -> S)
  \\
  \Gamma |- t_2 ~> t_2'
  <= S'
}{
  \Gamma |-
  t_1\ t_2 ~> t_1'\ t_2'
  <= S
}

\inferrule[F-Ctor]{
  \fname{FV}(l') \subseteq \left\{\overline{x_{j}}^j\right\}
  \\\\
  \left( c\,\langle\,\overline{x_{j} |-> S_{j}}^j\,\rangle->l' \right) \in S
  \\
  \overline{\,\Gamma |- t_{j} ~> t_{j}' <= S_{j}\,}^j
}{
  \Gamma |-
  c\ \overline{t_{j}}^{j} ~>
  \kw{let } \overline{x_{j} = t_{j}'}^j \kw{ in } l'
  <= S
}

\inferrule[F-CtorSkip]{
  \left(c\,\langle \overline{S_j}^j \rangle\right) \in S
  \\
  \overline{\, \Gamma |- t_j ~> t_j' <= S_j\,}^j
}{
  \Gamma |-
  c\ \overline{\, t_j\,}^j ~> c\ \overline{\, t_j'\,}^j
  <=
  S
}

\inferrule[F-Case]{
  \Gamma |- t ~> t' <= \left\{\overline{c_i\,\langle \overline{x_{i,j} |-> S_{i,j}}^j \rangle->l_i'}^i \right\}
  \\
  \overline{\,
    \Gamma\cat
    \overline{
      (x_{i,j} |-> S_{i,j})
    }^j
    |-
    l_i ~> l_i' <= S
  \,}^i
  \\
  \overline{\fname{FV}(l_i) \cup \fname{FV}(l_i') \subseteq \left\{ \overline{x_{i,j}}^j \right\}}^i
}{
  \Gamma |-
  \kw{case}\ t\ \kw{of}\ \overline{\,
    c_i\ \overline{x_{i,j}}^j -> l_i
  \,}^i
  ~>
  t'
  <=
  S
}

\inferrule[F-CaseSkip]{
  \Gamma |- t ~> t' <= \left\{\overline{c_i\,\langle \overline{S_{i,j}}^j \rangle}^i\right\}
  \\
  \overline{\,
    \Gamma \cat \overline{(x_{i,j} |-> S_{i,j})}^j |- l_i ~> l_i' <= S
  \,}^i
}{
  \Gamma |-
  \kw{case}\ t\ \kw{of}\ \overline{\,
    c_i\ \overline{x_{i,j}}^j -> l_i
  \,}^i
  ~>
  \kw{case}\ t'\ \kw{of}\ \overline{\,
    c_i\ \overline{x_{i,j}}^j -> l_i'
  \,}^i
  <=
  S
}

\inferrule[F-Top]{
  \Gamma |- t ~> t' <= I
}{
  \Gamma |- t ~> t' <= \top
}
\end{mathpar}}
\caption{Fusion rules.
  }
\label{fig:fusion-prescient}
\end{figure*}

\subsection{Examples}

Let us explain these rules through a few formal examples.
In our examples, those identifiers that are not
in italics are considered to be top-level bindings
always in scope.
So they do not count as free variables and do not invalidate
the \emph{free variables} premises of \textsc{F-Ctor}
and \textsc{F-Case}.
Also, for %
clarity, the %
example source programs below are \emph{not}
``{thunked}'' 
---
the necessary \emph{thunking} of all relevant terms
is %
implicitly done on the fly as fusion strategies are introduced.

\subsubsection{Simplest Example}

Consider the term $t_1\ t_2$,
where
\begin{align*}
  t_1 &~=~ \lam p.\
    \id{f}\ (\kw{case}\ p\ \kw{of}\ x::xs -> \id{g}\ (\kw{case}\ xs\ \kw{of}\ y::ys -> \id{h}\ y\ ys))
  \\
  t_2 &~=~ 1::2::[]
\end{align*}
Here $t_1$ is the \emph{consumer}
of the list %
we want to fuse.
It consumes the list through two simple non-recursive strategies
--- call them $S_1$ followed by $S_2$.
On the other hand,
$t_2$ \emph{produces} the corresponding list constructors,
also non-recursively,
making the fusion quite straightforward.

Viewing $t_1$ as a function that takes an input of
strategy $S_1$ and produces something that we don't want to
transform any further,
we can formally show
using
\textsc{F-Lam},
\textsc{F-App},
\textsc{F-Case},
and
\textsc{F-Var}
that
$\ep |-
  t_1 ~> \lam p.\ \id{f}\ (p\ \dummy) <= S_1 -> \top
$
where
\SP
$S_1 =
  \left\{\,
    \left( \left(::\right)\,\langle x |-> \top,\,xs |-> S_2 \rangle -> \lam\dummy.\ \id{g}\ (xs\ \dummy)\right)
  \,\right\}
$
and
\\
\UP
$S_2 =
  \left\{\,
    \left( \left(::\right)\,\langle y |-> \top,\,ys |-> \top\rangle -> \lam\dummy.\ \id{h}\ y\ ys\right)
  \,\right\}
$.
\\ 
We see that the $S_1$ strategy
consumes the $(::)$ constructor, and then
continue to consume its second field $xs$
using strategy $S_2$,
which continues into (stops at) the %
top strategy 
$\top$.
Moreover,
we can also show
$\ep |- t_2 ~> t_2' <= S_1$
where
$t_2' = \kw{let }
x=1;\;
xs=(\kw{let }y=2;\;ys=[]\kw{ in } \lam ().\;\id{h}\ y\ ys)
\kw{ in }
\lam ().\;\id{g}\ (xs\ \dummy)$
so that
$\ep |-
  t_1\ t_2 ~> (\lam p.\ \id{f}\ (p\ \dummy))\ t_2' <=
  \top
$.

\subsubsection{Simple Recursive Example}
\label{sec:simple-rec-example-sumi-mapsqi}
Consider the %
simple problem of
fusing sum and map on infinite lists.
Although in this paper we are
mainly focused on call-by-value semantics,
this simple recursive example is helpful to illustrate our formal system
and it %
would work fine in a call-by-need setting,
where our transformations are also valid,
assuming lazy natural number addition.
\begin{align*}
  \id{sumi}\ (\id{mapsqi}\ L)
  \quad
  \kw{where}
  &\quad\id{sumi}\ (x::xs) ~=~ x + \id{sumi}\ xs 
  \\
  &\quad\id{mapsqi}\ (y::ys) ~=~ (y * y) :: \id{mapsqi}\ ys
\end{align*}

We start by analyzing the presumed \emph{consumer}
$\id{sumi}$,
whose body is
$$t_{\id{sumi}} =
  \lambda\,p.\ 
  \kw{case}\ p\ \kw{of}\ x::xs -> x+\id{sumi}\ xs$$

We can show that
$(\id{sumi}|->(S_1 -> \top)) |- t_{\id{sumi}} ~> t_{\id{sumi}}' <=
  S_1 -> \top
$
where
$t_{\id{sumi}}' = \lambda\,p.\ (p\ \dummy)$
and
$$S_1 =
  \left\{\,
    \left((::)\,\langle x |-> \top,\, xs |-> S_1 \rangle -> \lam\dummy.\ x+\id{sumi}\ xs \right)
  \,\right\}
$$

Now, we can use this strategy $S_1$
when transforming the \emph{producer} $\id{mapsqi}$,
whose body is
$$
t_{\id{mapsqi}} =
  \lambda\,p.\
  \kw{case}\ p\ \kw{of}\ y::ys -> (y * y)::\id{mapsqi}\ ys
$$

We can show that
$(\id{mapsqi}|->S_{\fname{list}}->S_1)
|-
t_{\id{mapsqi}} ~> t_{\id{mapsqi}}'
<= S_{\fname{list}}->S_1$
where
$t_{\id{mapsqi}}' =
  \lambda\,p.\ 
  \kw{case}\ p\ \kw{of}\ y::ys ->
  \kw{let } x = (y * y);\;xs = \id{mapsqi}\ ys
  \kw{ in } \lam\dummy.\,x + \id{sumi}\ xs
$, and 
$S_{\fname{list}} = \left\{
  \left((::)\,\langle \top,\,S_{\fname{list}} \rangle \right)
\right\}$, which is an
identity strategy and means that
there will be no transformation performed on the related $\kw{case}$ expression
in $t_{\textsf{mapsqi}}$ because the input list $L$ is unknown.

Then we can rewrite the original program to the fused
$\id{sumi}\ (\id{mapsqi}\ L)$ where
\begin{align*}
  &\quad\id{sumi}\ p ~=~ p\ \dummy 
  \\
  &\quad\id{mapsqi}\ (y::ys) ~=~
  \kw{let } x = (y * y);\;xs = \id{mapsqi}\ ys
  \kw{ in } \lam\dummy.\, x + \id{sumi}\ xs
\end{align*}

\subsubsection{Additional Examples}

More explanations and examples 
are shown in \App{app:advanced-examples}.

\section{Correctness}
\label{sec:correctness}

\newcommand{\prog}[1]{\xhookrightarrow{\mathit{#1}}}
\newcommand{\calE}[2]{\mathcal{E}_{#1}[|#2|]}
\newcommand{\calG}[2]{\mathcal{G}_{#1}[|#2|]}
\newcommand{\calV}[2]{\mathcal{V}_{#1}[|#2|]}
\newcommand{\calEeq}[2]{\mathcal{E}'_{#1}[|#2|]}
\newcommand{\calGeq}[2]{\mathcal{G}'_{#1}[|#2|]}
\newcommand{\calVeq}[2]{\mathcal{V}'_{#1}[|#2|]}
\newcommand{\defeq}{\stackrel{\mathit{def}}{=}}

We show the correctness of our transformation using the standard technique
of proving contextual approximation between the input and output programs via
a method based on step-indexed logical relations.
This technique allows us to give relational interpretations
of strategies as indexed sets of pairs of values or terms
related by our transformation,
going beyond a simple syntactic treatment.

We first explain several important auxiliary definitions
and then introduce the main theorems for the soundness of our transformation.
In the following paragraphs, $\prog{n}$ denotes $n$ steps of call-by-value
evaluation; $\prog{*}$ denotes zero or more steps of call-by-value evaluation;
$[\overline{x_i |-> v_i}^i]\,t$ denotes applying the substitution $\overline{x_i |-> v_i}^i$
to term $t$; $\fname{irred}(t)$ denotes that $t$ is not reducible.

\begin{definition}[Relational Interpretations of Strategies $S$ and Strategy Contexts $\Gamma$]\label{relinterpretationstrat}
  $\calV{}{S}$ and $\calE{}{S}$ are the relational interpretations of strategies
  in terms of values and computations, respectively.
  $\calV{}{S}$ is defined as a set of triples of the form $(n,\,v,\,v')$, where $n$ is
  a natural number and $v$ and $v'$ are \textit{closed values}. $\calE{}{S}$ is defined similarly
  as a set of triples, but the second and the third elements are \textit{terms}.
  Intuitively, $(n,\,v,\,v') \in \calV{}{S}$ means that in any computation
  running for no more than $n$ steps, $v$ approximates $v'$
  \textit{under strategy $S$}; 
  $(n,\,t,\,t') \in \calE{}{S}$ means that if $t$ evaluates to a value $v$ in $m < n$ steps,
  $t'$ also evaluates to a value that is related to $v$ for the remaining $n - m$ steps.
  Data structure values are related if they share the same constructor and
  their fields are also related and function values are related if they compute related
  results from related input arguments.
  This all follows closely the logical relations approach of \citet{amal-step-index}
  that uses the step-index to stratify definitions.
  In addition, our system also reasons about approximation under
  \textit{fusion strategies} that guide the \textit{transformation} of terms,
  so terms %
  are also related if
  one can be transformed into the other.
  Notably, under the \emph{constructor fusion strategy}, %
  a data structure value approximates
  a \emph{thunked} computation that matches what \System
  imports from consumer to producer according to the term attached in the
  corresponding strategy.
  Lastly, we also define $\calG{}{\Gamma}$
  such that related strategy contexts map variables to related values.
  \begin{itemize}
    \item
    $\calV{}{\top} \defeq \left\{(n, v, v') ~\mid~ n \ge 0\right\}$
    \\
    $\calV{}{\bot} \defeq \{\}$
    \\
    $\calV{}{\mu S}
    \defeq
    \bigcup_n \calV{}{S^n(\bot)}$ where $S^n(\bot)$ denotes
    $\underbrace{S(S(\cdots S(}_{n\ \textit{times}}\hspace{-0.2em}\bot)))$
    \\
    $\calV{}{S_1 -> S_2}
    \defeq
    \begin{aligned}[t]
      \left\{\left(n,\, \lam x.t_1,\, \lam x.t_2\right) ~\mid~
      \right.&\forall m < n.\ 
      \forall v_1, v_2.\ \left(
        (m,\,v_1,\,v_2) \in \calV{}{S_1}
      \right)
      \\
      &\left.\left.
      \implies
      \left(m,\ [x |-> v_1]t_1,\ [x |-> v_2]t_2\right) \in \calE{}{S_2}
      \right)
      \right\}
    \end{aligned}$
    \\
    $\calV{}{\left\{\overline{
      c_i\ \langle\overline{S_{i,j}}^j\rangle
    }^i\right\}} \defeq
    \bigcup_i \left\{\left(
      n,\ c_i\ \overline{v_{i,j}}^j,\ c_i\ \overline{v_{i,j}'}^j
    \right) ~\mid~
    \wedge_j \forall m < n.\ (m,\,v_{i,j},\,v_{i,j}') \in \calV{}{S_{i,j}}
    \right\}
    $
    \\
    $\calV{}{\left\{\overline{c_i\ \langle\overline{x_{i,j} |-> S_{i,j}}^j\rangle -> l_i}^i\right\}}
    \defeq
    \bigcup_i
    \begin{aligned}[t]
      \left\{\left(
        n,\ 
        c_i\ \overline{v_{i,j}}^j,\right.\right.
      &\left.\left[\overline{x_{i,j} |-> v_{i,j}'}^j\right]l_i
      \right)
      ~\mid~
      \fname{FV}(l_i) \subseteq \left\{\overline{x_{i,j}}^j\right\}
      \\
      &\land\left. \forall m < n.\ \overline{(m,\ v_{i,j},\ v_{i,j}') \in \calV{}{S_{i,j}}}^j
      \right\}
    \end{aligned}
    $
    \item $
    \begin{aligned}[t]
      \calE{}{S}
      \defeq
      \Big\{(n,\ t,\ t') ~\mid~
      &\forall m < n.\ \forall v.\ \left((t \prog{m} v) \land \fname{irred}(v)\right)\\
      &\implies\left.
      \left(\exists v'.\ t' \prog{*} v' \land \left(n - m,\ v,\ v'\right) \in \calV{}{S}\right)\right\}
    \end{aligned}
    $
    \item 
    $\calG{}{\ep} \defeq \{(n,\ \ep,\ \ep)\}$\\
    $\begin{aligned}
      \calG{}{\Gamma \catsp (x |-> S)} \defeq \{(n,\ \phi_1 \catsp (x|->v),\ \phi_2 \catsp (x|->v'))
      ~\mid~
      &(n,\ \phi_1,\ \phi_2) \in \calG{}{\Gamma} \\
      &\land (n,\ v,\ v') \in \calV{}{S}\}
    \end{aligned}$
  \end{itemize}
\end{definition}

\begin{definition}[Approximation Between Terms \textbf{under Strategy}]
  $$
  \Gamma |- t \lesssim^{\fname{str}} t': S
  \ \defeq \ 
  \forall n \geq 0.\ \forall
  \left(n,\,\phi,\,\phi'\right) \in \calG{}{\Gamma}.\ (n,\ [\phi]t,\ [\phi']t') \in \calE{}{S}
  $$
\end{definition}

\begin{definition}[Typing Rules for Terms and Context]
  The typing rules for terms $\Xi |- t: I$ are entirely conventional.
  
  Contexts, represented by the meta-variable $C$, are
  essentially incomplete terms with a hole.
  Their typing judgement $\Xi |- C: (\Xi' |- I) -> I'$
  reads
  ``under the typing context $\Xi$,
  $C$ has a hole that expects a term of type $I$ under $\Xi'$,
  and when that hole is filled with a term $t$ such that $\Xi' |- t: I$,
  the resulting complete term $C[t]$ has type $I'$''.
  Context typing rules are also entirely conventional and can be easily adapted
  from previous work \cite{amal-step-index}.
  They are given
  \tr{
  in appendix (\Def{syntaxForCtx}, \Def{termtyping}, \Def{ctxtyping})
  }{in \thetr}
  due to space limitations.
\end{definition}

\begin{definition}[Relational Interpretations of Types $I$ and Typing Contexts $\Xi$]
  We also define $\mathcal{V}'$, $\mathcal{E}'$ as the
  relational interpretation of \textit{types} in terms of
  values and computations, respectively. Similarly, 
  $\mathcal{G}'$ is defined for \textit{typing} contexts
  in terms of mappings from variables to values.
  
  Note that since types $I$ are a proper subset of $S$,
  these definitions are almost identical to the relational interpretations of
  strategies $S$ and strategy contexts $\Gamma$.
  Compared to \Def{relinterpretationstrat}, 
  all strategies $S$ are replaced by some $I$
  and the definitions related to
  constructor fusion strategies $\left\{\overline{
    c\ \langle\overline{x |-> S}\rangle -> l
  }\right\}$ are removed.
  Again, the
  full definitions are in
  \tr{appendix (\Def{relinterpretationypes})}{\thet}.

\end{definition}

\begin{definition}[Approximation Between Typed Terms]
  \label{def:approx-typed-term}
  $$
  \Xi |- t \lesssim^{\fname{type}} t': I
  \ \defeq \ 
  \forall n \geq 0.\ \forall
  \left(n,\ \phi,\ \phi'\right) \in \calGeq{}{\Xi}.\ 
  (n,\ [\phi]t,\ [\phi']t') \in \calEeq{}{I}
  $$
\end{definition}

\begin{definition}[Contextual Approximation]\label{def:ctxapprox}
  We use the usual definition of contextual approximation for possibly
  non-terminating programs, which says that $t_2$ contextually approximates
  $t_1$ if for all contexts $C$ that respect the types of $t_1$ and $t_2$,
  $C[t_2]$ terminates as long as $C[t_1]$ terminates.
  
  The intuition behind this
  is that two programs are considered equal if there is no way of telling them apart
  --- indeed, if there was a way of telling them apart,
  that could be used in a context that chooses to halt or loop forever as a result.
  Therefore, termination preservation in \textit{arbitrary} contexts implies that
  the results are preserved as well.
  \begin{align*}
    \Xi |- t_1 \lesssim^{\fname{ctx}} t_2: I
    \ \defeq \ 
    \forall C, I'.\ \ep |- C : (\Xi |- I) -> I' \land C[t_1]\Downarrow \implies C[t_2]\Downarrow
  \end{align*}

  We also define what
  it means for two typed \textit{contexts} to be logically related:
  \begin{align*}
    \Xi' |- C \lesssim C': (\Xi |- I) -> I'
    \ \defeq \ 
    \forall t, t'.\ \Xi |- t \lesssim^{\textcolor{black}{\fname{type}}} t': I
    \implies
    \Xi' |- C[t] \lesssim^{\fname{type}} C'[t']: I'
  \end{align*}
\end{definition}

The main soundness theorems are presented below.
\Thm{fundamentalterm} shows that the
relational interpretation of fusion strategies
is indeed compatible with the behavior of the fusion rules.
Then \Thm{soundnessoftransformation} concludes that
for identity fusion strategies, which
preserve the type and results when \System
uses them to transform top-level terms, 
their relational interpretations
respect the standard definition of contextual approximation,
which means that our transformation is sound.

At a high level, the proof goes as follows:
if term $t$ is transformed into $t'$ under an identity
strategy $I$, then $t$ approximates $t'$ under the strategy $I$;
if $t$ approximates $t'$ under the strategy $I$,
then \Def{def:approx-typed-term} holds between $t$ and $t'$.
And lastly, \Def{def:approx-typed-term} implies contextual approximation in \Def{def:ctxapprox}.
The detailed proof is presented in \App{sec:partial-proofs} and \App{app:proofs}.

\begin{theorem}[Fundamental Property for Terms]\label{fundamentalterm}
  $\Gamma |- t ~> t' <= S \implies \Gamma |- t \lesssim^{\fname{str}} t': S$
\end{theorem}

\begin{theorem}[Soundness of Transformation]\label{soundnessoftransformation}
  Our transformation is sound
  (see \Def{def:ctxapprox}):
  $$\Xi |- t ~> t' <= I \implies \Xi |- t \lesssim^{\fname{ctx}} t': I$$
\end{theorem}

\section{Inference of Fusion Strategies}
\label{sec:inference}

We now formalize the main algorithmic aspects of \System,
which consist of three main phases:
\begin{enumerate}
  \item \emph{Subtype inference}:
    Traverse the source program and
    infer a corresponding set of subtyping constraints
    between \emph{fusion strategy candidates} (\aka \emph{inferred types}).
  \item \emph{Constraint solving and propagation}:
    Check the satisfiability of the inferred constraints
    by propagating the subtyping %
    relations
    through type variable bounds and type constructors,
    making sure all inferred lower bounds of a given type variable
    are consistent with all inferred upper bounds of the same type variable.
  \item \emph{Unification}:
    Solidify the propagated constraints into a substitution $\phi$
    that can be applied to turn any inferred fusion strategy candidate
    into a proper 
    fusion strategy, used to reconstruct
    a typing derivation in the declarative type-and-elaboration system
    of \Sec{sec:formal}.
\end{enumerate}

\subsection{Syntax}

\begin{figure}
{\small
\begin{minipage}{.4\linewidth}
\begin{align*}
\ty\Pos,\tz\Pos &\ ~::=~ \ 
  \al
  ~\mid~
  \al -> \al
  ~\mid~
  c\,\langle\,\overline{\al}\,\rangle
\\
\ty\Neg,\tz\Neg &\ ~::=~ \ 
  \al
  ~\mid~
  \al -> \al
  ~\mid~
  \left\{\,\overline{
    c\,\langle\,\overline{x |-> \al}\,\rangle -> l
  }\,\right\}: \al
\\
\ty%
  &\ ~::=~ \ \ty\Pos ~\mid~ \ty\Neg
\\
\phi &\ ~::=~ \ 
  [\,\overline{\al|->S}\,]
\end{align*}
\end{minipage}
\qquad
\begin{minipage}{.4\linewidth}
\begin{align*}
&\textit{Contexts:}
\\
\Gamma &\ ~::=~ \ 
  \ep
  ~\mid~
  \Gamma \cdot (x|->\al)
\\
\SCtx,\SDtx &\ ~::=~ \ 
  \ep
  ~\mid~
  \ty\Pos =<\, \ty\Neg
\\
\CCtx,\CDtx &\ ~::=~ \ 
  \ep
  ~\mid~
  \ty\Pos =<\, \al
  ~\mid~
  \al =<\, \ty\Neg
\end{align*}
\end{minipage}
}
\caption{Syntax of the fusion strategy inference system.}
\label{fig:syntax-inference}
\end{figure}

The syntax of %
fusion strategy inference %
is given in \Fig{fig:syntax-inference}.
It resembles that of
the declarative type-and-elaboration system of \Sec{sec:formal}
except that it includes type variables $\al$, is \emph{polarized} into
distinct \emph{positive} and \emph{negative} types,
and is slightly more refined/specific
to accommodate the algorithmic nature of the inference process.
We also define 
several context forms.
$\SCtx$ denotes a set of arbitrary constraints
and $\CCtx$ denotes a set of bounds on type variables.
We use $\tv{\ty}$ to represent the statement that
$\ty$ is a type variable and $\ntv{\ty}$ to represent
the statement that $\ty$ is \emph{not} a type variable.

\subsection{Algorithmic Typing Rules}

Recording static information on how
producers flow into consumers is a core part of our strategy
inference algorithm.
In the the relatively recent \emph{MLsub},
and following a long tradition
\cite{Pottier98:thesis,rehof_type-base_2001,Palsberg-popl98:from-polyvar-anal-flow-info-to-inter-union-tys,Smith00:polyvar-flow-anal-constr-tys,Faxen97:polyvar-poly-flow-anal},
\citet{dolan_polymorphism_2017}
proposed that subtyping can
serve as a powerful tool to get an approximation of the
directed flow of data in a higher-order program
(a problem known as \emph{control-flow analysis}).
\emph{Simple-sub}
\cite{Parreaux20:simple-essence-alg-subt} further distilled
the essence of MLsub 
by introducing a simplified %
algorithm.
In this paper, we adapt Simple-sub's idea
of using subtyping with bounds to track the data flows from producers to
consumers to %
infer \emph{fusion strategy candidates}.
Instead of doing ``biunification'' as in MLsub, Simple-sub addresses the
problem of type inference by collecting subtyping
constraints and solving them on the fly.
Our approach follows the same pattern except that constraint solving
is done separately.

\begin{figure*}
\mprset{sep=1.2em}
{\small\input{contents/fusion-inference-rules.tex}}
\caption{Fusion strategy candidate inference rules.}
\label{fig:fusion-inference}
\end{figure*}

\Fig{fig:fusion-inference}
presents our constraint
collection process as subtype inference rules.
During this process, each term is assigned
\emph{a type variable}
(which also serves as a unique identifier)
and subtyping
constraints are produced according to the shape of the term
and accumulated as part of a $\SCtx$ output for later constraint solving.
Again, these rules are used to accumulate constraints without solving them;
on their own, they only check for the well-scopedness of the input program:
for any term $t$,
there exists a type variable $\al$ and $\SCtx$
such that $\ep |- t : \al \conR \SCtx$
if and only if $t$ is closed.

In particular, %
\textsc{I-Ctor} generates a fresh type variable $\be$
and assign it as the type of the term, then constrains the actual
constructor type to be a subtype of it;
\textsc{I-Lam} also generates a type variable $\ga$ with a function type as its lower bound;
\textsc{I-Case} generates new type variables for identifiers introduced in patterns,
traverses the scrutinee and each branch and constrains that
the types of all branches are subtypes of
the type of representing the result of this $\kw{case}$ term; \textsc{I-LetRec} generates a fresh type variable $\al$
for the let-bound variable, traverses the right-hand-side and body with the updated
context, and constrains the type of the right-hand-side to be a subtype
of $\al$ to tie the recursive knot.
Lastly, \textsc{I-Var} and \textsc{I-App} are 
standard.

\subsection{Constraint Solving}

The constraint-solving
rules are presented in \Fig{fig:constraint-solving}.

Judgment $\CCtx |> \SCtx$
reads ``$\SCtx$ is solved under bounds context $\CCtx$''\!.
Every constraint-solving derivation ends with a single application of the leaf rule
\ruleName{C-Done} and thus contains exactly one premise of the form $\textbf{\textit{output }} \CCtx'$. This $\CCtx'$ is considered to be the output of derivation,
which we write ``$\CCtx |> \SCtx$ outputs $\CCtx'$''\!.
Constraint solving uses a %
\emph{worklist algorithm},
where the worklist is the input constraints context $\SCtx$,
which is updated by the rules
and grow more constraints as the original constraints are decomposed
and as type variable bounds are checked.
\begin{figure*}
\mprset{sep=1.2em}
{\small\input{contents/constraint-solving-rules.tex}}
\vspace{.5em}
\caption{Constraint solving and unification rules.
}
\label{fig:constraint-solving}
\end{figure*}
We write 
$\CCtx\Pos$ to denote contexts of only \emph{lower}-bounds
$\ty\Pos =<\, \al$,
$\CCtx\Neg$ to denote contexts of only \emph{upper}-bounds
$\al =<\, \ty\Neg$,
where $\ty\Pos$ and $\ty\Neg$ are not type variables, 
and $\CCtx^v$ to denote contexts of only \emph{variable-to-variable} bounds
$\al =<\, \be$.
An invariant of constraint solving is that 
in the current bounds context $\CCtx$,
for all type variable $\al$,
every lower bound $\ty\Pos =< \al$ %
that occurs on the \emph{right}
of an upper bound $\al =< \ty\Neg$ %
is assumed to already be transitively checked for consistency,
\ie an outer derivation has already placed $\ty\Pos =< \ty\Neg$ on the working list.
Conversely, if 
$\ty\Pos =< \al$ occurs on the \emph{left} of $\al =< \ty\Neg$,
these bounds still need to be checked later.

Rule 
\ruleName{C-Done} asserts that constraint solving is done when
the working list is empty
\emph{and} the current bounds context $\CCtx$
is ordered so that all upper bounds $\CCtx\Neg$
come before all lower bounds $\CCtx\Pos$,
with variable bounds $\CCtx^v$ in the middle.
\ruleName{C-VarL}
and \ruleName{C-VarR}
insert type variable constraints into %
$\CCtx$ so that they will need to commute across all existing bounds in $\CCtx$,
checking them for consistency.
\ruleName{C-Pass} commutes unrelated constraints occurring in the wrong order
while \ruleName{C-Match} ensures two opposite bounds on the same type variable
are consistent while commuting them.
\ruleName{C-Skip}
ignores constraints that have already been registered in $\CCtx$;
it is important for termination in the presence of cyclic bounds. 
\ruleName{C-Fun} and \ruleName{C-Ctor}
respectively constrain function types and constructors.
Finally, there are three ``merge'' rules: \ruleName{C-VarMrg}, \ruleName{C-FunMrg}
and \ruleName{C-CtorMrg}. They ensure that
for all $(\al =< \be) \in \CCtx$, $\al$ and $\be$ always share the
\emph{same} set of concrete upper bounds and that the types of the fields
of these upper bounds are 
all merged into the equivalent type variables.
This is needed because later
we will \emph{unify} away
all type variables
with their upper bounds
(recall that the declarative system does not have type variables),
so type variables related by subtyping must be constrained to be equivalent to each other
and upper bounds on the same type variable must be compatible.

\label{constr-solving-exhaustive}
The constraint-solving algorithm tries constraint-solving rules one by one
in the order they are presented in \Fig{fig:constraint-solving}:
a rule is considered applicable only if none of the rules before it is applicable.
This ensures that type variable bounds
from the worklist $\SCtx$ are placed in the correct order in $\CCtx$
and that by the time \ruleName{C-Done} is used,
the ``merge'' rules have been applied exhaustively.
In the absence of a matching rule for a given input of $\CCtx$ and $\SCtx$,
constraint solving fails.

\subsection{Unification}

The unification rules are also presented in \Fig{fig:constraint-solving}.
We write the identity substitution as $[]$
and the usual substitution composition operator as $\circ$.
The goal of unification is to substitute all inferred type variable placeholders
with proper declarative fusion strategies $S$.
We do this by simply unifying all the concrete upper bounds inferred for each type variable
using the $\fname{unif}$ function.
When a type variable has no inferred concrete upper bound
or does not appear in the domain of the substitution,
that means it corresponds to some data that is never consumed,
so we simply substitute it with $\top$.
We use the shorthand $\left\{\,\overline{
  c_i\,\langle\,\overline{x_{i,j} |-> \al_{i,j}}^j\,\rangle
  \;(-> l_1)
}^{i}\,\right\}\left(:\be\right)$
to denote either the identity strategy
$\left\{\,\overline{
  c_i\,\langle\,\overline{\al_{i,j}}^j\,\rangle
}^{i}\,\right\}$
or the fusion strategy
$\left\{\,\overline{
  c_i\,\langle\,\overline{x_{i,j} |-> \al_{i,j}}^j\,\rangle
  -> l_1
}^{i}\,\right\}:\be$
in the second-last equation:
if a producer has multiple upper bounds (multiple consumers),
it is given an identity strategy and is prevented from fusing
because it is impossible to import computations from \textit{both} consumers
into the same producer.
Note that
$\fname{unif}$ only has to pick one of the concrete upper
bounds of a type variable thanks to the ``merge'' constraint-solving rules
explained above, which ensure that all the field types of the concrete upper bounds
of a type variable are 
made equivalent.

Notice that the syntax of strategies that unification yields
does not exactly match that of the declarative system
presented in \Sec{sec:formal-syntax}:
the difference is that we use fusion types of the form
$\left\{\,\overline{
  c_i\,\langle\,\overline{x_{i,j} |-> S_{i,j}}^j\,\rangle
  -> l_i
}^{i}\,\right\}:S$
instead of just
$\left\{\,\overline{
  c_i\,\langle\,\overline{x_{i,j} |-> S_{i,j}}^j\,\rangle
  -> l_i
}^{i}\,\right\}$.
This syntax corresponds to an \emph{intermediate step} on the way to
producing strategies in the declarative system of  \Sec{sec:formal}.
We call this intermediate step the \emph{pseudo-algorithmic} fusion rules,
whose syntax and rules are presented in \App{app:pseudo-algorithmic},
which also explains the subtle but inessential difference between the two syntaxes.

For simplicity of the presentation,
in this paper we assume that constructors are always produced and matched consistently;
that is, a given data type always uses the same constructors,
and this exact same set of constructors is always matched in the corresponding
pattern matching expressions.
While this assumption greatly simplifies the definition of $\fname{unif}$,
and while it generally holds true of ML-family languages like OCaml and Haskell,
it is \emph{not} necessary in practice,
and our actual implementation can handle cases where arbitrary
mixes of constructors are produced and matched.

\begin{example}
Consider the unification of cyclic context
$\,\CCtx = (\al =< \ty_1 -> \be)\catsp(\be =< \ty_2 -> \al)$.
The unification algorithm proceeds as follows:
\smallskip
\begin{mathpar}
\inferrule{
  \inferrule{
    \ep \conR []
    \\
    S_\be = \mu(\lam\be.\ \ty_2 -> \al)
  }{
    (\be =< \ty_2 -> \al) \conR
    [\be|->S_\be] \circ [] = [\be|->S_\be]
  }
  \\
  S_\al = \mu(\lam\al.\ [\be|->S_\be](\ty_1 -> \be))
}{
  \CCtx \conR
  [\al|->S_\al] \circ [\be|->S_\be]
  = [\al|->S_\al,\,\be|->\mu(\lam\be.\ \ty_2 -> S_\al)]
}
\end{mathpar}

\smallskip
After expanding the $S_\al$ and $S_\be$ definitions,
the final inferred substitution is
$\CCtx \conR \phi =
  [
  \al|->\mu(\lam\al.\ \ty_1 -> \mu(\lam\be.\ \ty_2 -> \al))
  ,\,
  \be|->\mu(\lam\be.\ \ty_2 -> \mu(\lam\al.\ \ty_1 -> \mu(\lam\be.\ \ty_2 -> \al)))
  ]
  $.
\end{example}

Naturally, in the implementation,
we preserve the sharing of strategies $S$
so as to avoid the size blowup that would be incurred by substituting their
definitions as in $\phi$ above.

\subsection{Correctness}

The correctness of our fusion inference system hinges on the following theorem,
which uses the pseudo-algorithmic fusion rules of \App{app:pseudo-algorithmic}
(\Fig{fig:fusion-algo})
denoted by %
$~>_{\mathcal{A}}$.

\begin{theorem}
\label{thm:algorithm-correctness}
  If $\ep |- t : \al \conR \SCtx$,
  $\ep |> \SCtx$ outputs $\CCtx$
  and $\CCtx \conR \phi$
  then
  $\ep |- t ~>_{\mathcal{A}} t' <= \phi(\al)$
  for some $t'$.
\end{theorem}

To prove this, we need to introduce some important intermediate definitions.

\begin{definition}[Self-contained constraints]
\label{def:self-contained}
  A set of constraints $\SCtx$ is said to be \emph{self-contained} if
  for all $(\ty\Pos =< \ty\Neg) \in \SCtx$ where
  $\ty\Neg = \left\{\overline{c_i\,\langle\,\overline{x_{i,j}|->\al_{i,j}}^j\,\rangle->l_i}^i\right\}:\be$,
  we have $\overline{\overline{\left(x_{i,j}:\al_{i,j}\right)}^j |- l_i: \ga_i \conR \SDtx_i}^i$
  for some $\overline{\SDtx_i}^i$ and $\overline{\ga_i}^i$,
  and $\overline{\SDtx_i \catsp (\ga_i =< \be)}^i \subseteq \SCtx$.
\end{definition}

\begin{definition}[Equivalence under unrolling]
  We define the \emph{equivalence modulo recursive types} relation
  $(==)$
  as the smallest relation satisfying the following rules:
\mprset{sep=1.5em}
{\small
\begin{mathpar}
\inferrule{}{
  S == S
}

\inferrule{
  S(\mu S) == S'
}{
  \mu S == S'
}

\inferrule{
  S == \mu S'
}{
  S == S'(\mu S')
}
\end{mathpar}}
\end{definition}

\begin{definition}[Consistency]
  We say that $\phi$ is \emph{consistent}
  with constraint $\ty\Pos =< \ty\Neg$
  when:
    \begin{itemize}
      \item If
        $\ty\Pos = \ga$ and $\ty\Neg = \left\{\overline{c_i\,\langle\,\overline{x_{i,j}|->\al_{i,j}}^j\,\rangle->l_i}^i\right\}:\be$
        then
        $\phi(\ty\Pos) ==
          \left\{\overline{c_i\,\langle\,\overline{S_{i,j}}^j\,\rangle}^i\right\}$

        where $\overline{S_{i,j} == \phi(\al_{i,j})}^{i,j}$
        or
        $\phi(\ty\Pos) ==
          \left\{\overline{c_i\,\langle\,\overline{x_{i,j}|->S_{i,j}}^j\,\rangle->l_i}^i\right\}:S'$
        where
        
        $\overline{S_{i,j} == \phi(\al_{i,j})}^{i,j}$, $S' == \phi(\be)$.
      
      \item If $\ty\Pos = \ga$ and $\ty\Neg = \al -> \be$, then $\phi(\ty\Pos) == S_1 -> S_2$
      where $S_1 == \phi(\al)$ and $S_2 == \phi(\be)$.

      \item If $\ty\Pos = c_n\,\langle\,\overline{\al_{n,j}}^j\,\rangle$,
        then
        $\phi(\ty\Neg) == \top$
        or
        $\phi(\ty\Neg) ==
          \left\{\overline{c_i\,\langle\,\overline{S_{i,j}}^j\,\rangle}^i\right\}$
          where $\overline{S_{n,j} == \phi(\al_{n,j})}^{j}$
        or

        $\phi(\ty\Neg) ==
          \left\{\overline{c_i\,\langle\overline{x_{i,j}|->S_{i,j}}^j\rangle->l_i}^i\right\}:S'$
        where $\overline{S_{n,j} == \phi(\al_{n,j})}^j$.
        
      \item If $\ty\Pos = \al -> \be$, then $\phi(\ty\Neg) == \top$
      or $\phi(\ty\Neg) == S_1 -> S_2$ where $S_1 == \phi(\al)$ and $S_2 == \phi(\be)$.
      
      \item Otherwise,
        $\phi(\ty\Pos) == \phi(\ty\Neg)$ or $\phi(\ty\Neg) == \top$.
    \end{itemize}
\end{definition}

We also lift the notion of consistency to whole constraint contexts $\SCtx$.
\Thm{thm:algorithm-correctness} is then proved by combining the following two core lemmas,
whose proofs are sketched in \tr{\App{app:inference-proofs}}{\thet}:

\begin{lemma}
\label{lem:constr-implies-cons}
  If
  $\ep |> \SCtx$ outputs $\CCtx$
  and $\CCtx \conR \phi$
  then $\phi$ is consistent with $\SCtx$.
\end{lemma}

\begin{lemma}
\label{lem:cons-implies-wt}
  If
  $\Gamma |- t : \al \conR \SCtx$
  and $\phi$ is consistent with self-contained $\SCtx\catsp\SDtx$
  then
  there exists a $t'$ such that
  $\phi(\Gamma) |- t ~>_{\mathcal{A}} t' <= \phi(\al)$.
\end{lemma}

We then show that
if $\ep |- t ~>_{\mathcal{A}} t' <= \phi(\al)$ then $\ep |- t ~> t' <= S$ for some $S$
in the original declarative system of \Sec{sec:formal}.
The idea for obtaining this mapping is presented in
\App{sec:pseudo-algo-to-prescient}.

\subsection{Other Implementation Considerations}

We conclude this section with %
some considerations
that are important for practical implementations.

\mypara{Preprocessing for Well-scopedness and Thunking}
\label{sub-sec:thunking-free-vars}
We have been silent on a particular subtlety of the transformation process:
notice that in \ruleName{F-Ctor} and \ruleName{F-Case},
there are \emph{free-variable} side conditions.
These ensure that the transformation process remains well-scoped,
as explained in \Sec{sec:computations-fvs}.
Yet, the algorithm described here does not check for well-scopedness.
Indeed, we assume that a \emph{pre-processing}
pass is performed on the program to ensure that \emph{all} 
@case@ branches are \emph{thunks}
(as in \Sec{sec:preserve-cbv})
that are \emph{closed under variables bound in their matching patterns}
(as in \Sec{sec:computations-fvs}).
This is simply done by making
every @case@ expression return functions of the free variables
of its branches and
immediately applying
the result of the @case@ with the corresponding free variables as arguments.
The formal transformation rule is presented in \Fig{fig:thunking} of \App{sec:scope-extr}.
This transformation is easy to revert after the fusion transformation
in case it was not useful (\ie in case that pattern matching expression was not fused).
In the actual implementation of \System, though,
we only do the thunking transformation on the fly as the need for it arises
during fusion.

\mypara{Duplication to Resolve Conflict}
\label{sub-sec:duplication-to-resolve-conflict}
Because each term in the source program is associated with only one
unique identifier (in the form of a type variable)
to which a single fusion strategy is attached,
the same producer may be matched
with multiple consumers across different call sites,
resulting in spurious clashes (lack of polymorphism).
Therefore, we implemented a simple process that duplicates
top-level functions 
for each of their external call sites.
For (mutually) recursive top-level functions, this process
stops as soon as a recursive call is reached to tie the knot.
This is essentially an ad-hoc implementation of %
ML-style \emph{let polymorphism} for top-level binders.
As we will see later in the %
experimental evaluation section, this rather simple and naive
expansion does not cause serious problems
for the majority of the programs we tested.
In the future, we want to investigate %
ways of duplicating
and also unrolling definitions %
only when needed by fusion.
\label{ref:better-def-dup}

\section{Experimental Evaluation}
\label{sec:bench}

We use the OCaml compiler as our backend
since: (1) it uses call-by-value semantics, which is our focus;
and (2) we can make use of its
powerful type system %
with support for recursive types
and polymorphic variants \cite{garrigue1998programming}, which are useful since
our transformation may produce programs that require
recursive types 
and may require changing the types of constructor fields.

We adapted to OCaml the standard \emph{nofib} benchmark suite \cite{Partain93:nofib},
which was originally written for Haskell.
The main reason for choosing this over existing benchmark suites for OCaml,
such as \emph{sandmark},\!\footnote{\url{https://github.com/ocaml-bench/sandmark}}
is that the latter tend to use
imperative or object-oriented features, 
while we only support purely functional programs.
There are three subsets of benchmark programs in
\emph{nofib}, namely: ``imaginary'' for toy example programs,
``spectral'' for more realistic ones,
and ``real'' for actual applications.
We tested \System on the ``spectral'' subsets.
8 benchmarks could not be ported faithfully due to their
heavy reliance on laziness or extensive use of features
like arrays, so at this point, we have not yet evaluated \System on them.
Our benchmarking %
workflow consisted of three steps.
First, for each benchmarked Haskell program, we performed
some manual adaptions, making sure the gist of the program was faithfully %
translated to call-by-value semantics~---~we notably made sure to avoid
producing unnecessary data structures, using opt-in laziness to reflect
the original program's intention.
Second, we automatically converted the ASTs of those Haskell programs to our core
language and performed deforestation.
Finally, we pretty-printed
the optimized core language ASTs into OCaml programs
and used the \textit{core\_bench}\footnote{\url{https://github.com/janestreet/core_bench}} benchmarking library for OCaml to run the code repeatedly
and collect statistical performance data.

\begin{figure}
  \Description{run time} %
  \includegraphics[width=\textwidth]{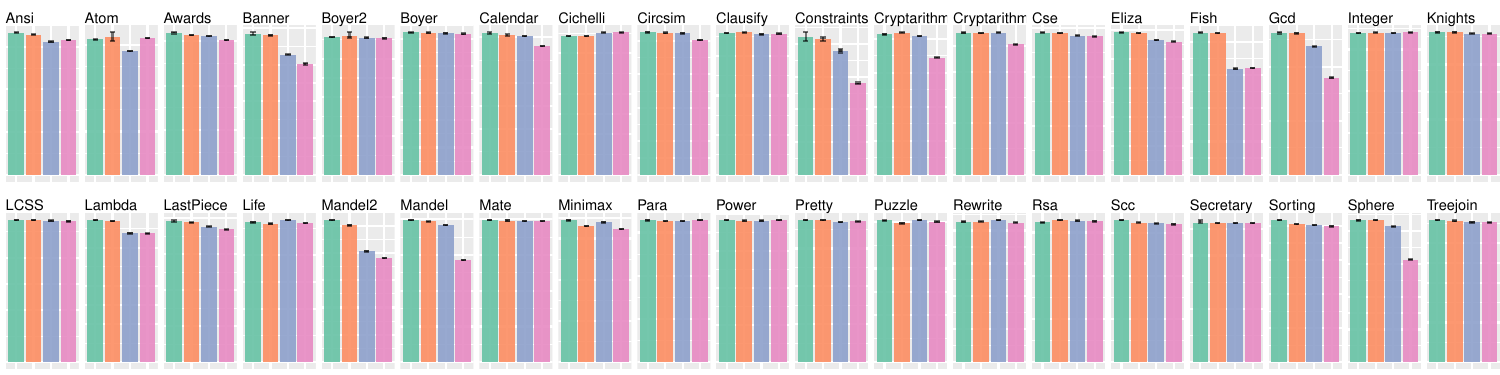}
  \centering
  \caption{Execution times of programs compiled by the 
    \emph{ocamlopt} compiler with flags \emph{-O3} and \emph{-rectypes}.
    The first (green) bar represents the original program;
    the second (orange) bar is the program after only definition duplication;
    the third (blue) bar is the program immediately after producer/consumer rewriting;
    the fourth (pink) bar is the program after producer/consumer rewriting \emph{and}
    floating out of functions.
  }
  \label{fig:runtime}
\end{figure}

\begin{figure}
  \Description{} %
  \includegraphics[width=\textwidth]{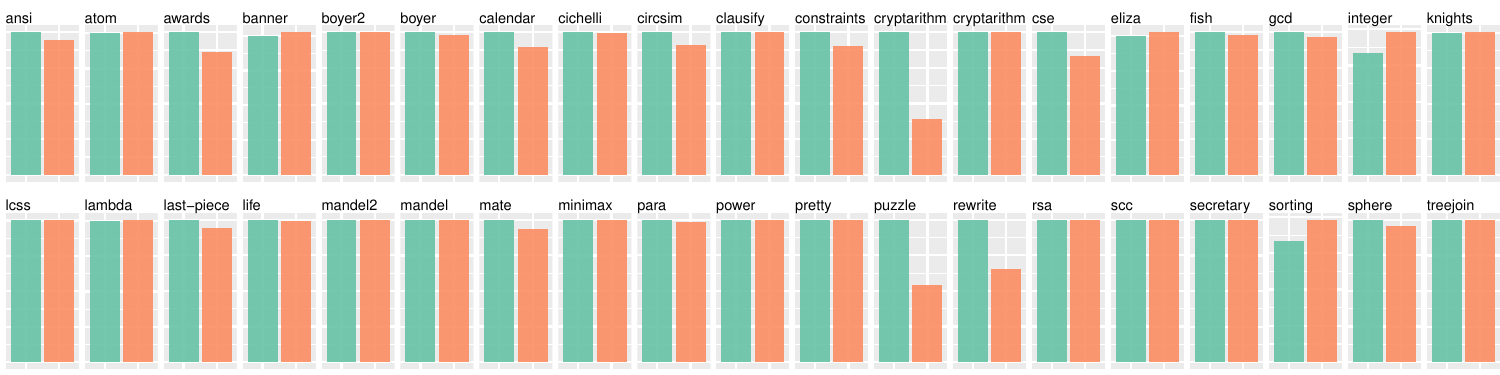}
  \centering
  \caption{Execution times of programs compiled by \emph{GHC}. The first (green) bar represents the
  program compiled 
  by the GHC without fusion;
  the second (orange) bar represents the program compiled by GHC with fusion.}
  \label{fig:hs-time}
\end{figure}

\begin{figure}
  \Description{code size} %
  \includegraphics[width=\textwidth]{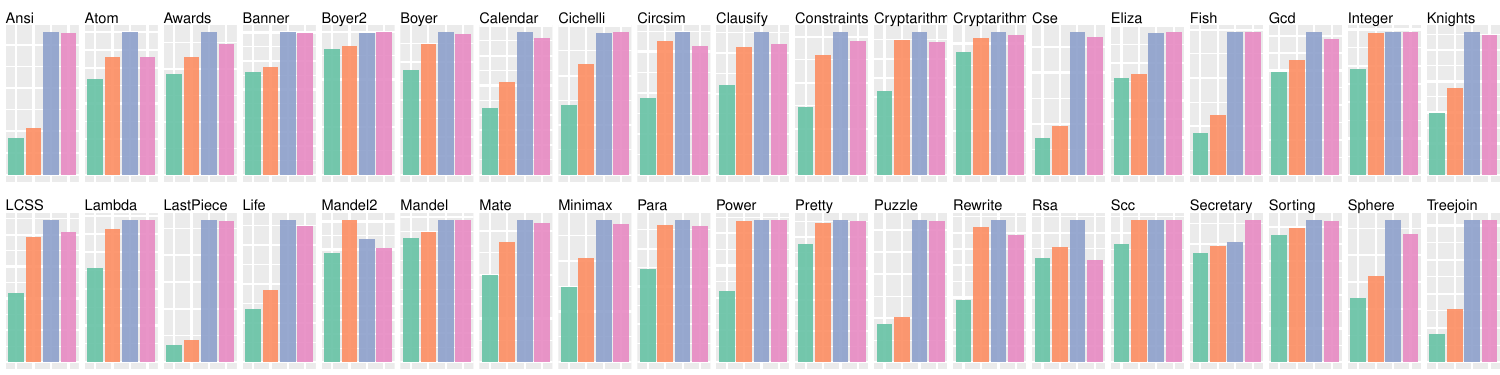}
  \centering
  \caption{Binary code sizes corresponding to the programs benchmarked in \Fig{fig:runtime}.}
  \label{fig:codesize}
\end{figure}

\subsection{Manual Adaptation of Haskell Programs}

Inevitably, some adaptions needed to be made manually to the original
Haskell programs %
to make them run in OCaml.
Our adaptations did \textit{not} aim to fully recover {every}
implementation detail of the original benchmarked programs,
but rather to %
faithfully transcribe the core logic of the program %
so as to reflect the complexity of programs users may write in
real-world applications.
Importantly, our manual
adaptions were done \emph{without deforestation
in mind},
to prevent any bias or unfairness.

The main adaptions we made
are the following:

\begin{itemize}[topsep=3pt]
  \item Appropriate calls to @lazy@ and @force@ %
  from
  OCaml's built-in @Lazy@ module
  were added to the programs.
  We carefully inspected the Haskell source programs
  to determine where to add these annotations,
  which was mostly in places where infinite large or lazy
  data structures were defined 
  but \textit{intentionally}
  consumed lazily and sometimes 
  only partially.
  As an example, a common pattern 
  in Haskell is @takeWhile@ @f@ @ls@, where @ls@
  may be a very large or infinite list, and @f@ is a function
  that determines properties such as when the list converges.
  If @ls@ was large but finite,
  the program would still terminate without laziness annotations
  (and deforestation would %
  yield huge speedups),
  but omitting these annotations would go against
  the intent of the original program,
  making such benchmark result 
  inaccurate.
  Moreover, simple patterns like @take@ @n@ @[1..]@,
  which are also common in Haskell, were adapted to
  patterns like @[1..n]@ to prevent unnecessary construction of lazy
  computations in OCaml. The reason is similar:
    a normal OCaml programmer would not write
    such patterns in real-world scenarios, so leaving them
    in would create spurious optimization opportunities that would bias the results.
  
  \item 
  Common built-in functions like @map@ or
  @filter@ were manually implemented in direct recursive
  style without the @build@/@foldr@ combinators, which Haskell
  uses to enable its shortcut fusion.

  \item Methods defined in type class instances and
  user-defined operators were converted manually to top-level
  function definitions to facilitate 
  conversion.

  \item Some original benchmarked programs consisted of several Haskell
  source files. These were combined manually since
  \System currently only supports single-file compilation.
  We expect our approach to easily
  extend to multiple source files,
  and even to separate compilation
  since any part of the programs can be marked as unknown
  using the $\top$ strategy.

  \item We generally adopted the inputs that came with most tests in the nofib benchmarks.
  However, some of the input data just specified the number of iterations @main@ should repeat
  to run the same code -- those inputs were omitted since we rely on
  \textit{core\_bench} to do the repetition.
\end{itemize}

\subsection{Automatic Conversion From Haskell Programs}
\label{subsec:conversion-from-hs}

After manual adaption, we automatically convert the Haskell
program ASTs to our internal core language.
Since that language
is rather minimal, most of the work is in
desugaring list comprehension and handling nested patterns.
Some special care is also taken for large string literal values.

List comprehensions are desugared using the efficient method
introduced in the textbook by \citet{listcomp}, which does not use
@flatMap@ or @append@ and only performs the minimum number of
list constructor operations. Although there would be more
deforestation optimization opportunities if @flatMap@ or
@append@ were used, we believe that using the efficient method
makes our benchmark results closer to what can be
expected in real use cases, since the efficient desugaring is
more likely to be adopted in real compiler implementations.

Nested patterns are handled following the same textbook \cite{listcomp}.
No further aggressive
optimizations are done for handling pattern matchings so
there may be inefficiencies like scrutinees being examined
more than once. This is less ideal in terms of approximating
real use cases of our algorithm, but we argue that this will
not substantially invalidate our experiment results because this issue
will not create spurious optimization opportunities
--- instead, it is more likely to cause
difficulties for \System since a scrutinee
being examined twice always means a potential strategy clash
and a loss of fusion opportunities.
For example, the definition
@f Nothing Nothing = 1; f (Just x) (Just y) = x@ will
be transformed to a definition @f a b = ...@, where @b@ is examined
by two pattern-matching terms, one in the branch where @a@ is @Nothing@,
and another in the other branch.

Some benchmarked programs contain large string
literals that are either read from a file at
runtime or hard-coded in the source code. Since
strings are lists of characters, large constants may cause
problems in both compilation time and output code size if
fusion is performed on them. We currently
mark these as special large string constants to avoid fusion
being performed on them.

\subsection{Analysis of Transformation Performance}

We expect \System to have reasonable running times
and to scale to practical use cases.
Our as-of-yet unoptimized implementation
takes less than 2 minutes %
to deforest our 38 benchmarked programs totalling 5000 lines of code.
This notably includes the time taken for parsing,
definition duplication, type inference, constraint solving, unification and elaboration.
We find this encouraging
because it is much better than previous general deforestation techniques
and because there are still many opportunities for optimizing our %
implementation.
\System's data flow information gathering and constraint solving algorithm
are also very simple at their core,
and quite similar to context-insensitive control-flow analysis systems (0CFA)
that have been studied since the 1990s and deployed in many real-world products \cite{midtgaard-cfa}.
In fact, languages like OCaml and Haskell use much more powerful type inference and
constraint-solving approaches that %
are still quite practical for real programs,
so it is reasonable to expect that \System will be as well.

\subsection{Analysis of Results}
\label{sub-sec:analysis-of-res}

\Fig{fig:runtime} shows the execution times of the
benchmarks in a group of bar plots with 95\% confidence
interval error bars. The y-axis stands for the execution
time and the bars 
represent the results described in the figure's caption.
Note that the actual execution time per run
varies from nano-seconds to several seconds, so we omit the
units and actual values of the y-axes here and only show the
results on a relative scale. Detailed tables containing
the exact numbers can be found in \tr{the \App{app:exp-res}}{\thetr}.

\System leads to a number of performance
improvements in the benchmarked programs.
We found that while many
fusion opportunities come from standard list
operations like @map@ and @[a..b]@, like in
@Calendar@, the fusion of
custom-defined recursive data structures (that cannot be fused by short-cut fusion)
like @MandTree@, defined in @Mandel2@, also
contributes greatly to the improvements.
More explanations of the sources of improvements
in other test cases are given in \App{app:benchmark-explanations}.

\Fig{fig:codesize} shows the code size results.
As in \Fig{fig:runtime}, the values of the y-axis are omitted,
using a relative scale.
Overall, 68.4\% of the optimized programs have a code size that is less
than 2x of the sizes of their original programs.
On average (using the geometric mean to average ratios),
the code size grows by only about 1.79x.
The main reasons for the growth 
are
  (1) the duplication of top-level definition; and
  (2) the duplication of computations from a single consumer to multiple producers.
Some programs grow considerably in size without much improvement in the runtime,
which is caused by unproductive duplication of top-level definition.
These sources of duplication could be mitigated:
unproductive duplication of top-level definitions could
be %
avoided by pre-analyzing the detected fusion opportunities
and only duplicating as needed (as suggested in \Sec{ref:better-def-dup})
and we could insert new definitions to share the code of those branches
that are duplicated into multiple consumers.
The latter should be very straightforward to implement,
but we have not gotten around to it yet.
Moreover, by combining the
results from both figures, we can see from cases like @Mandel2@ or @Gcd@
that it is not necessary for the code size to grow 
significantly to achieve %
good speedups.

Our results contain some slightly pathological or suboptimal cases.
The binary sizes of some benchmarks like @LastPiece@, @Treejoin@, @Ansi@, @Cse@, @Puzzle@ and
@Fish@ more than triple after \System's transformation.
Although the size growth comes with some runtime improvements,
this may not always be a good tradeoff.
Note that while large \textit{string} literals
are handled specially to avoid the problems mentioned in \Sec{subsec:conversion-from-hs},
some large \textit{list} literals are still fused, leading to long chains of let bindings
duplicating matching branch bodies in the resulting programs.
This could be alleviated by the use of finer heuristics
that analyze and disable the fusion of large literals uniformly.

\section{Related Work}

We now discuss related work and 
compare \System with other fusion methods.

\subsection{Deforestation to Distillation}

\mypara{Deforestation and Supercompilation}
\emph{Deforestation} was proposed by \citet{wadler-tcs90:deforestation}
to refer to the optimization aiming at eliminating the allocation of intermediate data structures.
\citeauthor{wadler-tcs90:deforestation}'s original deforestation approach follows the \emph{unfold/fold} framework by \cite{burstall-transformation},
where \emph{unfolding} inlines functions and partially evaluates programs
to eliminate intermediate data structures and uncover more deforestation opportunities,
and \emph{folding} generates new function definitions to tie recursive knots.
The original deforestation method only works on a subset of \emph{first-order} programs,
which severely limits the practical use of this technique.
There were some extensions of it
\cite{hamilton-deforestation-first-order, chin-deforestation, Chin-1994-safe-deforestation, marlow-higher-order-deforestation},
and more later works were developed under the \textit{supercompilation} framework \cite{turchin-supercompiler}.

\citet{sorensen-unifying} revealed the connection between deforestation and supercompilation.
Supercompilation is more powerful and applies to general higher-order programs because
more information is propagated during its \emph{unfolding} process,
which produces a ``process tree''
capturing all information on the program execution states;
and in addition to \emph{folding}, they typically incorporate powerful termination-checking processes
based on the \emph{homeomorphic embedding testing}
among all traversed expressions \cite{lloyd-homeomorphic,leuschel-homeo}.
Many variations of supercompilers were proposed
\cite{Srensen1996APS, sorensen-perfect-supercomp, mitchell-superohaskell,
bolingbroke-supercoompilation-benchmark, Bolingbroke-hs10:supercomp-by-eval,
Jonsson09:pos-supercomp-ho-cbv}.
Most of them work for non-strict programs, and
while all these supercompilers show significant performance improvements,
they share the problem of long compilation time and large code size growth, 
which is concerning for their scalability and inherent complexity.
\citet{mitchell-superohaskell} reported that some simple programs in the ``imaginary'' \emph{nofib} subset
can take up to 5 minutes to compile, and one test also crashed their implementation
for unknown reasons.
Also, as reported in \cite{bolingbroke-supercoompilation-benchmark},
the compilation time may get approximately 1000 times slower in some benchmark tests and there are cases
where exponential code size growth may occur.
\citet{mitchell-superohaskell} identified that the main reason for the long compilation time
is the expensive \emph{homeomorphic embedding tests},
and there are some works on alleviating the scalability issues involving various curated heuristics
\cite{jonsson-code-explosion-of-supercomp, mitchell-rethinking}.
Their preliminary results are encouraging, but
it remains %
to be seen how supercompilers scale to real use cases.

\mypara{Distillation}
\citet{sorensen1994turchin} shows that supercompilation can only lead
to linear speedups of programs while \emph{distillation},
proposed by \citet{hamilton-distillation}, can achieve
super-linear improvements for certain programs.
Notably, the na\"ive list reversal function with time complexity $O(n^2)$
defined through list appension can be transformed into
the fast list reversal function that uses an accumulator and runs
in linear time.
While distillation can fuse more patterns, it also \emph{exacerbates} the
problems of supercompilation,
which is already very expensive, notably due to the requirement
of detecting and tying recursive knots on the fly while
expanding ``process trees'',
which requires the use of large sets of patterns to be compared
as the process progresses.
This patten-matching process is linear in the
number of patterns, so it grows linearly with the number
of traversed expressions,
and some expressions are traversed many times,
which likely makes the approach
unusable in practice on real-sized programs. Distillation
makes this much worse by having the pattern-matching process
itself defined on ``process trees'', a more complicated data
structure.

On the other hand, we believe that improving the complexity of user
programs is a non-goal and could lead to very brittle performance characteristics,
leading to unpredictable asymptotic performance changes,
depending on whether the optimizer succeeds or not, which
could make the difference between a program with acceptable
performance and a program that is
practically non-terminating.
This is somewhat
similar to the problems of SQL query engines, where it is
well-known that updates to the query engine software,
changes in the data distribution, or slight changes in the
query can make an efficient query impractically slow simply
because the optimizer stops using a plan with the right
complexity.
All in all, distillation seems 
generally impractical.

\mypara{Practicality concerns}
\label{related-implementation}
We now discuss 
the few available supercompilation implementations we could find
\!\footnote{\url{https://github.com/poitin/Higher-Order-Deforester}}
\!\footnote{\url{https://github.com/poitin/Higher-Order-Supercompiler}}
\!\footnote{\url{https://github.com/poitin/Distiller}}
and how they compare with \System.
These implementations range from higher-order deforestation
to supercompiler %
to distiller. They share
the same basic infrastructure of a small language parser and
its corresponding evaluator which can summarize the number
of reductions and allocations after each run.

In terms of \textit{optimization power}, these methods
are more powerful than \System.
As mentioned above, distillation can
produce super-linear improvements. 
Moreover, because of the \textit{unfolding} process which extensively unrolls the definitions
and statically evaluate programs,
they inherently do not have the limitation of misaligned recursive length
mentioned in \ref{sub-sec:misaligned-length}.

In terms of \textit{scalability},
the overall results we get agree with our above practicality concerns.
Among 66 small programs with an average file size
of just 16 lines of code (much smaller than the \emph{nofib} spectral programs),
3 programs make higher
order deforestation and supercompilation unable to terminate
within a reasonable time. For the distiller, this number is 7.
The ``msort'' program also crashes the distiller.
All three implementations also produce a slower output
for the ``maprev'' program,
according to the number of reductions reported by the
language evaluator. Lastly, for the distiller, the output
program of ``isort'' contains free variables that do
not exist in the input program.

Binary code size growth is harder to compare
because their language used is not compiled.
Due to the complicated unfolding/folding process involving homeomorphic embedding testing,
it is generally difficult to %
prevent code explosion when it occurs, which it does.
By contrast, code size growth in \System only comes from
one of the two cases explained in \Sec{sub-sec:analysis-of-res}.
This makes the output code more predictable, and
in practice, the above two
causes can more easily be controlled%
by fine-tuning 
the behavior of \System to strike a balance between
speedup and size growth.

We also tried the system presented as \emph{supercompilation by evaluation} by
\citet{Bolingbroke-hs10:supercomp-by-eval}
and quickly managed to find a closed program input for which it generated
an output program with free variables.
Moreover, according to the second author of this paper (from a private conversation),
they could never achieve supercompilation results where the speedup
outweighed the extreme
incurred code duplication.

\subsection{Lightweight Fusion}

\emph{Lightweight fusion} \cite{Ohori-popl07:light-fusion} also follows the \emph{unfold/fold} framework,
but it deserves special attention because 
it is much simpler and was evaluated against
a set of medium to large benchmarks.
This technique focuses on the very specific but common
pattern of two adjacent recursive function calls, of the form @f(g($\ldots$))@,
which is often a strong indication of a possible
fusion opportunity, whereby @f@ and @g@ could potentially be fused together into a new
definition @f_g@ equivalent to their composition @f@ @.@ @g@.
Lightweight fusion simply inlines the candidate
recursive functions only \textit{once}.
After this inlining, the calls to @f@ are distributed into the branches
branches of @g@ in the hope of getting the same pattern of @f(g($\ldots$))@.
If so, a new equivalent definition
@f_g@ is generated and \textit{promoted out} from the fixed
point combinator to replace all @f(g($\ldots$))@ call sites.
\citeauthor{Ohori-popl07:light-fusion} show that this lightweight approach
works well for fusing common examples;
however, their benchmark results show that the real-world performance improvements are limited.
Lightweight fusion also has severe additional limitations:
it cannot fuse mutual recursive functions or
successive calls to the same function, such as @map@ @.@ @map@.

It remains unclear how
much fusion opportunities can be discovered by this method
due to its heavy reliance on \textit{syntactic patterns} for discovering
fusion opportunities. We believe this makes it less
powerful in practice than \System, which uses
subtype inference to collect data flow information in the program.
This notably allows \System to see through helper definitions and efficiently
find matching pairs of producers and consumers.
As one concrete example, consider
the list returned from @flatten@ and consumed by @sum@ below.
Interestingly, both @flatten@ and
@append@ are producers of lists that are consumed
by @sum@ in this example.
\begin{lstlisting}
  sum (flatten list)  where  append [] ys = ys
                             append (x :: xs) ys = x :: (append xs ys)
      flatten [] = []                                          sum [] = 0
      flatten (x :: xs) = append x (flatten xs)                sum (x :: xs) = x + (sum xs)
\end{lstlisting}
Lightweight fusion can see
that after the body of @flatten@ is inlined once and the
call of @sum@ is distributed to the tail positions, we get
@sum@ @[]@ and @sum@ @(append@ @x@ @(flatten@ @xs))@ in the two
matching arm bodies of @flatten@. After this, it is clear
that another occurrence of @sum@ @(flatten@ @($\ldots$))@ will not appear
even if the definition of @sum@ is further inlined.
On the other hand, our approach can
directly discover that the @[]@ in the body of
@flatten@ and the @::@ in the body of @append@ can be
fused with @sum@.

\citeauthor{Ohori-popl07:light-fusion} also anticipate
extensions of Lightweight fusion to impure call-by-value semantics,
but the correctness theorem of their
transformation currently depends on \emph{call-by-name} semantics.
By contrast, we proved the correctness of \System
in \emph{call-by-value} semantics.

\subsection{Shortcut Fusion and Its Extensions}

\emph{Shortcut fusion} methods follow 
a different philosophy and focus on solving only
a subset of the problem by designing dedicated combinators for
commonly used data structures like lists and fusing those.
Those combinators typically capture the essence of constructing and
consuming the data structure and most commonly used
functions manipulating the data structure can
be rewritten using these combinators,
although the rewritten form may be \emph{less efficient} on its own.
The original @build/foldr@ \emph{list fusion} proposed by
\citet{shortcut-deforest} handles many common functions,
such as @map@, @filter@ and @repeat@,
but has various limitations, such as being unable to fully fuse the list appension function @(++)@ or @zip@.
Later, @augment/foldr@ \cite{Gill-96:cheap-deforestation} handles @(++)@;
@unfoldr/destroy@ \cite{Svenningsson-zipshortcut} handles @zip@ but not @filter@;
and \emph{stream fusion} due to \citet{coutts-icfp07:stream-fusion} handles all functions above,
but has difficulties with @concatMap@ and list comprehension.
The obvious shortcoming of these approaches is that they are not
\emph{fully automatic} --- they require programmers to manually rewrite
functions for specific data structures using these combinators to obtain optimization possibilities.

Previous work like warm fusion, introduced by \citet{warmfusion},
attempted to increase
the degree of automation by designing algorithmic
transformations to derive the combinator style of functions
from their original recursive definitions.
Later, \citet{Chitil-ifp99} focused mainly on the list data
structure and improved upon warm fusion to get
a more elegant method leveraging the idea of type inference \cite{lee-algoM}.
Some more complicated mechanisms were also
proposed using deeper algebraic properties of general data
constructors.
Structural hylomorphism, proposed by
\citet{hu-hylo}, and its later implementation HYLO
\cite{Onoue-hylo-haskell} for core Haskell seems to give
promising results on the three test cases reported by their
authors. Results from larger test cases are yet to be seen, and we also look
forward to more related progress in the future.

\mypara{More detailed comparison with GHC's shortcut fusion}
\label{comparison-with-ghc}
Considering that the list fusion implemented in GHC \cite{peytonjones2001playing}
appears to be the only practical fusion system used by an industry-strength compiler,
it is interesting to compare \System with it in detail, despite the latter one not being fully automatic.
Note that this is \emph{not}
a rigorous \emph{quantitative} comparison:
\System generates OCaml programs while GHC works on Haskell, so
the differences between the two evaluation strategies (call-by-need and call-by-value)
themselves may lead to additional differences in the speedup ratios.

To evaluate the speedup from GHC's list fusion,
we ran the original Haskell versions of the \emph{nofib} benchmark programs
using two versions of GHC:
the original version 9.4.8 of GHC,
and a modified version of GHC 
containing the following changes:
\begin{itemize}
  \item The list fusion rewrite rules like @foldr/build@ or @foldr/augment@ were deleted.
  \item The library functions manipulating lists that were implemented using @foldr@ and @build@
    to enable list fusion through the rewrite rules
    (which usually made them slower in the absence of actual fusion)
    were reimplemented using their obvious recursive implementations.
  \item The desugaring of list comprehension is modified to use the efficient method
    introduced in the textbook by \citet{listcomp} instead of @foldr@ and @build@.
\end{itemize}

The results are presented in \Fig{fig:hs-time}. The overall geometric mean of the ratio between
the runtime using original GHC vs modified GHC is 0.935, which is coincidentally close to \System.
However, for individual programs, \System and GHC's list fusion do not always
give a similar improvement ratio.
There are many reasons for this, including
the vast differences between their fundamental evaluation models and
the interaction between list fusion with later optimization passes performed by
GHC.
Notably, two programs run significantly slower (by an average of 17.5\%) when
compiled by unmodified GHC. As expected, this is mainly
because the @build@ and @foldr@ combinators used to implement most basic list functions
turn out not to be fused in this case.
As an example, for @sorting@, unmodified GHC generates a program that runs
about 17.7\% \emph{slower}
because @sorting@ uses the @partition@ function, which is implemented in
unmodified GHC using @foldr@ and does not get fused.
On the other hand, \System does not suffer from this efficiency penalty because
it does not generally require functions to be 
rewritten into a less efficient and non-idiomatic style for fusion to happen.

\section{Conclusion \& Future Work}
\label{sec:conclusion}

We demonstrated \System, a new general deforestation technique
based on a subtype inference analysis followed by an elaboration pass.
This architecture allows us to find producer/consumer pairs and rewrite them
by importing the relevant consumer computations into the producers
while preserving pure call-by-value semantics.
We showed that further local simplifications then allowed
to remove the creation of intermediate data structures,
resulting in significant speedups for a variety of programs
from the standard \emph{nofib} benchmark suite.
One major advantage of our approach is that it is relatively
simple to specify and implement. We formally proved its soundness
using logical relations, which only required less than 6 pages of proofs
following a standard structure.

In the future, we would like to generalize \System to the
\emph{impure} functional programming setting,
whereby programs may perform computations with side effects
like mutation and I/O.
We envision that this can be handled by threading an abstract ``state of the world''
parameter through these computations
while leaving pure computations unaffected,
so that the pure parts of the program are still fused as before.

\begin{acks}                            %
  This research was partially funded by Hong Kong Research Grant Council project number 26208821.
\end{acks}

\section*{Data-Availability Statement}
The artifact of this paper, which was granted the ``Functional'' badge, is available
on Zenodo at \url{https://zenodo.org/records/11500626} \cite{artifact-lumberhack}.
The corresponding Github
repository is available at \url{https://github.com/hkust-taco/lumberhack}.

\appendix

\section{Additional Explanations and Examples}
\label{app:advanced-examples}

We first provide a formal explanation for the \emph{thunking}
transformation, and then demonstrate more examples on
\emph{chained fusion} and fusing \emph{accumulator-passing} producers.

\subsection{The Thunking Transformation}
\label{sec:scope-extr}

As we have described in \Sec{sec:preserve-cbv}, \Sec{sec:computations-fvs}
and \Sec{sub-sec:thunking-free-vars},
we can preprocess the input program to perform the \emph{thunking} transformation
to preserve call-by-value semantics and handle the problem of scope extrusion.
We now explain it formally in \Fig{fig:thunking} below.
We denote a repetition of variables as $\overline{a}$
and posit the special case that an \emph{empty} repetition of variables
yields a \emph{unit} literal $()$.

\begin{figure*}[h]
  \mprset{sep=1.5em}
  {\small
  \begin{mathpar}
  \boxed{t => t}

  \inferrule[Thunk-Var]{
  }{
    x => x
  }

  \inferrule[Thunk-App]{
    t_1 => t_1'
    \\
    t_2 => t_2'
  }{
    t_1\, t_2 => t_1'\, t_2'
  }

  \text{\textcolor{gray}{(other rules are similar except for the cases below)}}

  \inferrule[Thunk-Case]{
    \overline{t_i => t_i'}^i
    \\
    s => s'
    \\
    \overline{a} = \bigcup_i \left(\fname{FV}(t_i')\backslash \left\{\overline{x_{i,j}}^j\right\}\right)
  }{
    \left(\kw{case}\ s\ \kw{of}\ \left(\,\overline{\,
    c_i\ \overline{x_{i,j}}^j -> t_i
    \,}^i\,\right)\right) =>
    \left(\kw{case}\ s'\ \kw{of}\ \left(\,\overline{\,
    c_i\ \overline{x_{i,j}}^j -> \lambda \overline{a}.\,t_i'
    \,}^i\,\right)\right)\,\overline{a}
  }
  \end{mathpar}}
  \caption{Thunking rules}
  \label{fig:thunking}
\end{figure*}

Notice that since the computation performed by
a matching branch can potentially be imported into an
arbitrary place where the producer exists, we need to
capture all the variables referred to by the computation
(except for those introduced by its matching pattern) to make it a
standalone ``thunk'' of computation in order to freely move it around.
And when there is no free variable, a dummy one is created to
form a thunk for preserving call-by-value semantics.
This can be achieved by a local program transformation as a
pre-process before the main deforestation procedure of \System.

This transformation ``thunk''s pattern-matching branch bodies
by wrapping them inside a series
of lambda abstractions, where the parameters are the variables that are not 
in the matching patterns, but referred to in the branch bodies.
Then, to retain the original semantics of the programs, those
lambda abstractions in the branches are applied immediately by the
same set of variables.

\subsection{Chained Fusion}

As an example of what we call
\emph{chained fusion},
consider 
`$\id{sumi}\ (\id{mapi}\ F\ (\id{mapi}\ G\ L))$'
where:
\begin{align*}
  \quad\id{sumi}\ (x::xs) &~=~ x + \id{sumi}\ xs 
  \\
  \quad\id{mapi}\ f\ (y::ys) &~=~ f\ y :: \id{mapi}\ ys
\end{align*}

Notice that when considered as a \emph{consumer}, $\id{mapi}$ is problematic,
as its body captures a parameter $f$ defined outside of the case expression of interest.
Moreover, $\id{mapi}$ is used twice with different strategies here
--- once as a pure producer and once as a producer \emph{and} consumer.
Therefore, we first apply a pass of 
definition duplication refactoring,
as explained in
\Sec{sec:preserve-struct-defs},
and then a pass of 
\emph{thunking} (for the consumer version of $\id{mapi}$,
which we choose to name $\id{mapi}_1$),
yielding:
\begin{align*}
  &\kw{let rec}\ \id{sumi}\ (x::xs) ~=~ x + \id{sumi}\ xs\ \kw{in}\ 
  \\
  &\kw{let rec}\ \id{mapi}_1\ f_1\ ys ~=~ (\kw{case}\ ys\ \kw{of}\ (y_1::ys_1)\ \to \lam f_1.\ (f_1\ y_1 :: \id{mapi}_1\ f_1\ ys_1))\ f_1\ \kw{in}\ 
  \\
  &\kw{let rec}\ \id{mapi}_2\ f_2\ (y_2::ys_2) ~=~ f_2\ y_2 :: \id{mapi}_2\ f_2\ ys_2\ \kw{in}\ 
  \\
  &\id{sumi}\ (\id{mapi}_1\ (\id{mapi}_2\ G\ L)\ F)
\end{align*}

As before, we can show that
$(\id{sumi}|->(S_1 -> \top)) |- t_{\id{sumi}} ~> t_{\id{sumi}}' <=
  S_1 -> \top
$, with the same $S_1$ and $t_{\id{sumi}}'$ as in \Sec{sec:simple-rec-example-sumi-mapsqi}.
Moreover, we can show that
$(\id{mapi}_1|->S_2->\top->S_1) |- t_{\id{mapi}_1} ~> t_{\id{mapi}_1}' <= S_2->\top->S_1
$
where
$t_{\id{mapi}_1}' =
  \lam f_1.\ \lam p_1.\ 
  p_1\ f_1
$
and
\begin{align*}
S_2 =
  \left\{\left(
    (::)\,\langle y_1 |-> \top, ys_1 |-> S_2\rangle ->
      \lam f_1.\ \kw{let } x = f_1\ y_1;\;xs = \id{mapi}_1\ f_1\ ys_1 \kw{ in } \lam\dummy.\ x+\id{sumi}\ xs
  \right)\right\}
\end{align*}
and
$(\id{mapi}_2|->\top->S_{\fname{list}}->S_2) |- t_{\id{mapi}_2} ~> t_{\id{mapi}_2}' <= \top->S_{\fname{list}}->S_2$,
where
$t_{\id{mapi}_2}' =
  \lambda\,f_2\ p_2.\ 
  \kw{case}\ p_2\ \kw{of}\ y_2::ys_2 ->
  \kw{let } y_1 = f_2\ y_2;\;ys_1 = \id{mapi}_2\ f_2\ ys_2
  \kw{ in }
  \lam f_1.\ 
  \kw{let } x = f_1\ y_1;\;xs = \id{mapi}_1\ f_1\ ys_1
  \kw{ in }
  \lam\dummy.\ 
  x+\id{sumi}\ xs
$.

So we can rewrite the original program to the fused one:
\begin{align*}
  &\kw{let rec}\ \id{sumi}\ p ~=~ p\ \dummy\ \kw{in}\ 
  \\
  &\kw{let rec}\ \id{mapi}_1\ f_1\ p_1 ~=~ p_1\ f_1\ \kw{in}\ 
  \\
  &\kw{let rec}\ \id{mapi}_2\ f_2\ (y_2::ys_2) ~=~
  \\
  &\quad
  \kw{let } y_1 = f_2\ y_2;\;
  ys_1 = \id{mapi}_2\ f_2\ ys_2
  \kw{ in }
  \\
  &\qquad
  \lam f_1.\ \kw{let } x = f_1\ y_1;\;xs = \id{mapi}_1\ f_1\ ys_1
  \kw{ in }
  \lam\dummy.\ x+\id{sumi}\ xs
  \\
  &\kw{in }
  \id{sumi}\ (\id{mapi}_1\ (\id{mapi}_2\ G\ L)\ F)
\end{align*}
which can be simplified to just:
\begin{align*}
  &\kw{let rec}\ \id{mapi}_2\ f_2\ (y::ys)\ f_1\ \dummy ~=~
  f_1\ (f_2\ y) + \id{mapi}_2\ f_2\ ys\ f_1\ \dummy
  \\
  &\kw{in } \id{mapi}_2\ G\ L\ F\ \dummy
\end{align*}

\subsection{When Fusion Accumulates Continuations}
\label{sec:accum-cont}

As explained in \Sec{subsec:accum-fusion},
fusing an accumulator-passing producer with a parameterized consumer
will result in a fused function that accumulates a continuation.

For example, consider:
\begin{align*}
  &\kw{let rec}\ \id{idxSum}\ i\ xs ~=~
    \kw{case}\ xs\ \kw{of}\ 
    x::xs -> i + \id{idxSum}\ (i + 1)\ xs,\ [] -> 0
  \\
  &
  \kw{in}\ 
  \kw{let rec}\ \id{rev}\ a\ ys ~=~
    \kw{case}\ ys\ \kw{of}\ y::ys -> \id{rev}\ (y::a)\ ys,\ [] -> a
  \\
  &\kw{in }
  \id{idxSum}\ 1\ (\id{rev}\ []\ L)
\end{align*}

Refactoring produces:
\begin{align*}
  &\kw{let rec}\ \id{idxSum}\ i\ xs ~=~
    (\kw{case}\ xs\ \kw{of}\ 
    x::xs -> \lam i.\ i + \id{idxSum}\ (i + 1)\ xs,\ [] -> \lam i.\ 0)\ i
  \\
  &
  \kw{in}\ 
  \kw{let rec}\ \id{rev}\ a\ ys ~=~
    \textit{(same as before)}
  \\
  &\kw{in }
  \id{idxSum}\ 1\ (\id{rev}\ []\ L)
\end{align*}

Then fusion and simplification produce:
\begin{align*}
  &\kw{let rec}\ \id{rev}\ a\ ys ~=~
    \kw{case}\ ys\ \kw{of}\ 
    y::ys -> \id{rev}\ (\lam i.\ i + a\ (i + 1))\ ys,\ [] -> a
  \\
  &\kw{in }
  \id{rev}\ (\lam i.\ 0)\ L\ 1
\end{align*}

On the one hand, doing such transformations may lead to further fusion opportunities
and therefore be important to apply.

On the other hand,
these new `$\lam i.\ \ldots$' accumulated continuations
could actually be counter-productive in terms of performance
if no further fusion happens, as in our example:
we have essentially gone from allocating first-order
concrete constructors and pattern matching on them, in the original program,
to allocating higher-order functions
and calling them (which will use virtual dispatch),
in the rewritten program.
The number of allocations, which is usually a dominating measure in the runtime
of pure functional language, has remained the same.
But we have traded pattern matching for virtual calls,
and the latter is often slower.

A standard optimization that helps in that case is defunctionalization,
which replaces each distinct lambda by a new constructor
and replaces applications of such lambdas by a case analysis on these constructors.
In our example, %
defunctionalization results in:
\begin{align*}
  &\kw{let rec}\ \id{call}\ a\ i ~=~
    \kw{case}\ a\ \kw{of}\ 
    \id{C}_0 -> 0,\ \id{C}_1\ a -> i + \id{call}\ a\ (i + 1)
  \\
  &\kw{in }
  \kw{let rec}\ \id{rev}\ a\ ys ~=~
    \kw{case}\ ys\ \kw{of}\ 
    y::ys ->
    \id{rev}\ (\id{C}_1\ a)\ ys
      ,\ 
    [] -> a
  \\
  &\kw{in }
  \id{call}\ (\id{rev}\ \id{C}_0\ L)\ 1
\end{align*}
...which is isomorphic to the original program!
The only difference is that we no longer uselessly store the elements
of the list of which we want to sum the indices.
This seems like a hint that the program in question is fundamentally not really fusible
(at last not with a fusion approach like ours).

Therefore, the following strategy emerges:
perform fusion as much as possible even if it introduces accumulated continuations,
and then perform a pass of defunctionalization
to make sure we recover at least the performance of the original program.

We reserve investigating this approach for future work.

\section{Excerpts from Full Proofs}
\label{sec:partial-proofs}

The poofs of the theorems in this paper are
given in full in \tr{}{}.
We only provide excerpts of the proofs here
due to space limitations.

\begin{proof}[Proof of \Thm{fundamentalterm}]
  This proof relies on the corresponding \textit{compatibility lemmas}.
  There is one compatibility lemma for each of the fusion rules
  defined in \Fig{fig:fusion-prescient}.
  As their names suggest, \textit{compatibility lemmas} ensure that the logical relation
  interpretation of strategies as sets of triples given above are
  compatible with the fusion rules we defined. For example,
  the compatibility lemma corresponding to \textsc{F-Lam} is
  $
  \Gamma \cdot (x |-> S_1) |- t_1 \lesssim^{\fname{str}} t_2: S_2
  \implies \Gamma |- \lam x.\ t_1 \lesssim^{\fname{str}} \lam x.\ t_2: S_1 -> S_2
  $,
  which has the same structure as the \textsc{F-Lam} rule.
  Their proofs are given in \tr{\App{compatifusion}}{\thetr} and mainly
  just involve the expansion of definitions. %
  With the compatibility lemmas,
  this theorem is proven by induction on the fusion rules.
  Each case is immediately handled by its corresponding compatibility lemma.
\end{proof}

\begin{lemma}[Reflexitivity for Typed Contexts]\label{reflctx}
  $$
  \Xi |- C: (\Xi' |- I) -> I' => \Xi |- C \lesssim C: (\Xi' |- I) -> I'
  $$
\end{lemma}

\begin{proof}[Proof of \Thm{soundnessoftransformation}]

  By theorem \ref{fundamentalterm}, we have
  $\Xi |- t ~> t' <= I
  \implies
  \Xi |- t \lesssim^{\fname{str}} t': I$.
  Then, because $I$ is a proper subset of $S$, we have
  $\Xi |- t \lesssim^{\fname{str}} t': I \implies \Xi |- t \lesssim^{\fname{type}} t': I$.
  After that, we show
  $\Xi |- t \lesssim^{\fname{type}} t': I \implies \Xi |- t \lesssim^{\fname{ctx}} t': I$ as follows.
  By expanding the definition of $\lesssim^{\fname{ctx}}$ to
  $\forall C, I'.\ \varepsilon |- C: (\Xi |- I) -> I' \land C[t] \Downarrow \implies C[t'] \Downarrow$,
  we have $\varepsilon |- C: (\Xi |- I) -> I'$ and $C[t] \Downarrow$ as premises along with
  $\Xi |- t \lesssim^{\fname{type}} t' : I$.
  Now we need to show that $C[t'] \Downarrow$.
  
  By \Lem{reflctx}, we have
  $\varepsilon |- C \lesssim C: (\Xi |- I) -> I'$, which means that the context $C$
  is logically related to itself. Instantiate this with $t$ and $t'$, note that this
  instantiation is valid because we have $\Xi |- t \lesssim^{\fname{type}} t' : I$
  in the premises. Now we have $\varepsilon |- C[t] \lesssim^{\fname{type}} C[t']: I'$.

  Since $C[t] \Downarrow$, we have that there exits $v, k$ such that $C[t] \prog{k} v$. Then
  we expand the definition of $\lesssim^{\fname{type}}$ and $\calEeq{}{I'}$ 
  in $\varepsilon |- C[t] \lesssim^{\fname{type}} C[t']: I'$
  and instantiate it
  by $k + 1$ and two empty substitutions $\varepsilon$ and $\varepsilon$.
  This instantiation is valid because
  $(k + 1) \ge 0$, $(k+1, \varepsilon, \varepsilon) \in \calGeq{}{\varepsilon}$, and
  $\fname{irred}(v)$. So now we have
  $\exists v'.\,C[t'] \prog{*} v' \land (k + 1 - k,\ v,\ v') \in \calVeq{}{I'}$. Then by
  definition of $\calVeq{}{I'}$, we conclude that $C[t'] \Downarrow v'$.

\end{proof}

\section{Pseudo-Algorithmic Fusion Rules}
\label{app:pseudo-algorithmic}

\newcommand{\calA}{\mathcal{A}}
\newcommand{\squigtoa}{\rightsquigarrow_{\calA}}

\newcommand{\calM}{\mathcal{M}}%
\newcommand{\calT}{\mathcal{T}}%

The simple and intuitive declarative system presented in \Sec{sec:formal} is ``prescient'',
in the sense that in \textsc{F-Case} and \ruleName{F-Ctor}, it uses \emph{rewritten}
$\kw{case}$ branch terms $l_i'$
in constructor fusion strategies,
while the type inference rules, on the other hand,
does not yet know how to rewrite these terms and uses the untransformed
$\kw{case}$ branch terms $l_i$ in the inferred scrutinee strategy.

Thankfully, this difference is benign:
we can reformulate the simple declarative rules of \Sec{sec:formal}
into ``pseudo-algorithmic'' rules that use $l_i$ instead of $l_i'$.
And ``pseudo-algorithmic'' derivations can always be translated to
derivations in the ``prescient'' declarative system
(see \Sec{sec:pseudo-algo-to-prescient}).

\subsection{Syntax and Rules}

The syntax of pseudo-algorithmic fusion strategies, which are
produced by the unification process, is presented in \Fig{fig:pseudo-algo-syntax}.
It is mostly the same
as the fusion strategies we presented in \Sec{sec:formal}.
The only difference is
the addition of a ``$:S$'' at the end of constructor fusion strategies.

\begin{figure}
  {\small
  \begin{align*}
  S &~::=~
  \top
  ~\mid~
  \bot
  ~\mid~
  S -> S
  ~\mid~
  \mu F_S
  ~\mid~
  \left\{\,\overline{
    c_i\,\langle\,\overline{\,S_{i,j}\,}^j\,\rangle
  }^i\right\}
  ~\mid~
  \left\{\,\overline{
    c_i\,\langle\,\overline{x_{i,j} |-> S_{i,j}}^j\,\rangle->l_i
  }^i\,\right\}: S
  \end{align*}
  }
  \caption{Syntax of pseudo-algorithmic fusion strategies}
  \label{fig:pseudo-algo-syntax}
\end{figure}

The pseudo-algorithmic fusion rules are presented in \Fig{fig:fusion-algo}.
In these rules,
the newly added ``$:S$'' at the end of constructor fusion strategies
represents the result of $\kw{case}$ expressions,
and they are used to transform the attached $\kw{case}$ branch terms
in \ruleName{P-Ctor} and \ruleName{P-Case}.

The pseudo-algorithmic rules are more complicated,
which is why the main body of the paper only presents the simple declarative rules
to help intuition and understanding.

\begin{figure*}
  \mprset{sep=1.3em}
  {\footnotesize
  \begin{mathpar}
  \boxed{\mathstrut\Gamma |- t ~>_{\mathcal{A}} t <= S}
  
  \inferrule[P-Var]{
    (x |-> S) \in \Gamma
  }{
    \Gamma |-
    x ~>_{\mathcal{A}} x
    <= S
  }
  
  \inferrule[P-Lam]{
    \Gamma\cat(x|->S') |-
    t
    ~>_{\mathcal{A}}
    t'
    <= S
  }{
    \Gamma |-
    \lam x.\ t
    ~>_{\mathcal{A}}
    \lam x.\ t'
    <= S' -> S
  }
  
  \inferrule[P-Top]{
    \Gamma |- t ~>_{\mathcal{A}} t' <= I
  }{
    \Gamma |- t ~>_{\mathcal{A}} t' <= \top
  }
  
  \inferrule[P-LetRec]{
    \Gamma \cat (x |-> S_1) |-
    t_1 ~>_{\mathcal{A}} t_1'
    <= S_1
    \\
    \Gamma \cat (x |-> S_1) |-
    t_2 ~>_{\mathcal{A}} t_2'
    <= S_2
  }{
    \Gamma |-
    \kw{let rec}\ x = t_1\ \kw{in}\ t_2
    ~>_{\mathcal{A}}
    \kw{let rec}\ x = t_1'\ \kw{in}\ t_2'
    <=
    S_2
  }
  
  \inferrule[P-App]{
    \Gamma |- t_1 ~>_{\mathcal{A}} t_1'
    <= (S' -> S)
    \\
    \Gamma |- t_2 ~>_{\mathcal{A}} t_2'
    <= S'
  }{
    \Gamma |-
    t_1\ t_2 ~>_{\mathcal{A}} t_1'\ t_2'
    <= S
  }
  
  \inferrule[P-Ctor]{
    \left( c_k\,\langle\,\overline{x_{k,j} |-> S_{k,j}}^j\,\rangle->l_k \right): S_{\fname{res}} \in S_{\mathcal{A}}
    \\
    \overline{\,\Gamma |- t_{k,j} ~>_{\mathcal{A}} t_{k,j}' <= S_{k,j}}^j
    \\\\
    \Gamma\cat
      \overline{
        (x_{k,j} |-> S_{k,j})
      }^j
      |-
      l_k ~>_{\mathcal{A}} l_k' <= S_{\fname{res}}
    \\
    \fname{FV}(l_k) \cup \fname{FV}(l_k') \subseteq \left\{\overline{x_{k,j}}^j\right\}
  }{
    \Gamma |-
    c_k\ \overline{t_{k,j}}^j ~>_{\mathcal{A}}
    \kw{let } \overline{x_{k,j} = t_{k,j}'}^j \kw{ in } l_k'
    <= S_{\mathcal{A}}
  }
  
  \inferrule[P-CtorSkip]{
    \left(c\,\langle \overline{S_j}^j \rangle\right) \in S
    \\
    \overline{\, \Gamma |- t_j ~>_{\mathcal{A}} t_j' <= S_j}^j
  }{
    \Gamma |-
    c\ \overline{\, t_j\,}^j ~>_{\mathcal{A}} c\ \overline{\, t_j'\,}^j
    <=
    S
  }
  
  \inferrule[P-Case]{
    \Gamma |- t ~>_{\mathcal{A}} t' <= \left\{\overline{c_i\,\langle \overline{x_{i,j} |-> S_{i,j}}^j \rangle->l_i}^i \right\}: S
    \\
    \overline{\Gamma\cat
      \overline{
        (x_{i,j} |-> S_{i,j})
      }^j
      |-
      l_i ~>_{\mathcal{A}} l_i' <= S}^i
    \\
    \overline{\fname{FV}(l_i) \cup \fname{FV}(l_i') \subseteq \left\{ \overline{x_{i,j}}^j \right\}}^i
  }{
    \Gamma |-
    \kw{case}\ t\ \kw{of}\ \overline{\,
      c_i\ \overline{x_{i,j}}^j -> l_i
    \,}^i
    ~>_{\mathcal{A}}
    t'
    <=
    S
  }
  
  \inferrule[P-CaseSkip]{
    \Gamma |- t ~>_{\mathcal{A}} t' <= \left\{\overline{c_i\,\langle \overline{S_{i,j}}^j \rangle}^i\right\}
    \\
    \overline{\,
      \Gamma \cat \overline{(x_{i,j} |-> S_{i,j})}^j |- l_i ~>_{\mathcal{A}} l_i' <= S
    \,}^i
  }{
    \Gamma |-
    \kw{case}\ t\ \kw{of}\ \overline{\,
      c_i\ \overline{x_{i,j}}^j -> l_i
    \,}^i
    ~>_{\mathcal{A}}
    \kw{case}\ t'\ \kw{of}\ \overline{\,
      c_i\ \overline{x_{i,j}}^j -> l_i'
    \,}^i
    <=
    S
  }
  \quad
  \inferrule[P-RecFold]{
    \Gamma |- t ~>_{\mathcal{A}} t' <= S(\mu S)
  }{
    \Gamma |- t ~>_{\mathcal{A}} t' <= \mu S
  }
  \quad
  \inferrule[P-RecUnfold]{
    \Gamma |- t ~>_{\mathcal{A}} t' <= \mu S
  }{
    \Gamma |- t ~>_{\mathcal{A}} t' <= S(\mu S)
  }
  \end{mathpar}}
  \caption{Pseudo-algorithmic fusion rules.}
  \label{fig:fusion-algo}
\end{figure*}

\subsection{Translation}
\label{sec:pseudo-algo-to-prescient}

The pseudo-algorithmic derivation tree is traversed twice for performing the translation.
The first traversal produces
a mapping from pseudo-algorithmic constructor fusion strategies
to their corresponding \emph{transformed} $\kw{case}$ branch terms.
The mapping is easily accumulated every time
the traversing process encounters a derivation of \ruleName{P-Ctor} or \ruleName{P-Case},
which contains the related transformation of $\kw{case}$ branch terms.
Then the second traversal performs the actual translation:
the mapping accumulated during the first traversal is applied to every
constructor fusion pseudo-algorithmic strategy in the derivation so that
their $\kw{case}$ branch terms are replaced with the transformed ones
and the ``$:S$'' is dropped. Additionally,
the sub-derivations in \ruleName{P-Ctor} related to the
transformation of $\kw{case}$ branch terms are also
dropped since they are no longer necessary in the corresponding \ruleName{F-Ctor} rule.
A more detailed description is presented in \App{sec:pseudo-algo-to-prescient-full}.

\section{Detailed Benchmark Result Analysis}
\label{app:benchmark-explanations}

We ran the benchmarks on a MacBook Pro 2021 (M1 Max 32GB).
The command line argument passed to the optimizing OCaml
compiler @ocamlopt@ are: @-O3@ and @-rectypes@. @-O3@ ensures that
@ocamlopt@ can perform extensive inlinings to simplify
our output programs; @-rectypes@ enables the support
of recursive types in OCaml.

\begin{itemize}
  \item \textbf{Ansi}: In this program, many string constants get consumed immediately. Since smaller string constants are treated as lists, they get fused, leading to improvements.
  \item \textbf{Banner}: The program extensively uses many standard functions on lists. The main fusion site that contributes to the efficiency improvement is that the @concat@ fuses with @map@ in the definition of @unline@, which is called on large lists in the original program.
  \item \textbf{Calendar}: The program extensively uses many standard functions on lists. Many of them get fused, which collectively contributes to the efficiency improvement.
  \item \textbf{Constraints}: Besides lists that get fused like those in @fillTable@ or @emptyTable@, this program also extensively uses a recursively defined @Tree@ data structure: @data Tree a = Node a [Tree a]@, and traverses them during the computation. \System manages to fuse some of them like those @mapTree@s in @bm@ and @lookupCache@, which contributes significantly to the improvements.
  \item \textbf{Fish}: The fusion site in the entry point function @testFish_nofib@ where the result of @fmt@ flowing into @length@ greatly contributes to the efficiency improvement.
  \item \textbf{Gcd}: Two @[a..b]@ lists in the @test@ function get immediately consumed by a list comprehension, and this fusion site greatly contributes to the efficiency improvement.
  \item \textbf{Mandel2}: The tree built by @build_tree@ is directly consumed by @finite@, and our fusion totally fuses it.
  \item \textbf{Mandel}: Two intermediate lists created in @mandelset@, @[1..screenY]@ and @[1..screenX]@ are directly consumed by the list comprehension, and the result of the list comprehension is immediately mapped over in @parallelMandel@ and then returned as part of the return value. All the intermediate lists above are fused. 
\end{itemize}

\clearpage

\bibliography{paper}

\ifdefined\TechReport

\clearpage

\setlength{\abovedisplayskip}{3pt}
\setlength{\belowdisplayskip}{4pt}
\setlength{\abovedisplayshortskip}{3pt}
\setlength{\belowdisplayshortskip}{4pt}

\section{Proofs of Declarative System}
\label{app:proofs}

\begin{definition}[Syntax for context $C$]\label{syntaxForCtx}
  For the conciseness of presentation, we use $h$ to denote a context $C$
  or a concrete term $t$, and require that there can only be one 
  $C$ in the branches or the fields for the cases of
  $\kw{case}\ t\ \kw{of}\ \overline{c_i\ \overline{x_{i,j}}^j \to \lam x_i.\,h_i}^i$
  and $c\ \overline{h_i}^i$.
  In typing rules for contexts and the following formalization, $i_C$ %
  represents the index such that $h_{i_C}$ is a context $C$.
  And $i_t$ ranges over the indices of other $h$s that represent terms.
  $$
  \begin{aligned}
    C &~::=~ [\cdot] ~\mid~ C\ t ~\mid~ t\ C ~\mid~ \lam x.\,C
    ~\mid~ \kw{case}\ C\ \kw{of}\ \overline{c_i\ \overline{x_{i,j}}^{j} \to l_i}^i ~\mid~
    \\
    &\phantom{~::=~}
    \kw{case}\ t\ \kw{of}\ \overline{c_i\ \overline{x_{i,j}}^j \to \lam x_i.\,h_i}^i
    ~\mid~
    c\ \overline{h_i}^i
    \text{  where there is only one context $C$ in $\left\{\overline{h_i}^i\right\}$}%
    \\
    h &~::=~ C ~\mid~ t
  \end{aligned}
  $$
\end{definition}

\begin{definition}[Typing Rules for Terms]\label{termtyping}
  Typing rules for terms are defined in \Fig{fig:termtyping}
  \begin{figure*}
    \mprset{sep=1.3em}
    {\small
    \begin{mathpar}
    \boxed{\mathstrut\Xi |- t: I}
    \hfill
    \\
    \inferrule[TTerm-Var]{
      (x |-> I) \in \Xi
    }{
      \Xi |-
      x: I
    }
    
    \inferrule[TTerm-Lam]{
      \Xi\cat(x|->I') |-
      t: I
    }{
      \Xi |-
      \lam x.\ t
      : I' -> I
    }
    
    \inferrule[TTerm-RecFold]{
      \Xi |- t: I(\mu I)
    }{
      \Xi |- t: \mu I
    }
    
    \inferrule[TTerm-RecUnfold]{
      \Xi |- t: \mu I
    }{
      \Xi |- t: I(\mu I)
    }

    \inferrule[TTerm-App]{
      \Xi |- t_1: (I' -> I)
      \\
      \Xi |- t_2: I'
    }{
      \Xi |- t_1\ t_2: I
    }
    
    \inferrule[TTerm-Ctor]{
      \left(c\ \langle \overline{I_j}^j \rangle\right) \in I
      \\
      \overline{\, \Xi |- t_j: I_j\,}^j
    }{
      \Xi |-
      c\ \overline{\, t_j\,}^j: I
    }
    
    \inferrule[TTerm-Case]{
      \Xi |- t: \left\{\overline{c_i\,\langle \overline{I_{i,j}}^j \rangle}^i\right\}
      \\
      \overline{\,
        \Xi \cat \overline{(x_{i,j} |-> I_{i,j})}^j |- l_i: I
      \,}^i
    }{
      \Xi |-
      \kw{case}\ t\ \kw{of}\ \overline{\,
        c_i\ \overline{x_{i,j}}^j -> l_i
      \,}^i : I
    }
        
    \inferrule[TTerm-Top]{
      \Xi |- t: I
    }{
      \Xi |- t: \top
    }

    \end{mathpar}}
    \caption{Typing Rules for Terms}
    \label{fig:termtyping}
  \end{figure*}
\end{definition}

\begin{definition}[Typing Rules for Contexts]\label{ctxtyping}
  Typing rules for contexts are defined in \Fig{fig:ctxtyping}
  \begin{figure*}
  {\small
  \begin{mathpar}
    \boxed{\mathstrut\Xi |- C: (\Xi' |- I) \to I'}
    \hfill
    \\
    \inferrule[TCtx-Id]{
      \Xi \subseteq \Xi'
    }{
      \Xi' |- [\cdot]: (\Xi |- I) \to I
    }
    
    \inferrule[TCtx-App1]{
      \Xi' |- C: (\Xi |- I) -> (I_1 -> I_2)
      \\
      \Xi' |- t: I_1
    }{
      \Xi' |- C\ t: (\Xi |- I) -> I_2
    }
    
    \inferrule[TCtx-App2]{
      \Xi' |- t: I_1 -> I_2
      \\
      \Xi' |- C: (\Xi |- I) -> I_1
    }{
      \Xi' |- t\ C: (\Xi |- I) -> I_2
    }
    
    \inferrule[TCtx-Lam]{
      \Xi' \catsp (x: I_1) |- C: \left(\Xi \catsp (x: I_1) |- I\right) -> I_2
    }{
      \Xi' |- \lam x.\,C: (\Xi \catsp (x: I_1) |- I) -> (I_1 -> I_2)
    }
    
    \inferrule[TCtx-Case1]{
      \Xi' |- C: (\Xi |- I) -> \left\{ \overline{c_i\ \langle \overline{I_{i,j}}^j \rangle}^i \right\}
      \\
      \overline{\Xi' \catsp \overline{(x_{i,j} |-> I_{i,j})}^j |- l_i: I'}^i
    }{
      \Xi' |- \kw{case}\ C\ \kw{of}\ \overline{c_i\ \overline{x_{i,j}}^j -> l_i}^i
      :
      (\Xi |- I) -> I'
    }
    
    \inferrule[TCtx-Case2]{
      \Xi' |- t: \left\{ \overline{c_i\ \langle \overline{I_{i,j}}^j \rangle}^i \right\}
      \\
      \overline{\Xi' \cat \overline{x_{i_t,j} |-> I_{i_t,j}}^j |- \lam x_{i_t}.\,h_{i_t}: I'}^{i_t}
      \\
      \Xi' \cat \overline{x_{i_C,j} |-> I_{i_C,j}}^j |- \lam x_{i_C}.\,b_{i_C}: \left(\Xi \catsp \overline{(x_{i_C,j} |-> I_{i_C,j})}^j |- I\right) -> I'
    }{
      \Xi' |- \kw{case}\ t\ \kw{of}\ \overline{c_i\ \overline{x_{i,j}}^j -> \lam x_i.\,h_i}^i
      :
      \left(\Xi\catsp\overline{x_{i_C, j} |-> I_{i_C, j}}^j |- I\right) -> I'
    }

    \inferrule[TCtx-Ctor]{
      \left(c\ \langle \overline{I_i}^i \rangle\right) \in I'
      \\
      \Xi' |- h_{i_C}: (\Xi |- I) -> I_{i_C}
      \\
      \overline{\Xi' |- h_{i_t}: I_{i_t}}^{i_t}
    }{
      \Xi' |- c\ \overline{h_i}^i: (\Xi |- I) -> I'
    }

    \inferrule[TCtx-RecFold]{
      \Xi' |- C: (\Xi |- I) -> I'(\mu I')
    }{
      \Xi' |- C: (\Xi |- I) -> \mu I'
    }
    
    \inferrule[TCtx-RecUnfold]{
      \Xi' |- C: (\Xi |- I) -> \mu I'
    }{
      \Xi' |- C: (\Xi |- I) -> I'(\mu I')
    }
    
    \inferrule[TCtx-Ctx]{
      \Xi' |- C: (\Xi_1 |- I_1) -> I'
      \\
      \Xi_1 |- C_1: (\Xi |- I) -> I_1
    }{
      \Xi' |- C[C_1[\cdot]]: (\Xi |- I) -> I'
    }
    \\
    \boxed{\mathstrut\Xi' |- C[t]: I'}
    \hfill
    \\
    \inferrule[Ctx]{
      \Xi' |- C: (\Xi |- I) -> I'
      \\
      \Xi |- t: I
    }{
      \Xi' |- C[t]: I'
    }
  \end{mathpar}
  }
  \caption{Typing Rules for Contexts}
  \label{fig:ctxtyping}
  \end{figure*}
\end{definition}

\begin{definition}[Relational Interpretations of Types $I$ and Typing Contexts $\Xi$]\label{relinterpretationypes}
  $\calVeq{}{I}$ and $\calEeq{}{I}$ are the relational interpretations of strategies
  in terms of values and computations, respectively. Their definitions are stratified
  using the same step-index technique explained in \Def{relinterpretationstrat}.
  We also lift this notion to typing contexts by defining $\calGeq{}{\Xi}$
  such that related typing contexts map variables to related values.
  \begin{itemize}
    \item
    $\calVeq{}{\top} \defeq \left\{(n, v, v') ~\mid~ n \ge 0\right\}$
    \\
    $\calVeq{}{\bot} \defeq \{\}$
    \\
    $\calVeq{}{\mu I}
    \defeq
    \bigcup_n \calVeq{}{I^n(\bot)}$ where $I^n(\bot)$ denotes $I(I(\cdots I(\bot)))$
    \\
    $\calVeq{}{I_1 -> I_2}
    \defeq
    \begin{aligned}[t]
      \left\{\left(n,\, \lam x.t_1,\, \lam x.t_2\right) ~\mid~
      \right.&\forall m < n.
      \forall v_1, v_2. \left(
        (m,\,v_1,\,v_2) \in \calVeq{}{I_1}
      \right)
      \\
      &\left.\left.
      \implies
      \left(m, [x |-> v_1]\ t_1,\, [x |-> v_2]\ t_2\right) \in \calEeq{}{I_2}
      \right)
      \right\}
    \end{aligned}$
    \\
    $\calVeq{}{\left\{\overline{\left(
      c_i\ \langle\overline{I_{i,j}}^j\rangle
    \right)}^i\right\}} \defeq
    \bigcup_i \left\{\left(
      n,\ c_i\ \overline{v_{i,j}}^j,\ c_i\ \overline{v_{i,j}'}^j
    \right) ~\mid~
    \wedge_j \forall m < n.\ (m,\,v_{i,j},\,v_{i,j}') \in \calVeq{}{I_{i,j}}
    \right\}
    $
    \item $
    \begin{aligned}[t]
      \calEeq{}{I}
      \defeq
      \left\{(n,\ t,\ t') ~\mid~
      \right.&\forall m < n.\ \forall v. \left((t \prog{m} v) \land \fname{irred}(v)\right)\\
      &\implies\left.
      \left(\exists v'.\ t' \prog{*} v' \land \left(n - m,\ v,\ v'\right) \in \calVeq{}{I}\right)\right\}
    \end{aligned}
    $
    \item 
    $\calGeq{}{\ep} \defeq \{(n,\ \ep,\ \ep)\}$\\
    $\begin{aligned}
      \calGeq{}{\Xi \cat (x |-> S)} \defeq \{(n,\ \phi_1 \cat (x|->v),\ \phi_2 \cat (x|->v'))
      ~\mid~
      &(n,\ \phi_1,\ \phi_2) \in \calGeq{}{\Xi} \\
      &\land (n,\ v,\ v') \in \calVeq{}{I}\}
    \end{aligned}$
  \end{itemize}
\end{definition}

\begin{lemma}[Downward Closure For Types under Decreasing Index]\label{downclosuretype}
  This lemma is useful for proving \Lem{compatityping}
  
  There are two sub-lemmas, one for $\calVeq{}{\cdot}$ and another for $\calGeq{}{\cdot}$
  \begin{itemize}
    \item $\forall n, v, v', I.\ (n, v, v') \in \calVeq{}{I} \implies \forall m.\ m < n \implies (m, v, v') \in \calVeq{}{I}$
    \item $\forall n, \phi, \phi', \Xi.\ (n, \phi, \phi') \in \calGeq{}{\Xi} \implies \forall m.\ m < n \implies (m, \phi, \phi') \in \calGeq{}{\Xi}$
  \end{itemize}
\end{lemma}
  \begin{proof}[Proof of \Lem{downclosuretype}]
    Since types $I$ is a proper subset of strategies $S$,
    and the definitions of
    $\calVeq{}{\cdot}$ and $\calV{}{\cdot}$, $\calEeq{}{\cdot}$ and $\calEeq{}{\cdot}$
    are almost identical. The proof of \Lem{downclosurestrat} (shown below) can be directly
    adapted here.
  \end{proof}

\begin{lemma}[Compatibility Lemmas for Term Typing Rules]\label{compatityping}
  These lemmas ensure that the logical
  relation interpretation of types as sets of triples given above are compatible with the term typing rules
  we defined, so there will be one sub-lemma for each term typing rule.
  \begin{itemize}
    \item $\left(x |-> I\right) \in \Xi \implies \Xi |- x \lesssim^{\fname{type}} x: I$
    \item $\Xi \catsp (x |-> I_1) |- t_1 \lesssim^{\fname{type}} t_2: I_2
          \implies \Xi |- \lam x.\ t_1 \lesssim^{\fname{type}} \lam x.\ t_2: I_1 -> I_2$
    \item $\Xi |- t \lesssim^{\fname{type}} t': F_I(\mu F_I)
          \implies
          \Xi |- t \lesssim^{\fname{type}} t': \mu F_I$ and its converse
    \item $\Xi |- t_1 \lesssim^{\fname{type}} t_1': I_1 -> I_2
          \land
          \Xi |- t_2 \lesssim^{\fname{type}} t_2': I_1
          \implies
          \Xi |- t_1\ t_2 \lesssim^{\fname{type}} t_1'\ t_2': I_2$
    \item $\left(c,\, \langle\overline{I_j}^j\rangle\right) \in I
          \land
          \overline{\Xi |- t_j \lesssim^{\fname{type}} t_j': I_j}^j
          \implies
          \Xi |- c\ \overline{t_j}^j
          \lesssim^{\fname{type}}
          c\ \overline{t_j'}^j: I$
    \item $
      \Xi |- t \lesssim^{\fname{type}} t':
      \left\{\overline{\left(
        c_i,\,\langle \overline{I_{i,j}}^j \rangle
      \right)}^i\right\}
      \land
      \overline{\,
        \Xi \catsp \overline{(x_{i,j} |-> I_{i,j})}^j |- t_i \lesssim^{\fname{type}} t_i': I
      \,}^i
      \implies
      \\
      ~\hfill\Xi |-
      \left(\kw{case}\ t\ \kw{of}\ \left(\,\overline{\,
        c_i\ \overline{x_{i,j}}^j -> l_i
      \,}^i\,\right)\right)
      \lesssim^{\fname{type}}
      \left(\kw{case}\ t'\ \kw{of}\ \left(\,\overline{\,
        c_i\ \overline{x_{i,j}}^j -> l_i'
      \,}^i\,\right)\right) : I
    $
    \item $\Xi |- t \lesssim^{\fname{type}} t' : I \implies \Gamma |- t \lesssim^{\fname{type}} t' : \top$
  \end{itemize}
  \begin{proof}[Proof of \Lem{compatityping}]
    Again, since types $I$ is a proper subset of strategies $S$,
    and the definitions of term typing rules and fusion rules
    are almost identical.
    The proof of \Lem{compatifusion} (shown below) can be directly
    adapted here.

    We choose to elaborate on \Lem{downclosurestrat} and \Lem{compatifusion} instead of
    this lemma and \Lem{downclosuretype},
    because \Lem{downclosurestrat} and \Lem{compatifusion} reflects
    the main additions we have made to the system by introducing fusion strategies
    that contain terms that we want to import to the constructor site.
  \end{proof}
\end{lemma}

\begin{theorem}[Reflexitivity for Typed Terms]\label{refltypedterm}
  $$
  \Xi |- t: I => \Xi |- t \lesssim^{\text{type}} t: I
  $$
  \begin{proof}[Proof of \Thm{refltypedterm}]
    This can be proved directly by case analysis and induction on the
    term typing rules. Each case will be immediately handled
    by its corresponding compatibility lemma we developed in \Lem{compatityping}.
  \end{proof}
\end{theorem}

\begin{lemma}[Compatibility Lemmas for Context Typing Rules]\label{compatictxtyping}
  These lemmas make sure that our definition for two contexts being related to each
  other is compatible with the context typing rules we defined.
  \begin{itemize}
    \item $\Xi \subseteq \Xi' \implies \Xi' |- [\cdot] \lesssim [\cdot]: (\Xi |- I) -> I$
          \begin{proof}
            By expanding the definition of $\lesssim$, consider arbitrary
            $t$ and $t'$ such that $\Xi |- t \lesssim^{\fname{type}} t': I$.
            We are then required to show that
            $\Xi' |- [t] \lesssim^{\fname{type}} [t']: I$, which is to show that
            $\Xi' |- t \lesssim^{\fname{type}} t': I$.
          
            By expanding the definition of $\lesssim^{\fname{type}}$, consider
            arbitrary $n \ge 0$ and $\phi$, $\phi'$ such that
            $(n,\ \phi, \phi') \in \calGeq{}{\Xi'}$, we are required to show that
            $(n, \ \phi\,[t],\ \phi'\,[t']) \in \calEeq{}{I}$. Then we expand the
            definition of $\lesssim^{\fname{type}}$ in our assumption
            $\Xi |- t \lesssim^{\fname{type}} t': I$,
            and instantiate it with
            $n \ge 0$, $\phi|_{\fname{dom}(\Xi)}$, $\phi'|_{\fname{dom}(\Xi)}$. Note
            that this instantiation is valid because $\Xi \subseteq \Xi'$, so
            $(n,\ \phi|_{\fname{dom}(\Xi)},\ \phi'|_{\fname{dom}(\Xi)}) \in \calGeq{}{\Xi}$.
            Now we get
            $(n, [\phi|_{\fname{dom}(\Xi)}]\,t,\ [\phi'|_{\fname{dom}(\Xi)}]\,t') \in \calEeq{}{I}$.
            And because $\Xi \subseteq \Xi'$, it is equivalent to
            $(n, \ \phi\,[t],\ \phi'\,[t']) \in \calEeq{}{I}$.
          \end{proof}
    \item $\Xi' |- t_1 \lesssim^{\fname{type}} t_2: I_1
          \ \land
          \ \Xi' |- C_1 \lesssim C_2: (\Xi |- I) -> (I_1 -> I_2)
          \implies \\
          ~\hfill \Xi' |- C_1\ t_1 \lesssim C_2\ t_2: (\Xi |- I) -> I_2$
          \begin{proof}
            By expanding the definition of $\lesssim$, consider arbitrary
            $t$ and $t'$ such that $\Xi |- t \lesssim^{\fname{type}} t': I$.
            Now we are required to show that
            $\Xi' |- (C_1[t])\ t_1 \lesssim^{\fname{type}} (C_2[t'])\ t_2: I_2$.
            By expanding the definition of $\lesssim$ in our assumption
            $\Xi' |- C_1 \lesssim C_2: (\Xi |- I) -> (I_1 -> I_2)$ and instantiate
            it with $t$ and $t'$, note that this instantiation is valid because
            $\Xi |- t \lesssim^{\fname{type}} t': I$. Now we get that
            $\Xi |- C_1[t] \lesssim^{\fname{type}} C_2[t']: I_1 -> I_2$. Then we can show
            $\Xi' |- (C_1[t])\ t_1 \lesssim^{\fname{type}} (C_2[t'])\ t_2: I_2$ by
            applying \Lem{compatityping}.
          \end{proof}
    \item $\Xi' |- t_1 \lesssim^{\fname{type}} t_2: I_1 -> I_2
          \ \land
          \ \Xi' |- C_1 \lesssim C_2: (\Xi |- I) -> I_1
          \implies \\
          ~\hfill \Xi' |- t_1\ C_1 \lesssim t_2\ C_2: (\Xi |- I) -> I_2$
          \begin{proof}
            By expanding the definition of $\lesssim$, consider arbitrary
            $t$ and $t'$ such that $\Xi |- t \lesssim^{\fname{type}} t': I$.
            Now we are required to show that
            $\Xi' |- t_1\ (C_1[t]) \lesssim^{\fname{type}} t_2\ (C_2[t']): I_2$.
            By expanding the definition of $\lesssim$ in our assumption
            $\Xi' |- C_1 \lesssim C_2: (\Xi |- I) -> I_1$ and instantiating
            it with $t$ and $t'$, note that this instantiation is valid because
            $\Xi |- t \lesssim^{\fname{type}} t': I$. Now we get that
            $\Xi |- C_1[t] \lesssim^{\fname{type}} C_2[t']: I_1$. Then we can show
            $\Xi' |- t_1\ (C_1[t]) \lesssim^{\fname{type}} t_2\ (C_2[t']): I_2$ by
            applying \Lem{compatityping}.
          \end{proof}
    \item $\Xi' \catsp (x: I_1) |- C_1 \lesssim C_2: (\Xi \catsp (x: I_1) |- I) -> I_2
          \implies \\
          ~\hfill \Xi' |- \lam x.C_1 \lesssim \lam x.C_2: (\Xi \catsp (x:I_1) |- I) -> (I_1 -> I_2)
          $
          \begin{proof}
            By expanding the definition of $\lesssim$, consider arbitrary
            $t$ and $t'$ such that
            $\Xi\catsp(x: I_1) |- t \lesssim^{\fname{type}} t': I$,
            now we are required to show that
            $\Xi' |- \lam x.C_1[t] \lesssim^{\fname{type}} \lam x.C_2[t']: I_1 -> I_2$.
            By expanding the definition of $\lesssim$ in our assumption and instantiating it with
            $t$ and $t'$, we get that
            $\Xi'\catsp(x:I_1) |- C_1[t] \lesssim^{\fname{type}} C_2[t']: I_2$. Then we can show
            that $\Xi' |- \lam x.C_1[t] \lesssim^{\fname{type}} \lam x.C_2[t']: I_1 -> I_2$
            by applying \Lem{compatityping}.
          \end{proof}
    \item $\Xi' |- C_1 \lesssim C_2 : (\Xi |- I) -> \left\{ \overline{c_i\ \langle\overline{I_{i,j}}^j\rangle}^i \right\}
          \ \land \overline{\Xi' \catsp \overline{(x_{i,j} |-> I_{i,j})}^j |- l_{i1} \lesssim^{\fname{type}} l_{i2} : I'}^i
          \implies \\
          ~\hfill \Xi' |-
          \kw{case}\ C_1\ \kw{of}\ \overline{c_i\ \overline{x_{i,j}}^j -> l_{i1}}^i
          \lesssim
          \kw{case}\ C_2\ \kw{of}\ \overline{c_i\ \overline{x_{i,j}}^j -> l_{i2}}^i
          :
          (\Xi |- I) -> I'
          $
          \begin{proof}
            By expanding the definition of $\lesssim$, consider arbitrary
            $t$ and $t'$ such that $\Xi |- t \lesssim^{\fname{type}} t': I$,
            now we are required to show that
            $$
            \Xi' |- \kw{case}\ C_1[t]\ \kw{of}\ \overline{c_i\ \overline{x_{i,j}}^j -> l_{i1}}^i
            \lesssim^{\fname{type}}
            \kw{case}\ C_2[t']\ \kw{of}\ \overline{c_i\ \overline{x_{i,j}}^j -> l_{i2}}^i
            : I'
            $$
            Now by expanding the definition of $\lesssim$ in our assumption
            $\Xi' |- C_1 \lesssim C_2 : (\Xi |- I) -> \left\{ \overline{c_i\ \langle\overline{I_{i,j}}^j\rangle}^i \right\}$
            and instantiating it with $t$ and $t'$, we get that
            $$\Xi' |- C_1[t] \lesssim^{\fname{type}} C_2[t']: \left\{ \overline{c_i\ \langle\overline{I_{i,j}}^j\rangle}^i \right\}$$
            Then we can show the conclusion by applying \Lem{compatityping}.
          \end{proof}
    \item $\Xi' |- t_1 \lesssim^{\fname{type}} t_2: \left\{\overline{c_i\ \langle\overline{I_{i,j}}^j\rangle}^i\right\}
          \land\ \overline{\Xi'\catsp\overline{x_{i_t, j} |-> I_{i_t, j}}^j |- \lam x_{i_t1}.\,b_{i_t1} \lesssim^{\fname{type}} \lam x_{i_t2}.\,b_{i_t2}: I'}^{i_t}
          \\
          \land\ \Xi'\catsp\overline{x_{i_C, j} |-> I_{i_C, j}}^j |- \lam x_{i_C1}.\,h_{i_C1} \lesssim \lam x_{i_C2}.\,h_{i_C2}: \left(\Xi\catsp\overline{x_{i_C, j} |-> I_{i_C, j}}^j |- I\right) -> I'
          \implies
          \\
          ~\hfill
          \begin{aligned}
            \Xi' |-\,
            &\kw{case}\ t_1\ \kw{of}\ \overline{c_i\ \overline{x_{i,j}}^j -> \lam x_{i1}.\,h_{i1}}^i
            \lesssim
            \\
            &\kw{case}\ t_2\ \kw{of}\ \overline{c_i\ \overline{x_{i,j}}^j -> \lam x_{i2}.\,h_{i2}}^i
            : \left(\Xi\catsp\overline{x_{i_C, j} |-> I_{i_C, j}}^j |- I\right) -> I'
          \end{aligned}
          $
          \begin{proof}
            By expanding the definition of $\lesssim$, consider arbitrary
            $t$ and $t'$ such that
            $$\Xi'\catsp\overline{x_{i_C, j} |-> I_{i_C, j}}^j |- t \lesssim^{\fname{type}} t': I$$
            Now instantiate
            $\Xi'\catsp\overline{x_{i_C, j} |-> I_{i_C, j}}^j |- \lam x_{i_C1}.\,h_{i_C1} \lesssim \lam x_{i_C2}.\,h_{i_C2}: \left(\Xi\catsp\overline{x_{i_C, j} |-> I_{i_C, j}}^j |- I\right) -> I'$
            with $t$ and $t'$ and get that
            $\Xi'\catsp\overline{x_{i_C, j} |-> I_{i_C, j}}^j |- \lam x_{i_C1}.\,h_{i_C1}[t] \lesssim^{\fname{type}} \lam x_{i_C2}.\,h_{i_C2}[t']: I'$.
            Then we can show the conclusion by applying \Lem{compatityping}.
          \end{proof}
    \item $\left(c\ \langle\overline{I_i}^i\rangle\right) \in I'
          \ \land
          \ \Xi' |- h_{i_C1} \lesssim h_{i_C2}: (\Xi |- I) -> I_{i_C}
          \ \land
          \ \overline{\Xi' |- h_{i_t1} \lesssim^{\fname{type}} h_{i_t2}: I_{i_t}}^{i_t}
          \\
          ~\hfill \implies \Xi' |- c\ \overline{h_{i1}}^i \lesssim c\ \overline{h_{i2}}^i:
          \left(\Xi |- I\right) -> I'$
          \begin{proof}
            By expanding the definition of $\lesssim$, consider arbitrary
            $t$ and $t'$ such that
            $\Xi |- t \lesssim^{\fname{type}} t': I$. Then instantiate the
            assumption $\Xi' |- h_{i_C1} \lesssim h_{i_C2}: (\Xi |- I) -> I_{i_C}$
            with it, we get that
            $\Xi' |- h_{i_C1}[t] \lesssim^{\fname{type}} h_{i_C2}[t']: I_{i_C}$.
            Then we can show the conclusion by applying \Lem{compatityping}.
          \end{proof}
    \item $\Xi' |- C_1 \lesssim C_2: (\Xi |- I) ->I'(\mu I')
          \implies
          \Xi' |- C_1 \lesssim C_2: (\Xi |- I) -> \mu I'$
          \begin{proof}
            By expanding the definition of $\lesssim$, consider arbitrary
            $t$ and $t'$ such that
            $\Xi |- t \lesssim^{\fname{type}} t': I$. Then instantiate the assumption
            with it, we get that
            $\Xi' |- C_1[t] \lesssim^{\fname{type}} C_2[t']: I'(\mu I')$.
            Then we can show the conclusion by applying \Lem{compatityping}.
          \end{proof}
    \item $\Xi' |- C_1 \lesssim C_2: (\Xi |- I) -> \mu I'
          \implies
          \Xi' |- C_1 \lesssim C_2: (\Xi |- I) ->I'(\mu I')$
          \begin{proof}
            By expanding the definition of $\lesssim$, consider arbitrary
            $t$ and $t'$ such that
            $\Xi |- t \lesssim^{\fname{type}} t': I$. Then instantiate the assumption
            with it, we get that
            $\Xi' |- C_1[t] \lesssim^{\fname{type}} C_2[t']: \mu I'$.
            Then we can show the conclusion by applying \Lem{compatityping}.
          \end{proof}
    \item $\Xi' |- C_1 \lesssim C_2: (\Xi_1 |- I_1) -> I'
          \land\ \Xi_1 |- C_1' \lesssim C_2': (\Xi |- I) -> I_1
          \\
          ~\hfill \implies \Xi' |- C_1[C_1'[\cdot]] \lesssim C_2[C_2'[\cdot]]: (\Xi |- I) -> I'
          $
          \begin{proof}
            By expanding the definition of $\lesssim$, consider arbitrary
            $t$ and $t'$ such that
            $\Xi |- t \lesssim^{\fname{type}} t': I$. Then instantiate
            the assumption $\Xi_1 |- C_1' \lesssim C_2': (\Xi |- I) -> I_1$ with
            it and get
            $\Xi_1 |- C_1'[t] \lesssim^{\fname{type}} C_2'[t']: I_1$.
            Then use this to instantiate the first assumption
            ($\Xi' |- C_1 \lesssim C_2: (\Xi_1 |- I_1) -> I'$),
            and get that
            $\Xi' |- C_1[C_1'[t]] \lesssim^{\fname{type}} C_2[C_2'[t']]: I'$,
            which concludes this proof.
          \end{proof}
  \end{itemize}
\end{lemma}

\begin{proof}[Proof of \Lem{reflctx}]
  This can be proved directly by case analysis and induction on the
  context typing rules. Each case will be immediately handled
  by its corresponding compatibility lemma we developed in \Lem{compatictxtyping}.
\end{proof}

\begin{lemma}[Downward Closure For Fusion Strategies under Decreasing Index]\label{downclosurestrat}
  This lemma is useful for proving \Lem{compatifusion}.
  
  There are two sub-lemmas, one for $\calV{}{\cdot}$ and another for $\calG{}{\cdot}$
  \begin{itemize}
    \item $\forall n, v, v', S.\ (n, v, v') \in \calV{}{S} \implies \forall m.\ m < n \implies (m, v, v') \in \calV{}{S}$
    \item $\forall n, \phi, \phi', \Gamma.\ (n, \phi, \phi') \in \calG{}{\Gamma} \implies \forall m.\ m < n \implies (m, \phi, \phi') \in \calG{}{\Gamma}$
  \end{itemize}
\end{lemma}

  \begin{proof}[Proof of \Lem{downclosurestrat}]
    The lemma for $\calV{}{\cdot}$ can be proved by
    induction and case analysis on
    strategy expressions. Cases other than
    recursive strategies can be proved by
    expanding definitions. For the case of
    $\mu F_S$, The proofs by
    \citeauthor{appel-equi-rec-type}
    which use the notion of \textit{well-founded strategy constructors}
    are very similar and can easily be adapted to our system, 
    so we only give an overview of the proof below.

    Intuitively, the \textit{well-foundedness} of a
    strategy constructor $F$ means that for any
    $k \ge 0$ and strategy $s$, in order to
    determine whether $(k, t, t') \in \calE{}{F(s)}$,
    it suffices to know the set
    $\{(n, t, t') \mid n < k \land (n, t, t') \in \calV{}{s}\}$.

    We can then firstly show that the case for recursive types $\mu F_S$ holds
    for all \textit{well-founded strategy constructors}
    $F_S$, and then show
    that all strategy constructors that can be built by the
    syntax, except for the constant bottom and top strategy constructor,
    are well-founded. In our system, the case for
    function strategies $S_1 -> S_2$ can be proved
    exactly like the proof by \citeauthor{appel-equi-rec-type}.
    And the two cases for data constructors,
    namely $\calV{}{\left\{\overline{\left(
      c_i\ \langle\overline{S_{i,j}}^j\rangle
    \right)}^i\right\}}$ and
    $\calV{}{\left\{\overline{\left(
      c_i\ \langle\overline{x_{i,j} |-> S_{i,j}}^j\rangle -> l_i
    \right)}^i\right\}}$
    can be proved
    easily by extending their proofs of product types.

    For the constant-bottom strategy
    constructor $F_{\bot}$, the above conclusion trivially
    holds because $\calV{}{\mu F_{\bot}} = \{\}$;
    For the constant-top strategy
    constructor $F_{\top}$, the above conclusion trivially
    holds because $\calV{}{\mu F_{\top}} = \{(n, v, v') ~\mid~ n \ge 0\}$,
    according to the definition of $\calV{}{\top}$.

    The lemma for $\calG{}{\cdot}$ can then be easily proved by
    induction on $\Gamma$ once we have the lemma for
    $\calV{}{\cdot}$.
  \end{proof}

\begin{lemma}[Compatibility Lemmas for Fusion Rules]\label{compatifusion}
  These lemmas ensure that the logical relation
  interpretation of strategies as sets of triples given above are
  compatible with the fusion rules we defined, so there
  will be one sub-lemma for each fusion rule. Note that
  there is no need to handle F-LetRec since
  \kw{let rec} is only a syntax sugar in our language.
  \begin{itemize}
    \item $\left(x |-> S\right) \in \Gamma \implies \Gamma |- x \lesssim^{\fname{str}} x: S$
    \begin{proof}
      By the definition of $\lesssim^{\fname{str}}$,
      consider arbitrary $n \ge 0, \phi, \phi'$
      such that $(n,\,\phi,\,\phi') \in \calG{}{\Gamma}$.
      We are required to show that $(n, [\phi]\ x, [\phi']\ x) \in \calE{}{S}$.
      
      By the definition of $\calE{}{S}$,
      consider arbitrary $m < n$ and $v$
      such that $[\phi]\ x \prog{m} v \land \fname{irred}(v)$.
      We are required to show that
      $\exists v'. [\phi']\ x \prog{*} v' \land (n - m,\,v,v') \in \calV{}{S}$.

      By the definition of $\calG{}{\Gamma}$, we have
      $(n, [\phi]\ x, [\phi']\ x) \in \calV{}{S}$.
      By the definition of $\calV{}{S}$, $([\phi]\ x)$ is a value,
      hence $m = 0$ and $v \equiv [\phi]\ x$.
      So $\exists v' = [\phi']\ x$, such that
      $[\phi']\ x \prog{0} v' \land (n - 0, v, v') \in \calV{}{S}$.
    \end{proof}

    \item $\Gamma \cdot (x |-> S_1) |- t_1 \lesssim^{\fname{str}} t_2: S_2
    \implies \Gamma |- \lam x.\ t_1 \lesssim^{\fname{str}} \lam x.\ t_2: S_1 -> S_2$
    \begin{proof}
      By the definition of $\lesssim^{\fname{str}}$,
      consider arbitrary $n \ge 0, \phi_1, \phi_2$
      such that $(n,\,\phi_1,\,\phi_2) \in \calG{}{\Gamma}$.
      We are required to show that
      $(n,\ \lam x. [\phi_1]\ t_1,\ \lam x. [\phi_2]\ t_2) \in \calE{}{S_1 -> S_2}$.

      By the definition of $\calE{}{S_1 -> S_2}$,
      consider arbitrary $m < n$ and $v$
      such that $\lam x. [\phi_1]\ t_1 \prog{m} v \land \fname{irred}(v)$.
      So $m$ can only be $0$ and $v \equiv \lam x. [\phi_1]\ t_1$, and we are
      required to show that
      $(n - 0,\ \lam x. [\phi_1]\ t_1,\ \lam x. [\phi_2]\ t_2) \in \calV{}{S_1 -> S_2}$.

      By the definition of $\calV{}{S_1 -> S_2}$,
      consider arbitrary $0 \le m < n$ and $v_1, v_2$ such that
      $(m,v_1,v_2) \in \calV{}{S_1}$.
      We are required to show that
      $(m,\ [x |-> v_1][\phi_1]t_1,\ [x |-> v_2][\phi_2]t_2) \in \calE{}{S_2}$.

      By the definition of $\lesssim^{\fname{str}}$ and the premise, we have that
      $\forall n' \ge 0.
      \forall \phi_1', \phi_2'.
      (n',\ \phi_1',\ \phi_2') \in \calG{}{\Gamma\cdot(x |-> S_1)}
      \implies (n',\ [\phi_1']\,t_1,\ [\phi_2']\,t_2) \in \calE{}{S_2}$.
      
      Instantiate this with $m$ as $n'$, $\phi_1 \cdot (x|->v_1)$ as $\phi_1'$ and
      $\phi_2 \cdot (x|->v_2)$ as $\phi_2'$. Note that this instantiation is valid,
      because $m \ge 0$, and from $m < n$, $(m,v_1,v_2) \in \calV{}{S_1}$,
      $(n,\,\phi_1,\,\phi_2) \in \calG{}{\Gamma}$ and lemma \ref{downclosurestrat},
      it can be derived that
      $(m,\ \phi_1\cdot(x|->v_1),\ \phi_2\cdot(x|->v_2)) \in \calG{}{\Gamma\cdot(x |-> S_1)}$.
      So we get $(m,\ [\phi_1\cdot(x|->v_1)]\,t_1,\ [\phi_2\cdot(x|->v_2)]\,t_2) \in \calE{}{S_2}$,
      which is equivalent to
      $(m,\ [x |-> v_1][\phi_1]t_1,\ [x |-> v_2][\phi_2]t_2) \in \calE{}{S_2}$.
    \end{proof}
    
    \item $\Gamma |- t \lesssim^{\fname{str}} t': F_S(\mu F_S)
    \implies
    \Gamma |- t \lesssim^{\fname{str}} t': \mu F_S$ and its converse
    \begin{proof}
      By the definition of $\lesssim^{\fname{str}}$,
      It suffices to show that $\calV{}{\mu F_S} = \calV{}{F_S(\mu F_S)}$.
      
      Again, the proofs by Appel and McAllester \cite{appel-equi-rec-type}
      which use the notion of \textit{well-founded strategy constructors}
      can be easily adapted to our case. They prove that
      the above equation holds for all \textit{well-founded strategy constructors}
      and show that all strategy constructors that can be built by the
      syntax, except for the constant bottom, are
      well-founded.

      Then for the constant-bottom strategy
      constructor $F_{\bot}$, the above conclusion trivially
      holds because $\calV{}{\mu F_{\bot}} = \{\}$. The same is true
      for constant-top strategy constructor $F_{\top}$, because
      $\calV{}{\mu F_{\top}} = \{(n, v, v') ~\mid~ n \ge 0\}$.
    \end{proof}

    \item $\Gamma |- t_1 \lesssim^{\fname{str}} t_1': S_1 -> S_2
    \land
    \Gamma |- t_2 \lesssim^{\fname{str}} t_2': S_1
    \implies
    \Gamma |- t_1\ t_2 \lesssim^{\fname{str}} t_1'\ t_2': S_2$
    \begin{proof}
      By the definition of $\lesssim^{\fname{str}}$, consider
      arbitrary $n \ge 0$, $\phi_1, \phi_2$ such that
      $(n, \phi_1, \phi_2) \in \calG{}{\Gamma}$.
      We are required to show that
      $(n, [\phi_1] t_1\, [\phi_1] t_2,\ [\phi_2] t_1'\, [\phi_2] t_2') \in \calE{}{S_2}$.

      By the definition of $\calE{}{S_2}$, consider arbitrary $m < n$ and $v$
      such that $[\phi_1]t_1\,[\phi_1]t_2 \prog{m} v \land \fname{irred}(v)$.
      We are required to show that
      $\exists v'.\ [\phi_2]t_1'\,[\phi_2]t_2'\prog{*}v' \land (n-m,\ v,\ v') \in \calV{}{S_2}$.

      By the operational semantics and
      $[\phi_1]t_1\,[\phi_1]t_2 \prog{m} v \land \fname{irred}(v)$,
      we get that
      $\exists p \le m, t_f$ such that $[\phi_1]t_1 \prog{p} t_f \land \fname{irred}(t_f)$.

      At the same time, by the definition of $\lesssim^{\fname{str}}$,
      we instantiate the first premise
      $\Gamma |- t_1 \lesssim^{\fname{str}} t_1': S_1 -> S_2$ with $n \ge 0$ and
      $(n, \phi_1, \phi_2) \in \calG{}{\Gamma}$, and get
      $(n, [\phi_1]t_1,\ [\phi_2]t_1') \in \calE{}{S_1 -> S_2}$.
      
      By the definition of $\calE{}{S_1 -> S_2}$, we get
      $\forall m' < n.\ \forall t_f''.
      ([\phi_1]t_1 \prog{m'} t_f'' \land \fname{irred}(t_f''))
      \implies (
        \exists t_f'. [\phi_2]t_1' \prog{*} t_f'
        \land
        (n - m', t_f'', t_f') \in \calV{}{S_1 -> S_2}
      )$. Then instantiate this with $p$ as $m'$, $t_f$ as $t_f''$. Note that this
      instantiation is valid because
      $p \le m < n$ and $[\phi_1]t_1 \prog{p} t_f \land \fname{irred}(t_f)$.
      By the definition of
      $\calV{}{S_1 -> S_2}$, we denote $t_f$ as $\lam x.\ t_f$ and
      $t_f'$ as $\lam x.\ t_f'$ below.
      So we get
      \begin{equation}\tag{1}\label{compatapp1}
        \exists t_f'.\ [\phi_2]t_1' \prog{*} \lam x.\ t_f' \land
        (n - p,\ \lam x.\,t_f,\ \lam x.\,t_f') \in \calV{}{S_1 -> S_2}
      \end{equation}

      Now that we have
      $$
      [\phi_1]t_1\,[\phi_1]t_2
      \prog{p}
      (\lam x.\,t_f)\ [\phi_1]t_2
      \prog{m - p}
      v
      \land \fname{irred}(v)
      $$
      By the operational semantics,
      we get that there exists a $q < m - p$ and an $a$ such that
      $[\phi_1]t_2 \prog{q} a \land \fname{irred}(a)$.

      At the same time, by the definition of $\lesssim^{\fname{str}}$,
      we instantiate the second premise
      $\Gamma |- t_2 \lesssim^{\fname{str}} t_2': S_1$ with $n - p$,
      and $\phi_1, \phi_2$. Note that this instantiation is also valid since
      $n - p \ge 0$ and $(n - p,\ \phi_1,\ \phi_2) \in \calG{}{\Gamma}$ by
      lemma \ref{downclosurestrat}. So we get that
      $(n - p,\ [\phi_1]t_2,\ [\phi_2]t_2') \in \calE{}{S_2}$.

      By the definition of $\calE{}{S_2}$, we get
      $\forall q' < n-p.\ \forall a''.\ 
      ([\phi_1]t_2 \prog{q'} a'' \land \fname{irred}(a''))
      \implies
      (\exists a'.\ [\phi_2]t_2' \prog{*} a' \land (n-p-q', a'', a') \in \calV{}{S_2})$
      Then instantiate this with $q$ as $q'$, and $a$ as $a''$. Note that
      this instantiation is valid because $q < m-p < n-p$ and
      $[\phi_1]t_2 \prog{q} a \land \fname{irred}(a)$. So we get
      \begin{equation}\tag{2}\label{compatapp2}
        \exists a'.\ [\phi_2]t_2' \prog{*} a' \land (n-p-q,\ a,\ a') \in \calV{}{S_2}
      \end{equation}

      Now that we have
      $$
      [\phi_1]t_1\ [\phi_1]t_2
      \prog{p}
      (\lam x.\,t_f)\ [\phi_1]t_2
      \prog{q}
      (\lam x.\,t_f)\ a
      \prog{1}
      [x|->a]\,t_f
      \prog{m - p - q - 1}
      v \land \fname{irred}(v)
      $$
      Since $(n-p,\ \lam x.\,t_f,\ \lam x.\,t_f')\in\calV{}{S_1 -> S_2}$, by the definition of
      $\calV{}{S_1 -> S_2}$, we have
      $\forall m' < n - p.\ \forall v,v'.\ (m',\ v,\ v') \in \calV{}{S_1}
      \implies (m',\ [x|->v]t_f,\ [x|->v']t_f') \in \calE{}{S_2}$.
      Instantiate this with $n-p-q-1$ as $m'$, $a$ as $v$ and $a'$ as $v'$.
      This instantiation is easily verified valid, and we get
      $(n-p-q-1, [x|->a]t_f,\ [x|->a']t_f') \in \calE{}{S_2}$.

      By definition of $\calE{}{S_2}$, we get
      $\forall k < n-p-q-1.\ \forall v''.\
      ([x|->a]t_f \prog{k} v'' \land \fname{irred}(v''))
      \implies
      (
        \exists v'. [x|->a']t_f' \prog{*} v'
        \land (n-p-q-1-k,\ v'',\ v') \in \calV{}{S_2})$.
      Instantiate this with $m-p-q-1$ as $k$ and $v$ as $v''$.
      Note that this instantiation is easily verified, and we get
      \begin{equation}\tag{3}\label{compatapp3}
        \exists v', [x|->a']t_f' \prog{*} v'
        \land (n-m,\ v,\ v') \in \calV{}{S_2}
      \end{equation}

      Recall that we are required to show that
      $\exists v'.\ [\phi_2]t_1'\,[\phi_2]t_2'\prog{*}v' \land (n-m,\ v,\ v') \in \calV{}{S_2}$.
      By (1), (2) and (3) and the fact that
      $(n-p-q,\ a,\ a') \in \calV{}{S_2} \implies a'\text{ is a value}$, we have:
      $\exists t_f',\,a',\,v'$ such that
      $[\phi_2]t_1'\ [\phi_2]t_2'
      \prog{*}
      (\lam x.t_f')\ [\phi_2]t_2'
      \prog{*}
      (\lam x.t_f)\ a'
      \prog{1}
      [x|->a']\,t_f
      \prog{*}
      v'
      \land (n-m,\ v,\ v') \in \calV{}{S_2}$
    \end{proof}
    
    \item $
    \left(c\ \langle \overline{x_{j} |-> S_{j}}^j \rangle -> l'\right) \in S
    \land
    \overline{\Gamma |- t_{j} \lesssim^{\fname{str}} t_{j}': S_{j}}^j
    \land \fname{FV}(l') \subseteq \left\{\overline{x_{j}}^j\right\}
    \\
    ~\hfill
    \implies
    \Gamma |- c\ \overline{t_{j}}^j \lesssim^{\fname{str}}
    \left(\lam \overline{x_{j}}^j.\ l'\right)\ \overline{t_{j}'}^j: S
    $
    \begin{proof}
      By the definition of $\lesssim^{\fname{str}}$, consider arbitrary
      $n \ge 0, \phi, \phi'$ such that $(n,\phi,\phi') \in \calG{}{\Gamma}$.
      We are required to show that
      $\left(n,\ 
      c\,\overline{[\phi]t_{j}}^j,\ 
      \left(\lam \overline{x_{j}}^j.\,[\phi']l'\right)\,
      \overline{[\phi']t_{j}'}^j
      \right)
      \in
      \calE{}{S}$.

      Since $\fname{FV}(l') \subseteq \left\{\overline{x_{j}}^j\right\}$,
      $[\phi']l' \equiv l'$, we are required to show that
      $$\left(n,\ 
      c\,\overline{[\phi]t_{j}}^j,\ 
      \left(\lam \overline{x_{j}}^j.\,l'\right)\,
      \overline{[\phi']t_{j}'}^j
      \right)
      \in
      \calE{}{S}$$

      By the definition of $\calE{}{S}$, consider arbitrary
      $m < n$ and $\overline{v_j}^j$ , such that 
      $(c\,\overline{[\phi]t_{j}}^j \prog{m} c\,\overline{v_j}^j
      \land \overline{\fname{irred}(v_j)}^j
      )$, we are then required to prove that
      there exists a $v'$ such that
      $\left(\lam \overline{x_{j}}^j.\,l'\right)\,
      \overline{[\phi']t_{j}'}^j
      \prog{*} v'
      \land
      (n-m,\ c\ \overline{v_j}^j,\ v') \in \calV{}{S}$

      By the operational semantics, for each $j$ there exists
      an $m_j$ such that $[\phi]t_{j} \prog{m_j} v_j \land \fname{irred}(v_j)$,
      and $\Sigma_j m_j = m$.

      By the definitions of $\lesssim^{\fname{str}}$,
      for each $j$, instantiate the premises
      $\Gamma |- t_{j} \lesssim^{\fname{str}} t_{j}': S_{j}$
      with $n\ge0,\ \phi$ and $\phi'$. We get for each $j$ that
      $(n,\ [\phi]t_{j},\ [\phi']t_{j}') \in \calE{}{S_{j}}$. By the definition
      of $\calE{}{S_{j}}$ we get that for each $j$,
      $\forall m < n.\ \forall v.\ ([\phi]t_{j} \prog{m} v \land \fname{irred}(v))
      \implies (\exists v'.\ [\phi']t_{j}' \prog{*} v'
      \land (n - m,\ v,\ v')\in \calV{}{S_{`1j}})$.
      For each $j$, instantiate this with $m_j$ as $m$ and $v_j$ as $v$.
      Note that the validity of the instantiation is easily verified.
      So we get that for each $j$, $\exists v_j'.\ [\phi']t_{j}' \prog{*} v_j'
      \land (n - m_j,\ v_j,\ v_j')\in \calV{}{S_j}$. Note that by the definition
      of $\calV{}{S_j}$, we have for each $j$, $v_j'$ is a value.

      So $\exists v' = \left[\overline{x_{j} |-> v_j'}^j\right]l'$,
      such that
      $
      \left(\lam \overline{x_{j}}^j.\,l'\right)\,
      \overline{[\phi']t_{j}'}^j
      \prog{*}
      v'$.
      
      We then also need to verify that
      $(n - m,\, c\ \overline{v_j}^j,\, v') \in \calV{}{S}$.
      $\fname{FV}(l') \subseteq \left\{\overline{x_{j}}^j\right\}$ directly
      follows from the premise. And since we have for each $j$,
      $(n - m_j,\ v_j,\ v_j') \in \calV{}{S_{j}}$ and $n - m_j > n - m$,
      by lemma \ref{downclosurestrat},
      $\forall m' < n - m.\ \land_j (m', v_j, v_j') \in \calV{}{S_{j}}$.

    \end{proof}
    
    \item $\left(c\, \langle\overline{S_j}^j\rangle\right) \in S
    \land
    \overline{\Gamma |- t_j \lesssim^{\fname{str}} t_j': S_j}^j
    \implies
    \Gamma |- c\ \overline{t_j}^j
    \lesssim^{\fname{str}}
    c\ \overline{t_j'}^j: S$
    \begin{proof}
      By the definition of $\lesssim^{\fname{str}}$,
      consider arbitrary $n \ge 0, \phi_1, \phi_2$
      such that $(n,\,\phi_1,\,\phi_2) \in \calG{}{\Gamma}$.
      We are required to show that
      $(n,\ 
      c\,\overline{[\phi_1]t_j}^j,\ 
      c\,\overline{[\phi_2]t_j'}^j)
      \in \calE{}{S}$.
      By the definition of $\calE{}{S}$, consider arbitrary
      $m < n$ and $v$ such that
      $c\,\overline{[\phi_1]t_j}^j \prog{m} v \land \fname{irred}(v)$,
      then we are required to show that
      $\exists v'.\ c\,\overline{[\phi_2]t_j'}^j \prog{*} v'
      \land
      (n-m,\,v,\,v') \in \calV{}{S}$.

      By operational semantics, there exists $\overline{m_j}^j$ such that
      for each $j$, $[\phi_1]t_j \prog{m_j} v_j \land \fname{irred}(v_j)$.
      Also, $\Sigma_j m_j = m$, which implies $\forall j.\ m_j < m$.

      At the same time, for each $j$, we instantiate the premises
      $\overline{\Gamma |- t_j \lesssim^{\fname{str}} t_j': S_j}^j$
      with $n, \phi_1, \phi_2$ and get that for each $j$,
      $(n,\,[\phi_1]t_j,\,[\phi_2]t_j') \in \calE{}{S_j}$.
      Then by the definition of $\calE{}{S_j}$, we get for each
      $j$,
      $\forall m_j' < n.\ 
      \forall v_j''.\ 
      [\phi_1]t_j \prog{m_j'} v_j'' \land \fname{irred}(v_j'')
      \implies
      (\exists v_j'. [\phi_2]t_j' \prog{*} v_j'
      \land
      (n - m_j',\,v_j'',\,v_j') \in \calV{}{S_j})$. Instantiate
      this with $m_j$ as $m_j'$ and $v_j$ as $v_j''$, we get that
      for each $j$, $\exists v_j'. [\phi_2]t_j' \prog{*} v_j'
      \land
      (n - m_j,\,v_j,\,v_j') \in \calV{}{S_j}$.

      So $\exists v' = c\,\overline{v_j'}^j$ such that
      $c\,\overline{[\phi_2]t_j'}^j \prog{*} v'$. And
      $(n - m,\,c\,\overline{[\phi_1]t_j}^j,\,v') \in \calV{}{S}$
      can be easily verified by the definition of
      $\calV{}{S}$ and lemma \ref{downclosurestrat} since
      $\forall j.\ m_j < m$, and thus $n - m < n - m_j$.
    \end{proof}

    \item $
      \Gamma |- t \lesssim^{\fname{str}} t':
      \left\{\overline{\left(
        c_i\,
        \langle\overline{x_{i,j} |-> S_{i,j}}^j\rangle ->
        l_i'
      \right)}^i\right\}
      \land
      \overline{
        \Gamma \cdot \overline{x_{i,j} |-> S_{i,j}}^j |-
        l_i \lesssim^{\fname{str}} l_i' : S
      }^i
      \land
      \\
      ~\hfill
      \overline{\fname{FV}(l_i) \cup \fname{FV}(l_i') \subseteq \left\{
        \overline{x_{i,j}}^j
      \right\}}^i
      \implies
      \Gamma |- \left(\kw{case}\ t\ \kw{of}\ \left(\,\overline{\,
        c_i\ \overline{x_{i,j}}^j -> l_i
      \,}^i\,\right)\right)
      \lesssim^{\fname{str}}
      t'
      : S
    $
    \begin{proof}
      For the conciseness of writing, we use $S_c$ to denote
      $\left\{\overline{\left(
        c_i\,
        \langle\overline{x_{i,j} |-> S_{i,j}}^j\rangle ->
        l_i'
      \right)}^i\right\}$ in the following proof.

      By the definition of $\lesssim^{\fname{str}}$, consider
      arbitrary $n \ge 0, \phi_1, \phi_2$ such that
      $(n, \phi_1, \phi_2) \in \calG{}{\Gamma}$. We are required
      to show that
      $$\left(
        n,\ 
        \left(\kw{case}\ [\phi_1]t\ \kw{of}\ \left(\,\overline{\,
          c_i\ \overline{x_{i,j}}^j -> [\phi_1]l_i
        \,}^i\,\right)\right),\ 
        [\phi_2]t'
      \right)
      \in
      \calE{}{S}$$
      
      We have for each $i$ that
      $\fname{FV}(l_i) \subseteq \left\{\overline{x_{i,j}}^j\right\}$,
      so $[\phi_1]l_i \equiv l_i$, and
      by the definition of $\calE{}{S}$, consider arbitrary
      $m < n, v$ such that
      \begin{equation}\tag{1}
        \left(\kw{case}\ [\phi_1]t\ \kw{of}\ \left(\,\overline{\,
          c_i\ \overline{x_{i,j}}^j -> l_i
        \,}^i\,\right)\right)
        \prog{m}
        v
        \land
        \fname{irred}(v)
      \end{equation}
      We are then required to show that
      $\exists v'.\ [\phi_2]t' \prog{*} v'
      \land (n - m,\ v,\ v') \in \calV{}{S}$

      By (1) and the operational semantics,
      there exists $p < m$ and $v_c$ such that
      $[\phi_1]t \prog{p} v_c \land \fname{irred}(v_c)$.

      Then instantiate
      $\Gamma |- t \lesssim^{\fname{str}} t': S_c$
      with $n,\phi_1,\phi_2$,
      we get that
      $(n,\,[\phi_1]t,\,[\phi_2]t') \in \calE{}{S_c}$. By the
      definition of $\calE{}{S_c}$, we instantiate this again with
      $p$ and $v_c$ and get
      \begin{equation*}\tag{2}
        \exists v_c'.\ [\phi_2]t' \prog{*} v_c'
        \land (n-p,\ v_c,\ v_c') \in \calV{}{S_c}
      \end{equation*}

      By the definition of $\calV{}{S_c}$, we get that
      there exists an $i$, $\overline{v_{i,j}}^j$ and 
      $\overline{v_{i,j}'}^j$ such that
      $v_c \equiv c_i\ \overline{v_{i,j}}^j$,
      $v_c' \equiv 
      \left[\overline{x_{i,j} |-> v_{i,j}'}^j\right]l_i'$,
      $\forall m' < n - p$, $\forall j.$
      $(m',\ v_{i,j},\ v_{i,j}') \in \calV{}{S_{i,j}}$,
      and by \Lem{downclosurestrat}, we have
      \begin{equation*}\tag{3}
      \left(n - p - 1, \left[\overline{x_{i,j} |-> v_{i,j}}^j\right] l_i, \left[\overline{x_{i,j} |-> v_{i,j}'}^j\right] l_i'\right) \in \calE{}{S}
      \end{equation*}
      Then the conclusion holds since $l_i$ and $l_i'$ are lambda abstractions.

    \end{proof}

    \item $
      \Gamma |- t \lesssim^{\fname{str}} t':
      \left\{\overline{\left(
        c_i\,\langle \overline{S_{i,j}}^j \rangle
      \right)}^i\right\}
      \land
      \overline{\,
        \Gamma \cat \overline{(x_{i,j} |-> S_{i,j})}^j |- t_i \lesssim^{\fname{str}} t_i': S
      \,}^i
      \implies
      \\
      ~\hfill\Gamma |-
      \left(\kw{case}\ t\ \kw{of}\ \left(\,\overline{\,
        c_i\ \overline{x_{i,j}}^j -> t_i
      \,}^i\,\right)\right)
      \lesssim^{\fname{str}}
      \left(\kw{case}\ t'\ \kw{of}\ \left(\,\overline{\,
        c_i\ \overline{x_{i,j}}^j -> t_i'
      \,}^i\,\right)\right) : S
    $
    \begin{proof}
      \newcommand{\CASE}[2]{\kw{CASE}\langle#1,\,\overline{#2_i}^i\rangle}
      For the conciseness of writing, we define a shorthand
      $\CASE{s}{b}$ with two parameters
      $s$ and $\overline{b_i}^i$ to denote the case terms
      involed in the proof:
      $\kw{case } s \kw{ of }
      \left(
        \overline{c_i\ \overline{x_{i,j}}^j -> b_i}^i
      \right)$. And we use $S_c$ to denote
      the strategy for the scrutinee:
      $\left\{\overline{\left(
        c_i\,\langle \overline{S_{i,j}}^j \rangle
      \right)}^i\right\}$.

      By the definition of $\lesssim^{\fname{str}}$,
      consider arbitrary $n \ge 0,\phi_1,\phi_2$ such
      that $(n, \phi_1, \phi_2) \in \calG{}{\Gamma}$.
      We are required to prove that
      $\left(n,\ 
      \CASE{[\phi_1]t}{[\phi_1]t},\ 
      \CASE{[\phi_2]t'}{[\phi_2]t'} \right)
      \in
      \calE{}{S}$. Then by the definition of
      $\calE{}{S}$, consider arbitrary $m < n$ and 
      $v$ such that
      $\CASE{[\phi_1]t}{[\phi_1]t} \prog{m} v
      \land
      \fname{irred}(v)$. We are now required to prove
      that
      $\exists v'.\ 
      \CASE{[\phi_2]t'}{[\phi_2]t'} \prog{*} v'
      \land (n - m,\,v,\,v') \in \calV{}{S}$.

      By the operational semantics and
      $\CASE{[\phi_1]t}{[\phi_1]t} \prog{m} v
      \land \fname{irred}(v)$,
      we have that there exists a $p < m$ and
      $v_c$ such that
      $[\phi_1]t \prog{p} v_c \land \fname{irred}(v_c)$.

      Instantiate the first premise with $n>0$, $\phi_1$
      and $\phi_2$ and we get
      $\left(n, [\phi_1]t, [\phi_2]t'\right)
      \in \calE{}{S_c}$. Then by the definition of
      $\calE{}{S_c}$, we instantiate this with $p<n$
      and $v_c$ and get
      \begin{equation*}\tag{1}
        \exists v_c'.\ 
        [\phi_2]t' \prog{*} v_c' \land
        (n - p,\,v_c,\,v_c') \in \calV{}{S_c}
      \end{equation*}
      Then by the definition of $\calV{}{S_c}$,
      there exist an $i$, $\overline{v_{i,j}}^j$
      and $\overline{v_{i,j}'}^j$ such that
      $v_c \equiv c_i\ \overline{v_{i,j}}^j$,
      $v_c' \equiv c_i\ \overline{v_{i,j}'}^j$ and
      $\forall m' < n - p.\ \forall j.\ 
      (m', v_{i,j}, v_{i,j}') \in \calV{}{S_{i,j}}$.

      By the operational semantics, we now have
      \begin{align*}\tag{2}
        \CASE{[\phi_1]t}{[\phi_1]t}
        \prog{p}
        \CASE{c_i\ \overline{v_{i,j}}^j}{[\phi_1]t}
        &\prog{1}
        [\overline{x_{i,j} |-> v_{i,j}}^j][\phi_1]t
        \\
        &\prog{m - p - 1}
        v
        \land
        \fname{irred}(v)
      \end{align*}
      
      Then by definition of $\lesssim^{\fname{str}}$,
      we instantiate the premise
      $\Gamma \cat \overline{(x_{i,j} |-> S_{i,j})}^j
      |-
      t_i \lesssim^{\fname{str}} t_i': S$ with
      $n - p - 1$,
      $\phi_1\cdot\overline{x_{i,j}|->v_{i,j}}^j$
      and
      $\phi_2\cdot\overline{x_{i,j}|->v_{i,j}'}^j$.
      Note that this instantiation is valid because
      \begin{enumerate}
        \item $n - p - 1 \ge 0$, and
        \item $\begin{aligned}[t]
          &(n,\,\phi_1,\,\phi_2) \in \calG{}{\Gamma}
          \land
          \overline{(n - p - 1,\ v_{i,j},\ v_{i,j}') \in \calV{}{S_{i,j}}}^j
          \implies
          \\
          &(n-p-1,\ 
          \phi_1\cdot\overline{x_{i,j} |-> v_{i,j}}^j\ 
          \phi_2\cdot\overline{x_{i,j} |-> v_{i,j}'}^j) \in
          \calG{}{\Gamma \cdot \overline{x_{i,j} |-> S_{i,j}}^j}
          \text{ by lemma \ref{downclosurestrat}}
        \end{aligned}$
      \end{enumerate}
      So we have $(n-p-1,\ 
      \left[\overline{x_{i,j} |-> v_{i,j}}^j\right][\phi_1]t_i,\ 
      \left[\overline{x_{i,j} |-> v_{i,j}'}^j\right][\phi_2]t_i'
      ) \in \calE{}{S}$. 
      Then by (2) and the definition of $\calE{}{S}$,
      we instantiate this with $m - p - 1 < n - p - 1$
      and $v$, which can be reduced from
      $\left[\overline{x_{i,j} |-> v_{i,j}}^j\right][\phi_1]t_i$
      with $m - p - 1$ steps,
      and get
      \begin{equation*}\tag{3}
        \exists v'.\ 
        \left[\overline{x_{i,j} |-> v_{i,j}'}^j\right][\phi_2]t_i'
        \prog{*}
        v'
        \land
        (n - p - 1 - (m - p - 1),\ v,\ v') \in \calV{}{S}
      \end{equation*}
      
      Recall that we need to prove
      $\exists v'.\ 
      \CASE{[\phi_2]t'}{[\phi_2]t'} \prog{*} v'
      \land (n - m,\,v,\,v') \in \calV{}{S}$.
      So by (1) and (3),
      $\exists v'.\ 
      \CASE{[\phi_2]t'}{[\phi_2]t'}
      \prog{*}
      \CASE{c_i\ \overline{v_{i,j}'}^j}{[\phi_2]t'}
      \prog{1}
      \left[\overline{x_{i,j} |-> v_{i,j}'}^j\right][\phi_2]t_i'
      \prog{*}
      v'
      \land
      (n - m,\,v,\,v') \in \calV{}{S}$.
    \end{proof}
  
    \item $\Gamma |- t \lesssim^{\fname{str}} t' : I \implies \Gamma |- t \lesssim^{\fname{str}} t' : \top$
    \begin{proof}
      By the definition of $\lesssim^{\fname{str}}$, consider arbitrary $n \ge 0$ and $\phi$, $\phi'$
      such that $(n,\ \phi,\ \phi' \in \calG{}{\Gamma})$, now we are required to show that
      $(n,\ \phi[t],\ \phi[t']) \in \calV{}{\top}$. Instantiate the premise with
      $n, \phi$, and $\phi'$, then we have that $(n,\ \phi[t],\ \phi[t']) \in \calV{}{I}$.
      Then since for all strategy $S$ that is not $\top$,
      the definition of $\calV{}{I}$ contains triples
      of an integer and two values, according to the definition of $\calV{}{\top}$,
      it can be shown that $(n,\ \phi[t],\ \phi[t']) \in \calV{}{\top}$.
    \end{proof}
  \end{itemize}

\end{lemma}

\section{From the Algorithmic Fusion Rules to the Prescient Fusion Rules}

\label{sec:pseudo-algo-to-prescient-full}

We now introduce
two functions $\calM$ and $\calT$ which helps to transform the pseudo-algorithmic fusion derivations
to the fusion derivations we introduced in \Fig{fig:fusion-prescient}.

The transformation
is rather straightforward. Notice that the only difference between the pseudo-algorithmic fusion rules
and the prescient fusion rules is located in the \textsc{P-Ctor} and the \textsc{P-Case} rules:
the pseudo-algorithmic rules transform the terms of matching branches
(book-kept in the fusion strategies) in the rule \textsc{P-Ctor},
while the prescient \ruleName{F-Case} rule checks that
the terms attached in the prescient strategies indeed conform to the transformed terms.

\begin{definition}
  The function $\calM$ traverses pseudo-algorithmic fusion derivation trees to
  accumulate a mapping from \emph{constructor fusion strategies} to their corresponding \emph{transformed} terms:
  The output of $\calM$ is a mapping from constructor fusion strategies to
  another mapping from natural numbers to terms -- since a constructor fusion strategy
  has more than one $\kw{case}$ branch terms attached to it, the inner
  mapping maps from the branch ids to their corresponding terms.
\begin{itemize}
  \item
  $\calM[| \inferrule[P-Var]{(x |-> S) \in \Gamma}{\Gamma |- x ~>_{} x <= S} |] = \ep$
  
  \item
  $\calM[| \inferrule[P-Lam]{\Gamma \cat (x |-> S') |- t ~> t' <= S}{\Gamma |- \lam x.t ~> \lam x.t' <= S' -> S} |]
  =
  \calM[|\Gamma \cat (x |-> S') |- t ~> t' <= S|]$
  
  \item
  $\calM[| \inferrule[P-Ctor]{
    \left( c_k\,\langle\,\overline{x_{k,j} |-> S_{k,j}}^j\,\rangle->l_k \right): S_{\fname{res}} \in S_{\mathcal{A}}
    \\
    \overline{\Gamma |- t_{k,j} ~> t_{k,j}' <= S_{k,j}}^j
    \\
    \Gamma\cat
      \overline{
        (x_{k,j} |-> S_{k,j})
      }^j
      |-
      l_k ~> l_k' <= S_{\fname{res}}
    \\
    \fname{FV}(l_k) \cup \fname{FV}(l_k') \subseteq \left\{\overline{x_{k,j}}^j\right\}
  }{
    \Gamma |-
    c_k\ \overline{t_{k,j}}^j ~>
    \kw{let } \overline{x_{k,j} = t_{k,j}'}^j \kw{ in } l_k'
    <= S_{\mathcal{A}}
  } |] = \\
  ~\hfill \begin{aligned}
    &\calM[|\Gamma\cat\overline{(x_{k,j} |-> S_{k,j})}^j |- t_k ~> t_k' <= S_{\fname{res}}|] \\
    &\cup_j \calM[|\Gamma |- t_{k,j} ~> t_{k,j}' <= S_{k,j}|] \\
    &\cup \left\{S_{\mathcal{A}} |-> \left\{k |-> l_k'\right\}\right\}
  \end{aligned}$

  \item
  $\calM[|\inferrule[P-Case]{
    S_{\mathcal{A}} = \left\{\overline{c_i\,\langle \overline{x_{i,j} |-> S_{i,j}}^j \rangle->l_i}^i \right\}: S_{\fname{res}}
    \\
    \Gamma |- t ~> t' <= S_{\mathcal{A}}
    \\\\
    \overline{\Gamma\cat
      \overline{
        (x_{i,j} |-> S_{i,j})
      }^j
      |-
      l_i ~> l_i' <= S_{\fname{res}}}^i
    \\
    \overline{\fname{FV}(l_i) \cup \fname{FV}(l_i') \subseteq \left\{ \overline{x_{i,j}}^j \right\}}^i
  }{
    \Gamma |-
    \kw{case}\ t\ \kw{of}\ \overline{\,
      c_i\ \overline{x_{i,j}}^j -> t_i
    \,}^i
    ~>
    t'
    <=
    S_{\fname{res}}
  }|] =
  \\
  ~\hfill \begin{aligned}
    &\cup_i \calM[|\Gamma\cat \overline{ (x_{i,j} |-> S_{i,j})}^j |- l_i ~> l_i' <= S_{\fname{res}}|] \\
    &\cup \calM[|\Gamma |- t ~> t' <= S_{\mathcal{A}}|] \\
    &\cup \left\{S_{\mathcal{A}} |-> \left\{\overline{i |-> l_i'}^i\right\}\right\}
  \end{aligned}
  $
  \item \textcolor{gray}{All the other rules just traverse sub-derivation trees and accumulate substitutions}
\end{itemize}
\end{definition}

\begin{definition}
  The function $\calT_\phi$ traverses pseudo-algorithmic fusion derivation trees again to
  transform them into prescient fusion derivation trees according to the mapping $\phi$
  that is accumulated in the previous traversal of pseudo-algorithmic fusion derivation trees
  by $\calM$.

  There are several notations related to how the mapping $\phi$ is used:
  \begin{itemize}
    \item $\phi[S_{\mathcal{A}}][k]$ means the $k^{\text{th}}$ \emph{transformed} branch body term recorded in the mapping $\phi$ of the strategy $S_{\mathcal{A}}$.
    \item $\phi(S_{\mathcal{A}})$ means to apply the mapping to a pseudo-algorithmic strategy, this is similar to applying a substitution:
          \\
          \hspace*{10em}
          $\begin{aligned}
            \phi(\top) &= \top \\
            \phi(\bot) &= \bot \\
            \phi(S_1 -> S_2) &= \phi(S_1) -> \phi(S_2) \\
            \phi(\mu F_S) &= \mu (\phi(F_S)) \\
            \phi\left(\left\{\overline{c \langle \overline{S} \rangle}\right\}\right) &= \left\{\overline{c \langle \overline{\phi(S)} \rangle}\right\} \\
            \phi\left(S_{\mathcal{A}}\right) &= \left\{\overline{
              c_i\,\langle\overline{x_{i,j} |-> \phi(S_{i,j})}^j\rangle->\phi[S_{\mathcal{A}}][i]
            }^i\right\} \\
            &\text{where } S_{\calA} = \left\{\overline{
              c_i\,\langle\overline{x_{i,j} |-> S_{i,j}}^j\rangle->l_i
            }^i\right\}: S
          \end{aligned}$
  \end{itemize}

  Then the actual transformation $\mathcal{T}_{\phi}$ is defined as:
  \begin{itemize}
    \item 
    $\calT_\phi[|\inferrule[P-Ctor]{
      \left( c_k\,\langle\,\overline{x_{k,j} |-> S_{k,j}}^j\,\rangle->l_k \right): S_{\fname{res}} \in S_{\mathcal{A}}
      \\
      \overline{\Gamma |- t_{k,j} ~> t_{k,j}' <= S_{k,j}}^j
      \\\\
      \Gamma\cat
        \overline{
          (x_{k,j} |-> S_{k,j})
        }^j
        |-
        l_k ~> l_k' <= S_{\fname{res}}
      \\
      \fname{FV}(t_k) \cup \fname{FV}(t_k') \subseteq \left\{\overline{x_{k,j}}^j\right\}
    }{
      \Gamma |-
      c_k\ \overline{t_{k,j}}^j ~>
      \kw{let } \overline{x_{k,j} = t_{k,j}'}^j \kw{ in } l_k'
      <= S_{\mathcal{A}}
    }|] = \\
    ~\hfill \inferrule[]{
      \fname{FV}(\phi[S_{\mathcal{A}}][k]) \subseteq \left\{\overline{x_{k,j}}^j\right\}
      \\
      \hspace{-1em}\left( c_k\,\langle\,\overline{x_{k,j} |-> \phi(S_{k,j})}^j\,\rangle->\phi[S_{\mathcal{A}}][k] \right) \in \phi(S_{\mathcal{A}})
      \\
      \hspace{-1em}\overline{\calT_\phi[|\Gamma |- t_{k,j} ~> t_{k,j'} <= S_{k,j}|]}^j
    }{
      \phi(\Gamma) |-
      c_k\ \overline{t_{k,j}}^j ~>
      \kw{let } \overline{x_{k,j} = t_{k,j}'}^j \kw{ in } \phi[S_{\mathcal{A}}][k]
      <= \phi(S_{\mathcal{A}})
    }$

    \item \textcolor{gray}{All the other cases just recursively traverse sub-derivation trees without changing the shape of rules.}
  \end{itemize}
\end{definition}

With the definition of $\calM$ and $\calT$, we have the following \Thm{thm:algo-to-precsient}.
\begin{theorem}
  \label{thm:algo-to-precsient}
  if $\Gamma |- t ~>_{\mathcal{A}} t' <= S_{\mathcal{A}}$,
  then for all $\phi$, such that
  $\calM [|\Gamma |- t ~>_{\mathcal{A}} t' <= S_{\mathcal{A}} |] \subseteq \phi$,
  we have
  $\calT_{\phi} [|\Gamma |- t ~>_{\mathcal{A}} t' <= S_{\mathcal{A}}|] = \phi(\Gamma) |- t ~> t' <= \phi(S_{\mathcal{A}})$.
\end{theorem}

\begin{proof}[Proof of \Thm{thm:algo-to-precsient}]
  By induction on $\Gamma |- t \squigtoa t' <= S_\calA$.
  \begin{description}
    \item[Case \ruleName{P-Ctor}.]
    Let $\phi = \calM[|\Gamma |- t \squigtoa t' <= S_{\calA}|]$,
    then by the definition of $\calM$, for all $j$ we have
    $\calM[|\Gamma |- t_{k,j} \squigtoa t_{k,j}' <= S_{k,j}|] \subseteq \phi$.
    By IH, for all $j$ we have
    $\calT_{\phi}[|\Gamma |- t_{k,j} \squigtoa t_{k,j}' <= S_{k,j}|] = \phi(\Gamma) |- t_{k,j} ~> t'_{k,j} <= \phi(S_{k,j})$.
    Also, by definition of $\calM$ %
    we have $\left\{S_{\calA} |-> \left\{k |-> l_k'\right\}\right\} \in \phi$;
    and by the premise, $\fname{FV}(l_k') \subseteq \left\{\overline{x_{k,j}}^j\right\}$,
    so we can verify that $\fname{FV}(\phi[S_{\calA}][k]) \subseteq \left\{\overline{x_{k,j}}^j\right\}$.
    And by the definition of $\phi(S_{\calA})$, we can verify that
    $\left( c_k\,\langle\,\overline{x_{k,j} |-> \phi(S_{k,j})}^j\,\rangle->\phi[S_{\mathcal{A}}][k] \right) \in \phi(S_{\mathcal{A}})$.

    \item[Case \ruleName{P-Case}.] Sub-derivations follow directly from IH.
    And the free variable check directly follows from the assumption.

    \item[Other cases.] By direct IH.
  \end{description}
\end{proof}

\section{Proof for Inference System}
\label{app:inference-proofs}

To prove \Lem{lem:constr-implies-cons}, we first want to express the fact that
$\ep |> \SCtx\ \fname{outputs}\ \CCtx$
produces a context $\CCtx$ of \emph{subtype-consistent} bounds,
which can be defined via an auxiliary \emph{subtyping} relation,
shown in \Fig{fig:subtyping-rules}.
In that relation, we now have \emph{two} bounds context
on the left of the $|-$.
The first one is the usual bounds context,
containing known bounds on type variables which can be leveraged
to show subtyping,
and the other is a \emph{delayed} bounds context,
containing those assumptions %
to be leveraged later as the derivation
goes through at least a function type or a constructor type
-- notice that in \ruleName{S-Fun} and \ruleName{S-Ctor},
the premises \emph{move} the delayed context into the main context.
This more elegantly achieves what \citet{Parreaux20:simple-essence-alg-subt}
achieves with the $\rhd$ modality marker.

\begin{figure*}
\mprset{sep=1.3em}
{\small\input{contents/subtyping-rules.tex}}
\caption{Subtyping rules for fusion strategy candidates.
}
\label{fig:subtyping-rules}
\end{figure*}

\begin{definition}[Fully merged]
  $\CDtx$ is \emph{fully merged} if 
  \label{def:fully-merged}
  \begin{enumerate}
    \item For all $\al =< \be \in \CDtx$, we have $\be =< \al \in \CDtx$, and
    \item For all $\al,\be_1,\be_2,\ga_1,\ga_2$,
    if $\al =< \be_1 -> \be_2 \in \CDtx$ and $\al =< \ga_1 -> \ga_2 \in \CDtx$,
    then we have $\be_1 =< \ga_1 \in \CDtx$ and $\be_2 =< \ga_2 \in \CDtx$
    \item For all $\al,\overline{\be_{i,j}}^{i,j},\overline{\ga_{i,j}}^{i,j},\al_{\fname{res}},\al_{\fname{res}}',
    \overline{x}^{i,j},\overline{y}^{i,j},\overline{t_i}^i$ and $\overline{e_i}^i$,
    if $\left(\al =< \left\{\,\overline{
      c_i\langle\overline{x_{i,j} |-> \be_{i,j}}^{j}\rangle -> t_i
    }^{i}\right\}: \al_{\fname{res}}\right) \in \CDtx$
    and
    $\left(\al =< \left\{\overline{
      c_i\langle\overline{y_{i,j} |-> \ga_{i,j}}^{j}\rangle -> e_i
    }^{i}\right\}: \al'_{\fname{res}}\right) \in \CDtx$,
    then $\overline{\be_{i,j} =< \ga_{i,j} \in \CDtx}^{i,j}$.
  \end{enumerate}
  \NOTE{this property does not need to be parametrized,
  because they are irrelevant to how we put things from the worklist to the left or right of the ctx;
  in other words, the C-Mrg rules will handle all the ctx, regardless of whether they are garbage or not
  }
\end{definition}

\newcommand{\disjunion}{\mathbin{\mathaccent\cdot\cup}}

\textbf{Notation:}
We write $\CDtx |- \ty\Pos=<\ty\Neg$ as a shorthand for
$\CDtx ~|~ \ep |- \ty\Pos =< \ty\Neg$ and
$\CDtx |- \overline{\ty\Pos=<\ty\Neg}$
as a shorthand for $\overline{\CDtx ~|~ \ep |- \ty\Pos=<\ty\Neg}$.
We use $\CCtx \disjunion \CDtx$ to denote the \textit{disjoint union}
between two bounds contexts $\CCtx$ and $\CDtx$.

\begin{definition}[Bounds consistency]
  \label{def:bounds-consistency}
  $\CDtx$ is \emph{bounds-consistent} if
  for all $\al$, $\CDtx_{0}$, $\ty\Pos$ and $\ty\Neg$
  such that $\ntv{\ty\Neg}$,
  $\CDtx = \CDtx_0 \disjunion \{(\al =< \ty\Neg), (\ty\Pos =< \al)\}$:
  \begin{description}
    \item[If $\ntv{\ty\Pos}$]

    We have
    $\ep \mid \CDtx |- \ty\Pos=<\ty\Neg$.
    \NOTE{guard all the $\CDtx$}

    \item[If $\ty\Pos$ is a type variable $\be \neq \al$.] %
    We have
    $(\be=<\ty\Neg) \in \CDtx_0$.%
    \NOTE{all upper bounds of $\al$ are also upper bounds of $\be$}

  \end{description}

\end{definition}

\newcommand{\proc}[1]{\fname{processed}(#1)}
\newcommand{\unproc}[1]{\fname{unprocessed}(#1)}

\newcommand{\openbr}{\textcolor{gray}{\{\,}}
\newcommand{\closebr}{\textcolor{gray}{\,\}}}

\begin{lemma}
  \label{lem:constr-collect-self-contained}
  If $\Gamma |- t:\al \conR \SCtx$, then $\SCtx$ is self-contained
  (according to \Def{def:self-contained}).
\end{lemma}

\begin{proof}[Proof of \Lem{lem:constr-collect-self-contained}]
  By induction on the fusion strategy candidate inference rules.
  \begin{description}
    \item[Case \ruleName{I-Var}, \ruleName{I-Lam}, \ruleName{I-App}, \ruleName{I-Ctor}, \ruleName{I-LetRec}.]
    Direct by induction hypothesis since they do not introduce constraints
    containing $\left\{\,\overline{c_i\,\langle\,\overline{x_{i,j} |-> \be_{i,j}}^j\,\rangle -> l_i}^i\,\right\}: \de$
    \item[Case \ruleName{I-Case}.]
    By IH, $\SCtx_0$ and $\overline{\SDtx_i}^i$ are self-contained.
    Since we have $\overline{\Gamma \catsp \overline{(x_{i,j} |-> \be_{i,j})}^j |- l_i:\ga_i \conR \SDtx_i}^i$
    in the premise, the output $\SCtx$ containing all the constraints is self-contained.
  \end{description}
\end{proof}

\begin{lemma}
  \label{lem:constr-output-contains-constr-ctx}
  If $\CCtx |> \SCtx$ outputs $\CDtx$,
  then for all $\ty\Pos =< \ty\Neg \in \CCtx$,
  we have $\ty\Pos =< \ty\Neg \in \CDtx$,
  \ie
  $\CCtx \subseteq \CDtx$.
\end{lemma}

\begin{proof}[Proof of \Lem{lem:constr-output-contains-constr-ctx}]
  By induction on the constraint solving rules.
  \begin{description}
    \item[Case \ruleName{C-Skip}, \ruleName{C-Pass}, \ruleName{C-Match}, \ruleName{C-Fun}, \ruleName{C-Ctor}, \ruleName{C-VarMrg}, \ruleName{C-FunMrg}, \ruleName{C-CtorMrg}.]
    By direct IH because the bounds contexts $\CCtx$ are not changed in these cases.
    \item[Case \ruleName{C-Done}.] This case holds because $\CDtx = \CCtx = \CCtx\Neg \catsp \CCtx^v \catsp \CCtx\Pos$.
    \item[Case \ruleName{C-VarL}.] By IH $\CCtx \catsp (\al =< \ty\Neg) \subseteq \CDtx$, so $\CCtx \subseteq \CDtx$.
    \item[Case \ruleName{C-VarR}.] By IH $(\ty\Pos =< \al) \catsp \CCtx \subseteq \CDtx$, so $\CCtx \subseteq \CDtx$.
  \end{description}
\end{proof}

\begin{lemma}
  \label{lem:constr-implies-subtype}
  If $\CCtx |> \SCtx$ outputs $\CDtx$,
  then $\CDtx |- \SCtx$ and
  for all $(\ty\Pos =< \ty\Neg) \in \SCtx$,
  if $\ntv{\ty\Pos}$ and $\ntv{\ty\Neg}$, then
  $\ep ~|~ \CDtx |- \ty\Pos =< \ty\Neg$; otherwise
  $\CDtx |- \ty\Pos =< \ty\Neg$.
\end{lemma}

\begin{proof}[Proof of \Lem{lem:constr-implies-subtype}]
  By induction on the constraint solving rules.
  \begin{description}    
    \item[Case \ruleName{C-Skip}.]
    Since $\ty\Pos =< \ty\Neg \in \CCtx$, at least one of $\ty\Pos$ or $\ty\Neg$ is
    a type variable, we need to prove $\CDtx |- \ty\Pos =< \ty\Neg$.
    By \Lem{lem:constr-output-contains-constr-ctx}
    we have $\ty\Pos =< \ty\Neg \in \CDtx$,
    so $\CDtx |- \ty\Pos =< \ty\Neg$ by \ruleName{S-Hyp}.

    \item[Case \ruleName{C-VarL}.]
    We need to prove $\CDtx |- \al =< \ty\Neg$.
    By \Lem{lem:constr-output-contains-constr-ctx}
    we have $\al =< \ty\Neg \in \CDtx$,
    so $\CDtx |- \al =< \ty\Neg$ by \ruleName{S-Hyp}.

    \item[Case \ruleName{C-VarR}.]
    We need to prove $\CDtx |- \ty\Pos =< \al$.
    By \Lem{lem:constr-output-contains-constr-ctx}
    we have $\ty\Pos =< \al \in \CDtx$,
    so $\CDtx |- \ty\Pos =< \al$ by \ruleName{S-Hyp}.

    \item[Case \ruleName{C-Fun}.]
    We need to prove $\ep ~|~ \CDtx |- \al_1 -> \al_2 =< \be_1 -> \be_2$. By IH we have
    $\CDtx |- \be_1 =< \al_1$ and
    $\CDtx |- \al_2 =< \be_2$.
    By \ruleName{S-Fun}, $\ep ~|~ \CDtx |- \al_1 -> \al_2 =< \be_1 -> \be_2$.

    \item[Case \ruleName{C-Ctor}.]
    By IH we have
    $\forall j \in M_n.\ \CDtx |- \ga_j =< \al_{n,j}$.
    Then by \ruleName{S-Ctor}, we have 
    $\ep ~|~ \CDtx |- \left(c_n\ \langle\overline{\ga_j}^{j \in M_n}\rangle =<
    \left\{\,\overline{
      c_i\,\langle\,\overline{x |-> \al_{i,j}}^{j \in M_i}\,\rangle -> l_i
    }^{i\in N}\,\right\}: \be\right)$.

    \item[Case \ruleName{C-Pass}, \ruleName{C-Match}, \ruleName{C-VarMrg}, \ruleName{C-FunMrg}, \ruleName{C-CtorMrg}, \ruleName{C-Done}.] $\CDtx |- \ep$ holds.
  \end{description}
\end{proof}

\begin{lemma}
  \label{lem:constr-finally-solves}
  For any $\ty\Pos$, $\ty\Neg$, $\CCtx_{0,1,2}$ and $\SCtx$, where
  $\ty\Pos$ and $\ty\Neg$ are not both type variables,
  if $\CCtx_0 \cat (\ty\Pos =< \al) \cat \CCtx_1 \cat (\al =< \ty\Neg) \cat \CCtx_2 |> \SCtx$ outputs $\CDtx$,
  then there exists a sub-derivation
  $\CCtx_3 |> \ty\Pos =< \ty\Neg$ for some $\CCtx_3$.
\end{lemma}

\begin{proof}[Proof of \Lem{lem:constr-finally-solves}]
  By induction on $\CCtx_0 \cat (\ty\Pos =< \al) \cat \CCtx_1 \cat (\al =< \ty\Neg) \cat \CCtx_2 |> \SCtx$ outputs $\CDtx$.
  \begin{description}
    \item[Case \ruleName{C-Match}.] The immediate sub-derivation is of the form $\CCtx_3 |> \ty\Pos =< \ty\Neg$. 
    \item[Otherwise.]
    Their immediate sub-derivations are still of the form
    $\CCtx'_0 \cat (\ty\Pos =< \al) \cat \CCtx'_1 \cat (\al =< \ty\Neg) \cat \CCtx'_2 |> \SCtx'$,
    so the conclusion holds by induction hypothesis.
  \end{description}
\end{proof}

\begin{lemma}
  \label{lem:constr-add-bounds}
  If $\CCtx |> \SCtx$ outputs $\CDtx$ and
  $(\ty\Pos =< \ty\Neg) \in \CDtx$ but $(\ty\Pos =< \ty\Neg) \notin \CCtx$,
  then there exists $\CCtx'$ and $\SCtx'$ and a sub-derivation
  $\CCtx' |> (\ty\Pos =< \ty\Neg) \catsp \SCtx'$
  such that
  $(\ty\Pos =< \ty\Neg) \notin \CCtx'$
  and $\CCtx \subseteq \CCtx'$.
  
\end{lemma}
\begin{proof}[Proof of \Lem{lem:constr-add-bounds}]
  By induction on $\CCtx |> \SCtx$ outputs $\CDtx$.
  \begin{description}
    \item[Case \ruleName{C-VarR}]
    If $(\ty\Pos =< \ty\Neg)$ does not equal to $(\ty\Pos =< \al)$, then by induction hypothesis,
    there exists $\CCtx'$ and $\SCtx'$ and a sub-derivation
    $\CCtx' |> (\ty\Pos =< \ty\Neg) \catsp \SCtx'$ such that
    $(\ty\Pos =< \ty\Neg) \notin \CCtx'$ and $((\ty\Pos =< \al) \catsp \CCtx) \subseteq \CCtx'$.
    By $((\ty\Pos =< \al) \catsp \CCtx) \subseteq \CCtx'$, we also have
    $\CCtx \subseteq \CCtx'$, so the conclusion holds.
    If $(\ty\Pos =< \ty\Neg) = (\ty\Pos =< \al)$ then the
    sub-derivation is the current derivation itself and the conclusion holds.
    \item[Case \ruleName{C-VarL}.]
    Similar to case \ruleName{C-VarR}.
    
    \item[Case \ruleName{C-Done}.]
    Does not apply.
    
    \item[Otherwise.]
    By induction hypothesis because the content of $\CCtx$ does not change in their premises.
  \end{description}
\end{proof}

\begin{lemma}
  \label{lem:constr-implies-fully-merged}
  If $\CCtx |> \SCtx$ outputs $\CDtx$ then $\CDtx$ is fully merged.
\end{lemma}
\begin{proof}[Proof of \Lem{lem:constr-implies-fully-merged}]
  As explained in \Sec{constr-solving-exhaustive},
  a rule is considered applicable only if none of the rules before it is applicable.
  So by the time \ruleName{C-Done} is used,
  rules \ruleName{C-VarMrg}, \ruleName{C-FunMrg} and \ruleName{C-CtorMrg}
  have been applied exhaustively and the output $\CDtx$ is fully merged.
\end{proof}

\begin{lemma}
  \label{lem:constr-implies-bounds-cons-basic-2}
  If $\CCtx |> \SCtx$ outputs $\CDtx$,
  then for arbitrary $\CCtx_{0,1,2}, \CDtx_0, \al, \ty\Pos, \ty\Neg$ such that
  $\CDtx = \CDtx_0 \disjunion \{(\al =< \ty\Neg),\ (\ty\Pos =< \al)\}$,
  $\ntv{\ty\Neg}$ and
  $\CCtx = \CCtx_0 \catsp (\ty\Pos =< \al) \catsp \CCtx_1 \catsp (\al =< \ty\Neg) \catsp \CCtx_2$,
  if $\ntv{\ty\Pos}$ then $\ep ~|~ \CDtx |- \ty\Pos =< \ty\Neg$; or
  if $\ty\Pos$ is a type variable $\be$ then $(\be =< \ty\Neg) \in \CDtx_0$.
\end{lemma}

\begin{proof}[Proof of \Lem{lem:constr-implies-bounds-cons-basic-2}]
  By \Lem{lem:constr-finally-solves}, there exists a sub-derivation
  $\CCtx_3 |> \ty\Pos =< \ty\Neg$ for some $\CCtx_3$.
  \begin{description}
    \item[Case $\ntv{\ty\Pos}$.]
    By \Lem{lem:constr-implies-subtype} we have $\ep ~|~ \CDtx |- \ty\Pos =< \ty\Neg$.

    \item[Case $\ty\Pos$ is a type variable $\be$.]
    The sub-derivation can proceed by either \ruleName{C-Skip} or \ruleName{C-VarL}.
    In both cases, we have $(\be =< \ty\Neg) \in \CDtx$.
    Since $\CDtx = \CDtx_0 \disjunion \{(\al =< \ty\Neg),\ (\ty\Pos =< \al)\}$, we have
    $(\be =< \ty\Neg) \in \CDtx_0$.
  \end{description}
\end{proof}

\begin{lemma}
  \label{lem:constr-implies-bounds-cons}
  If $\ep |> \SCtx$ outputs $\CDtx$
  then $\CDtx$ is bounds-consistent.
\end{lemma}

\begin{proof}[Proof of \Lem{lem:constr-implies-bounds-cons}]
  
  According to \Def{def:bounds-consistency}, we need to prove that for arbitrary
  $\CDtx_{0}, \al, \ty\Pos$, $\ty\Neg$, %
  such that 
  $\CDtx = \CDtx_0 \disjunion \{(\al =< \ty\Neg), (\ty\Pos =< \al)\}$ and $\ntv{\ty\Neg}$,
  if $\ntv{\ty\Pos}$ then $\ep ~|~ \CDtx |- \ty\Pos =< \ty\Neg$;
  if $\ty\Pos$ is a type variable $\be \neq \al$ then $(\be =< \ty\Neg) \in \CDtx_0$.%

  By \Lem{lem:constr-add-bounds}, there exists two
  sub-derivations
  $\CCtx_1' |> (\ty\Pos =< \al) \catsp \SCtx_1'$
  and
  $\CCtx_2' |> (\al =< \ty\Neg) \catsp \SCtx_2'$
  for some $\CCtx_1', \SCtx_1', \CCtx_2'$ and $\SCtx_2'$,
  such that $(\ty\Pos =< \al) \notin \CCtx_1'$ and
  $(\al =< \ty\Neg) \notin \CCtx_2'$.
  Since $(\ty\Pos =< \al) \neq (\al =< \ty\Neg)$, these two
  sub-derivations cannot be the same, and one
  of them is again a sub-derivation of the other.
  In both cases, the derivation proceeds by \ruleName{C-VarR} and \ruleName{C-VarL}
  and gives $\CCtx_0 \catsp (\ty\Pos =< \al) \catsp \CCtx_1 \catsp (\al =< \ty\Neg) \catsp \CCtx_2 |> \SCtx'$
  for some $\CCtx_{0,1,2}$ and $\SCtx'$.
  So by \Lem{lem:constr-implies-bounds-cons-basic-2} we have that
  if $\ntv{\ty\Pos}$ then $\ep \mid \CDtx |- \ty\Pos =< \ty\Neg$, so $\CCtx_p ~|~ \CDtx |- \ty\Pos =< \ty\Neg$ holds;
  or if $\ty\Pos$ is a type variable $\be$, we have $(\be =< \ty\Neg) \in \CDtx_0$.

\end{proof}

\begin{lemma}
  \label{lem:bounds-cons-implies-merged-var}
  For any bounds-consistent and fully merged $\CDtx$, type variables $\al$, $\be$, if $\CDtx |- \al =< \be$,
  then for all
  $\ty\Neg$,
  if $\ntv{\ty\Neg}$ and $(\al =< \ty\Neg) \in \CDtx$,
  then we have $(\be =< \ty\Neg) \in \CDtx$;
  and for all
  $\ty\Neg$,
  if $\ntv{\ty\Neg}$ and $(\be =< \ty\Neg) \in \CDtx$,
  then we have $(\al =< \ty\Neg) \in \CDtx$.
\end{lemma}

\begin{proof}[Proof of \Lem{lem:bounds-cons-implies-merged-var}]
  By induction on $\CDtx |- \al =< \be$, there are only three possible rules as follows:
  \begin{description}
    \item[Case \ruleName{S-Var}.]
    The conclusion follows immediately since $\al = \be$.
    
    \item[Case \ruleName{S-Hyp}.]
    We have $(\al =< \be) \in \CDtx$.
    
    According to \Def{def:bounds-consistency},
    since $\CDtx$ is bounds-consistent,
    we have that for all $\ty\Neg$ such that $\ntv{\ty\Neg}$, if $(\be =< \ty\Neg) \in \CDtx$
    then $(\al =< \ty\Neg) \in \CDtx$.
    
    According to \Def{def:fully-merged},
    since $\CDtx$ is fully merged,
    $(\be =< \al) \in \CDtx$, so similarly by \Def{def:bounds-consistency},
    since $\CDtx$ is bounds-consistent,
    we have that for all $\ty\Neg$ such that $\ntv{\ty\Neg}$, if $(\al =< \ty\Neg) \in \CDtx$
    then $(\be =< \ty\Neg) \in \CDtx$.

    \item[Case \ruleName{S-Trans}.]
    Direct by IH.
  \end{description}
\end{proof}

\begin{lemma}
  \label{lem:bounds-cons-implies-fun-inclusion}
  For any bounds-consistent $\CDtx$, type variables $\al,\be_1,\be_2$, %
  if $\CDtx |- \al =< \be_1 -> \be_2$
  then $\al =< \be_1 -> \be_2 \in \CDtx$.
\end{lemma}
\begin{proof}[Proof of \Lem{lem:bounds-cons-implies-fun-inclusion}]
  By induction on $\CDtx |- \al =< \be_1 -> \be_2$. The possible cases
  are as follows:
  \begin{description}
    \item[Case \ruleName{S-Hyp}.] Immediate.
    \item[Case \ruleName{S-Trans}.]
    We have $\CDtx |- \al =< \ga$ and $\CDtx |- \ga =< \be_1 -> \be_2$.
    By IH we have $(\ga =< \be_1 -> \be_2) \in \CDtx$.
    Then by \Lem{lem:bounds-cons-implies-merged-var},
    $(\al =< \be_1 -> \be_2) \in \CDtx$.
  \end{description}

\end{proof}

\begin{lemma}
  \label{lem:bounds-cons-implies-ctor-inclusion}
  For any bounds-consistent $\CDtx$, type variables
  $\al,\overline{\be_{i,j}}^{i,j},\overline{\ga_{i,j}}^{i,j},\al_{\fname{res}},
  \overline{x}^{i,j}$ and $\overline{l_i}^i$,
  \

  if $\CDtx |- \al =< \left\{\overline{
    c_i\,\langle\overline{x_{i,j} |-> \be_{i,j}}^{j}\rangle -> l_i
  }^{i}\right\}: \al_{\fname{res}}$
  then $\al =< \left\{\overline{
    c_i\,\langle\overline{x_{i,j} |-> \be_{i,j}}^{j}\rangle -> l_i
  }^{i}\right\}: \al_{\fname{res}} \in \CDtx$.
\end{lemma}

\begin{proof}[Proof of \Lem{lem:bounds-cons-implies-ctor-inclusion}]
  By induction on
  $\CDtx |- \al =< \left\{\overline{
    c_i\,\langle\overline{x_{i,j} |-> \be_{i,j}}^{j}\rangle -> l_i
  }^{i}\right\}: \al_{\fname{res}}$. The possible cases are as follows:
  \begin{description}
    \item[Case \ruleName{S-Hyp}.] Immediate.
    \item[Case \ruleName{S-Trans}.]
    For conciseness, we denote $\left\{\overline{
      c_i\,\langle\overline{x_{i,j} |-> \be_{i,j}}^{j}\rangle -> l_i
    }^{i}\right\}: \al_{\fname{res}}$ as $S$ in the proof below. 
    We have $\CDtx |- \al =< \be$ and $\CDtx |- \be =< S$. By IH we have
    $\left(\be =< S\right) \in \CDtx$.
    Then we have
    $\left(\al =< \left\{\overline{
      c_i\,\langle\overline{x_{i,j} |-> \be_{i,j}}^{j}\rangle -> l_i
    }^{i}\right\}: \al_{\fname{res}}\right) \in \CDtx$ by \Lem{lem:bounds-cons-implies-merged-var}.
  \end{description}
\end{proof}

\begin{lemma}
  \label{lem:bounds-cons-implies-cons}
  If
  $\CCtx \conR \phi$,
  $\CCtx$ is fully merged,
  $\CCtx$ is bounds-consistent
  and $\CCtx |- \SCtx$,
  then $\phi$ is consistent with $\SCtx$.
\end{lemma}

\begin{proof}[Proof of \Lem{lem:bounds-cons-implies-cons}]

  For all $\ty\Pos =< \ty\Neg \in \SCtx$, we prove by case analysis on the
  possible forms of $\ty\Pos$ and $\ty\Neg$.
  \begin{description}
    \item[Case $\ty\Pos = \al$, $\ty\Neg = \be$.]
    If $\al = \be$, then $\phi(\al) = \phi(\be)$ holds.
    Otherwise, since $\CCtx |- (\al =< \be)$, by \Lem{lem:bounds-cons-implies-merged-var}
    $\al$ and $\be$ share the same set of upper bounds.
    So by the definition of $\fname{unif}$, we have $\phi(\al) == \phi(\be)$.
    
    \item[Case $\ty\Pos = \al$, $\ty\Neg = \be_1 -> \be_2$.]
    Since $\CCtx |- \al =< (\be_1 -> \be_2)$ and $\CCtx$ is bounds-consistent,
    By \Lem{lem:bounds-cons-implies-fun-inclusion}
    we have $\al =< \be_1 -> \be_2 \in \CCtx$.
    Then by \Def{def:fully-merged}, \Def{def:bounds-consistency},
    \Lem{lem:bounds-cons-implies-merged-var} and
    the definition of $\fname{unif}$, we have $\phi(\al) == \phi(\be_1) -> \phi(\be_2)$ or
    $\phi(\al) == \phi(\ga_1) -> \phi(\ga_2)$ where $\phi(\ga_1) == \phi(\be_1)$
    and $\phi(\ga_2) == \phi(\be_2)$.

    \item[Case $\ty\Pos = \al$, $\ty\Neg = \left\{\,\overline{c_i\,\langle\,\overline{x_{i,j} |-> \be_{i,j}}^j\,\rangle -> l_i}^i\,\right\}: \de$.]
    Since $\CCtx |- \al =< \left\{\,\overline{c_i\,\langle\,\overline{x_{i,j} |-> \be_{i,j}}^j\,\rangle -> l_i}^i\,\right\}: \de$
    and $\CCtx$ is bounds-consistent and fully merged,
    By \Lem{lem:bounds-cons-implies-ctor-inclusion}, \Lem{lem:bounds-cons-implies-merged-var} and
    the definition of $\fname{unif}$,
    if $\al$ only has one upper bound, then we have
    $\phi(\al) == \left\{\,\overline{c_i\,\langle\,\overline{x_{i,j} |-> \phi(\be_{i,j})}^j\,\rangle -> l_i}^i\,\right\}: \phi(\de)$ %
    If $\al$ has more than one upper bound, then we have
    $\phi(\al) == \left\{\overline{c_i\langle\overline{\phi(\ga_{i,j})}^j\rangle}^i\right\}$ where $\overline{\phi(\ga_{i,j}) == \phi(\be_{i,j})}^{i,j}$.

    \item[Case $\ty\Pos = \al_1 -> \al_2$, $\ty\Neg = \be$.]
    By induction on $\CCtx |- \al_1 -> \al_2 =< \be$. The only applicable cases are as follows:
    \begin{description}
      \item[Case \ruleName{S-Hyp}.] By the definition of $\fname{unif}$, if $\be$ has
      no concrete upper bounds then $\phi(\be) = \top$.
      Otherwise, since $\CCtx$ is bounds-consistent and fully merged,
      all the concrete upper bounds of $\be$ in $\CCtx$
      will be checked with $\ty\Pos$ for consistency.
      Then by the definition of $\fname{unif}$,
      $\phi(\be) == S_1 -> S_2$ where $S_1 == \phi(\al_1)$ and $S_2 == \phi(\al_2)$.

      \item[Case \ruleName{S-Trans}.]
      In this case we have $\CCtx |- \al_1 -> \al_2 =< \ga$ and $\CCtx |- \ga =< \be$
      by assumption. By IH, we have that $\phi$ is consistent with $\al_1 -> \al_2 =< \ga$.
      Then by the first case we have $\phi(\ga) == \phi(\be)$. So $\phi$ is consistent with
      $\al_1 -> \al_2 =< \be$.
    \end{description}

    \item[Case $\ty\Pos = \al_1 -> \al_2$, $\ty\Neg = \be_1 -> \be_2$.]
    Since $\CCtx |- (\al_1 -> \al_2) =< (\be_1 -> \be_2)$,
    the only applicable case is \ruleName{S-Fun}, then by
    the first case, we have $\phi(\al_1) == \phi(\be_1)$ and $\phi(\al_2) == \phi(\be_2)$,
    so $\phi(\ty\Pos) == \phi(\ty\Neg)$.
    
    \item[Case $\ty\Pos = c_n\,\langle\overline{\al_{n,j}}^j\rangle$, $\ty\Neg = \be$.]
    By induction on $\CCtx |- c_n\,\langle\overline{\al_{n,j}}^j\rangle =< \be$.
    The only applicable cases are as follows:
    \begin{description}
      \item[Case \ruleName{S-Hyp}.] By the definition of $\fname{unif}$, if $\be$ has
      no concrete upper bounds then $\phi(\be) = \top$.
      Otherwise, since $\CCtx$ is bounds-consistent and fully merged, 
      all the concrete upper bounds of $\be$ in $\CCtx$
      will be checked with $\ty\Pos$ for consistency.
      Then by the definition of $\fname{unif}$,
      $\phi(\be) ==
        \left\{\overline{c_i\,\langle\overline{x_{i,j}|->S_{i,j}}^j\rangle->l_i}^i\right\}:S'$
      where $\overline{S_{n,j} == \phi(\al_{n,j})}^j$ if it only has one concrete upper bound.
      Or $\phi(\be) ==
      \left\{\overline{c_i\,\langle\,\overline{S_{i,j}}^j\,\rangle}^i\right\}$
      where $\overline{S_{n,j} == \phi(\al_{n,j})}^{j}$
      if $\be$ has multiple concrete upper bounds. %
      
      \item[Case \ruleName{S-Trans}.] 
      In this case we have $\CCtx |- c_n\,\langle\overline{\al_{n,j}}^j\rangle =< \ga$
      and $\CCtx |- \ga =< \be$ by assumption.
      By IH, we have $\phi$ is consistent with $c_n\,\langle\overline{\al_{n,j}}^j\rangle =< \ga$.
      Then by the first case we have $\phi(\ga) == \phi(\be)$.
      $\phi$ is consistent with $c_n\,\langle\overline{\al_{n,j}}^j\rangle =< \be$.
    \end{description}
    
    \item[Case $\ty\Pos = c_n\,\langle\overline{\al_{n,j}}^j\rangle$, $\ty\Neg = \left\{\,\overline{c_i\,\langle\,\overline{x_{i,j} |-> \be_{i,j}}^j\,\rangle -> l_i}^i\,\right\}: \de$.]
    \

    Since $\CCtx |- c_n\,\langle\overline{\al_{n,j}}^j\rangle =< \left\{\,\overline{c_i\,\langle\,\overline{x_{i,j} |-> \be_{i,j}}^j\,\rangle -> l_i}^i\,\right\}: \de$,
    the only applicable case is \ruleName{S-Ctor},
    then by the first case, we have $\overline{\phi(\al_{n,j}) == \phi(\be_{n,j})}^j$. %
  \end{description}

\end{proof}

\begin{proof}[Proof of \Lem{lem:constr-implies-cons}]
  By
  \Lem{lem:constr-implies-subtype},
  \Lem{lem:constr-implies-fully-merged},
  \Lem{lem:constr-implies-bounds-cons} and
  \Lem{lem:bounds-cons-implies-cons}.
\end{proof}

\begin{proof}[Proof of \Lem{lem:cons-implies-wt}]
  By induction on the fusion inference rules.%
  \begin{description}
    \item[Case \ruleName{I-Var}.] 
      By direct application of \ruleName{P-Var}.
    \item[Case \ruleName{I-Lam}.] 
      By IH we have
      $\phi(\Gamma\catsp(x|->\al)) = \phi(\Gamma)\catsp(x|->\phi(\al))
        |- t ~> t' <= \phi(\be)$
      so
      by \ruleName{P-Lam}
      we have
      $\phi(\Gamma) |- \lam x.\ t ~> \lam x.\ t' <= \phi(\al) -> \phi(\be)$.
    \item[Case \ruleName{I-App}.] 
    By IH we have
      $\phi(\Gamma) |- t_1 ~> t_1' <= \phi(\al_1)$
      and
      $\phi(\Gamma) |- t_2 ~> t_2' <= \phi(\al_2)$.
      By consistency of $\phi$ with $\al_1 =< \al_2 -> \be$,
      we have $\phi(\al_1) == S_1 -> S_2$ where $S_1 == \phi(\al_2)$ and $S_2 == \phi(\be)$.
      So
      by \ruleName{P-RecUnfold}
      and \ruleName{P-App}
      we have
      $\phi(\Gamma) |- t_1\ t_2 ~> t_1'\ t_2' <= \phi(\be)$.
    \item[Case \ruleName{I-Ctor}.] 
      By IH we have
      $\overline{\,\phi(\Gamma) |- t_j ~> t_j' <= \phi(\al_j)\,}^j$.
      By definition of consistency,
      we can proceed by \ruleName{P-RecUnfold}
      and either \ruleName{P-Ctor}, \ruleName{P-CtorSkip}, or \ruleName{P-Top}.
      For the case of $\ruleName{P-CtorSkip}$, the induction hypothesis directly applies.
      For the case of $\ruleName{P-Top}$, take
      $S = \left\{\overline{c_i\,\langle\,\overline{\phi(\al_j)}^j\,\rangle}^i\right\}$
      and proceed by \ruleName{P-CtorSkip}.
      For the case of $\ruleName{P-Ctor}$, by self-containedness of $\SCtx$,
      $\overline{x_j |-> S_j}^j |- l ~>_{\mathcal{A}} l' <= S_{\fname{res}}$ holds,
      where $S_j == \phi(\al_j)$ for all $j$.
      So all premises for \ruleName{P-Ctor} hold and we can proceed by \ruleName{P-Ctor}.

    \item[Case \ruleName{I-Case}.] 
      By IH we have
      $\phi(\Gamma) |- t ~> t' <= \phi(\al)$
      and for every $i$, we have $
        \phi(\Gamma\catsp
        \overline{
          (x_{i,j} |-> \be_{i,j})
        }^j)
        =
        \phi(\Gamma)\catsp
        \overline{
          (x_{i,j} |-> \phi(\be_{i,j}))
        }^j
        |-
        t_i ~> t_i' <= \phi(\ga_i)
      $.
      By consistency of $\phi$ with
        $\ga_i=<\de$,
        and by the fact that $\ga_i$ cannot be a constructor type
        (easily demonstrated by induction on fusion inference rules),
        we have $\phi(\ga_i) == \phi(\al)$ for all $i$,
        so we can use \ruleName{P-RecUnfold} to make all the branches align to the
        same result strategy $S$.
      By definition of consistency of $\phi$ with
        $\al =<
          \left\{\,\overline{c_i\,\langle\,\overline{x_{i,j} |-> \be_{i,j}}^j\,\rangle -> l_i}^i\,\right\}:\de
        $,
      we can then proceed by %
      either \ruleName{P-Case} or \ruleName{P-CaseSkip}.
    \item[Case \ruleName{I-LetRec}.] 
      By straightforward induction and \ruleName{P-RecFold}.
  \end{description}
\end{proof}

\begin{proof}[Proof of \Thm{thm:algorithm-correctness}]
  By \Lem{lem:constr-collect-self-contained},
  \Lem{lem:constr-implies-cons},
  and \Lem{lem:cons-implies-wt}.
\end{proof}

\newpage
\section{Detailed Benchmark Results}
\label{app:exp-res}

\setlength{\belowcaptionskip}{0pt}

\begin{table}[!ht]
  \small\centering
  \begin{tabular}{lllll}
  \hline
      \multirow{2}{*}{\textbf{name}} & \multicolumn{2}{c}{\textbf{time}} & \multicolumn{2}{c}{\textbf{size}} \\ \cline{2-5}
      & \textbf{original} & \textbf{optimized} & \textbf{original} & \textbf{optimized} \\ \hline
      Ansi\_lh & 16.63 ms & 15.78 ms & 85216 & 327480 \\
      Atom\_lh & 156.6 ms & 158.45 ms & 26920 & 32952 \\
      Awards\_lh & 31.84 ms & 30.22 ms & 37944 & 49200 \\
      Banner\_lh & 11.52 ms & 9.06 ms & 181352 & 248720 \\
      Boyer2\_lh & 35.59 ms & 35.22 ms & 1323672 & 1496600 \\
      Boyer\_lh & 52.31 ms & 51.78 ms & 83792 & 113120 \\
      Calendar\_lh & 12.36 ms & 11.23 ms & 55920 & 114744 \\
      Cichelli\_lh & 34.31 ms & 35.17 ms & 69920 & 143888 \\
      Circsim\_lh & 21.28 ms & 20.14 ms & 93472 & 156680 \\
      Clausify\_lh & 35.77 ms & 35.62 ms & 44344 & 64536 \\
      Constraints\_lh & 210.82 ms & 140.22 ms & 58544 & 116520 \\
      Cryptarithm2\_lh & 10.33 ms & 8.62 ms & 54368 & 85920 \\
      Cryptarithm\_lh & 13.79 s & 12.66 s & 13992 & 15960 \\
      Cse\_lh & 2.91 ms & 2.83 ms & 35632 & 134864 \\
      Eliza\_lh & 2.46 ms & 2.3 ms & 781792 & 1151232 \\
      Fish\_lh & 427.68 ms & 320.36 ms & 85544 & 294896 \\
      Gcd\_lh & 131.47 ms & 89.88 ms & 14760 & 19528 \\
      Integer\_lh & 170.04 ms & 170.54 ms & 29248 & 39280 \\
      Knights\_lh & 1.29 ms & 1.28 ms & 99904 & 225368 \\
      LCSS\_lh & 136.47 us & 135.04 us & 21568 & 40928 \\
      Lambda\_lh & 932.69 us & 844.37 us & 64272 & 97288 \\
      LastPiece\_lh & 391.15 ms & 368.61 ms & 157040 & 1318064 \\
      Life\_lh & 13.09 ms & 13.02 ms & 67672 & 173520 \\
      Mandel2\_lh & 5.21 ms & 3.82 ms & 19480 & 20360 \\
      Mandel\_lh & 3.98 ms & 2.86 ms & 21352 & 24400 \\
      Mate\_lh & 797.64 ms & 791.88 ms & 198504 & 318544 \\
      Minimax\_lh & 17.24 ms & 16.17 ms & 74280 & 136608 \\
      Para\_lh & 65.61 us & 65.85 us & 49056 & 71992 \\
      Power\_lh & 768.02 us & 768.54 us & 110904 & 221016 \\
      Pretty\_lh & 7.49 us & 7.42 us & 37760 & 45208 \\
      Puzzle\_lh & 60.63 ms & 60.03 ms & 85560 & 319528 \\
      Rewrite\_lh & 42.67 ms & 42.39 ms & 191560 & 392848 \\
      Rsa\_lh & 64.76 ms & 65.19 ms & 16024 & 15840 \\
      Scc\_lh & 1.65 us & 1.6 us & 18320 & 22064 \\
      Secretary\_lh & 5.24 s & 5.22 s & 27728 & 36056 \\
      Sorting\_lh & 41.81 us & 39.88 us & 187680 & 209696 \\
      Sphere\_lh & 433.16 us & 313.32 us & 59744 & 119576 \\
      Treejoin\_lh & 895.49 us & 880.25 us & 16336 & 83344 \\ \hline
  \end{tabular}
  \caption{Runtime and object file size (in byte) between original and optimized programs.}
\end{table}

\begin{table}[!ht]
  \small\centering
  \begin{tabular}{lllllll}
  \hline
      \multirow{2}{*}{\textbf{name}} & \multicolumn{2}{c}{\textbf{minor}} & \multicolumn{2}{c}{\textbf{major}} & \multicolumn{2}{c}{\textbf{promotion}} \\ \cline{2-7}
      & \textbf{original} & \textbf{optimized} & \textbf{original} & \textbf{optimized} & \textbf{original} & \textbf{optimized} \\ \hline
      Ansi\_lh & 14.16Mw & 13.70Mw & 297.95kw & 297.78kw & 297.95kw & 297.78kw \\
      Atom\_lh & 41.02Mw & 39.48Mw & 22.43Mw & 24.84Mw & 22.43Mw & 24.84Mw \\
      Awards\_lh & 30.13Mw & 28.64Mw & 488.35kw & 490.06kw & 488.35kw & 490.06kw \\
      Banner\_lh & 2.55Mw & 2.37Mw & 1.70Mw & 1.44Mw & 1.70Mw & 1.44Mw \\
      Boyer2\_lh & 7.69Mw & 7.68Mw & 1.06Mw & 1.06Mw & 1.06Mw & 1.06Mw \\
      Boyer\_lh & 24.39Mw & 23.96Mw & 5.67Mw & 5.65Mw & 5.67Mw & 5.65Mw \\
      Calendar\_lh & 8.82Mw & 8.39Mw & 366.59kw & 288.87kw & 366.59kw & 288.87kw \\
      Cichelli\_lh & 12.78Mw & 13.50Mw & 155.52kw & 154.20kw & 155.52kw & 154.20kw \\
      Circsim\_lh & 13.61Mw & 12.31Mw & 329.37kw & 300.39kw & 329.37kw & 300.39kw \\
      Clausify\_lh & 11.68Mw & 11.68Mw & 3.55Mw & 3.55Mw & 3.55Mw & 3.55Mw \\
      Constraints\_lh & 56.92Mw & 51.21Mw & 22.51Mw & 15.77Mw & 22.51Mw & 15.77Mw \\
      Cryptarithm2\_lh & 16.21Mw & 13.85Mw & 183.44kw & 126.94kw & 183.44kw & 126.94kw \\
      Cryptarithm\_lh & 326.30Mw & 325.93Mw & 200.74Mw & 200.77Mw & 200.74Mw & 200.77Mw \\
      Cse\_lh & 2.01Mw & 1.81Mw & 164.61kw & 162.94kw & 164.61kw & 162.94kw \\
      Eliza\_lh & 1.87Mw & 1.83Mw & 16.75kw & 13.09kw & 16.75kw & 13.09kw \\
      Fish\_lh & 194.08Mw & 177.23Mw & 35.93Mw & 24.50Mw & 35.93Mw & 24.50Mw \\
      Gcd\_lh & 33.14Mw & 25.81Mw & 6.05Mw & 2.66Mw & 6.05Mw & 2.66Mw \\
      Integer\_lh & 21.68Mw & 21.68Mw & 18.88Mw & 18.88Mw & 18.88Mw & 18.88Mw \\
      Knights\_lh & 198.53kw & 200.27kw & 3.40kw & 3.43kw & 3.40kw & 3.43kw \\
      LCSS\_lh & 113.13kw & 112.11kw & 563.68w & 517.40w & 563.68w & 517.40w \\
      Lambda\_lh & 1.31Mw & 1.08Mw & 10.56kw & 9.17kw & 10.56kw & 9.17kw \\
      LastPiece\_lh & 132.06Mw & 137.81Mw & 14.53Mw & 14.60Mw & 14.53Mw & 14.60Mw \\
      Life\_lh & 12.74Mw & 12.96Mw & 1.03Mw & 1.03Mw & 1.03Mw & 1.03Mw \\
      Mandel2\_lh & 3.37Mw & 3.33Mw & 35\_044.75w & 978.62w & 35\_044.75w & 978.62w \\
      Mandel\_lh & 3.27Mw & 2.70Mw & 168.37kw & 19.49kw & 168.37kw & 19.49kw \\
      Mate\_lh & 287.92Mw & 287.58Mw & 5.85Mw & 5.86Mw & 5.85Mw & 5.86Mw \\
      Minimax\_lh & 13.84Mw & 12.80Mw & 915.08kw & 893.26kw & 915.08kw & 893.26kw \\
      Para\_lh & 57.77kw & 57.46kw & 965.84w & 899.30w & 965.84w & 899.30w \\
      Power\_lh & 388.83kw & 388.74kw & 50.72kw & 50.71kw & 50.72kw & 50.71kw \\
      Pretty\_lh & 10.09kw & 9.46kw & 61.48w & 56.26w & 61.48w & 56.26w \\
      Puzzle\_lh & 18.16Mw & 21.79Mw & 1.76Mw & 1.81Mw & 1.76Mw & 1.81Mw \\
      Rewrite\_lh & 35.76Mw & 34.94Mw & 468.32kw & 461.60kw & 468.32kw & 461.60kw \\
      Rsa\_lh & 8.87Mw & 8.78Mw & 34.07Mw & 33.88Mw & 413.62kw & 222.88kw \\
      Scc\_lh & 844.02w & 844.01w & 0.37w & 0.37w & 0.37w & 0.37w \\
      Secretary\_lh & 2.05Gw & 2.05Gw & 115.45Mw & 115.00Mw & 115.45Mw & 115.00Mw \\
      Sorting\_lh & 42.92kw & 42.66kw & 611.40w & 554.28w & 611.40w & 554.28w \\
      Sphere\_lh & 568.63kw & 469.16kw & 13.96kw & 2.17kw & 13.96kw & 2.17kw \\
      Treejoin\_lh & 435.34kw & 435.34kw & 45.80kw & 45.80kw & 45.80kw & 45.80kw \\ \hline
  \end{tabular}
  \caption{OCaml memory allocation report between original and optimized programs.
  Minor stands for the number of words allocated on the minor heap; major stands for the
  number of words allocated on the major heap; promotion stands for the number of words
  that are promoted from the minor heap to the major heap.}
\end{table}

\fi

\end{document}